\documentclass[11pt]{article}
\usepackage[margin = 1in]{geometry}

\usepackage{amsmath,amssymb,amsthm,hyperref,color}
\usepackage{float}
\usepackage[font=small]{caption}
\usepackage[font=footnotesize]{subcaption}
\usepackage[inline]{enumitem}
\usepackage{filecontents}
\usepackage{tikz}
\hypersetup{colorlinks=true,citecolor=blue, linkcolor=blue,
urlcolor=blue}

\newtheorem{theorem}{Theorem}
\newtheorem{lemma}[theorem]{Lemma}

\newtheorem{proposition}[theorem]{Proposition}

\newtheorem{corollary}[theorem]{Corollary}
\newtheorem{remark}[theorem]{Remark}
\newtheorem*{remarkno}{Remark}

\newtheorem{definition}[theorem]{Definition}

\newtheorem*{lemrealexp}{Proposition \ref{lem:realexp}}
\newtheorem*{lemcomplexexp}{Proposition \ref{lem:complexexp}}
\newtheorem*{lemopoptwo}{Lemma \ref{opop2}}
\newtheorem*{lemmainonetwoone}{Proposition \ref{lem:main121}}
\newtheorem*{lemfixpoints}{Lemma \ref{lem:fixpoints}}
\newtheorem*{lemhoptwo}{Lemma \ref{hop2}}
\newtheorem*{lemprecision}{Lemma \ref{lem:precision}}

\renewcommand\arg{\text{arg}}
\newcommand\Arg{\text{Arg}} 
\def\calI{\mathcal{I}}

\def\NP{\mathsf{NP}}
\def\numP{\#\mathsf{P}}
\def\Zin{Z^{\mathsf{in}}}
\def\Zout{Z^{\mathsf{out}}}
\def\Reals{\mathbb{R}}

\def\QQ{\mathbb{Q}}
\def\Complex{\mathbb{C}}
\def\CQ{\mathbb{C}_{\mathbb{Q}}}

\newcommand{\size}[1]{\mathrm{size}(#1)}
\def\Cnz{\mathbb{C}_{\neq0}}
\def\Riem{\widehat{\Complex}}
\def\LambdaD{\Lambda_\Delta}
\def\tLambdaD{\vartheta\Lambda_\Delta}
\def\Lbip{\mathcal{L}_{\mathrm{bip}}}
\newcommand{\fn}[2]{#1^{#2}}

\def\lambdab{\ensuremath{\boldsymbol{\lambda}}}

\newcommand{\im}{\mathrm{i}}
\newcommand{\emm}{\mathrm{e}}

\def\FactorHardCore#1{\#\ensuremath{\mathsf{BipHardCoreNorm}(\lambda,\Delta,#1)}}
\def\FactorMVHardCore#1{\#\ensuremath{\mathsf{MVHardCoreNorm}(\lambda,\Delta,#1)}}
\def\FactorMVBipHardCore#1{\#\ensuremath{\mathsf{MVBipHardCoreNorm}(\lambda,\Delta,#1)}}
\def\ArgHardCore#1{\#\ensuremath{\mathsf{BipHardCoreArg}(\lambda,\Delta,#1)}}
\def\ArgMVHardCore#1{\#\ensuremath{\mathsf{MVHardCoreArg}(\lambda,\Delta,#1)}}
\def\ArgMVBipHardCore#1{\#\ensuremath{\mathsf{MVBipHardCoreArg}(\lambda,\Delta,#1)}}
\def\ComplexHardCore{\#\ensuremath{\mathsf{BipComplexHardCore}(\lambda,\Delta)}}

\newcommand{\eps}{\epsilon}

\newcommand{\abs}[1]{\left|#1\right|}

\makeatletter
\def\prob#1#2#3{\goodbreak\begin{list}{}{\labelwidth\z@ \itemindent-\leftmargin
                        \itemsep\z@  \topsep6\p@\@plus6\p@
                        \let\makelabel\descriptionlabel}
                \item[\it Name]#1
               \item[\it Instance]                #2
                \item[\it Output]#3
                \end{list}}
\makeatother

\def\Bezakova{Bez\'{a}kov\'{a}}

\def\textin{\text{in}}

\title{Inapproximability of the independent set polynomial in the complex plane\footnote{A preliminary short version of the manuscript appeared in STOC 2018.}}
\author{
Ivona \Bezakova\thanks{
  Department of Computer Science, Rochester Institute of Technology, Rochester, NY, USA. Research
supported by NSF grant CCF-1319987.} \and
Andreas Galanis\thanks{
  Department of Computer Science, University of Oxford, Wolfson Building, Parks Road, Oxford, OX1~3QD, UK.
  The research leading to these results has received funding from the European Research Council under
  the European Union's Seventh Framework Programme (FP7/2007-2013) ERC grant agreement no.\ 334828. The paper
  reflects only the authors' views and not the views of the ERC or the European Commission.
  The European Union is not liable for any use that may be made of the information contained therein.}
  \and
  Leslie Ann Goldberg$^\ddag$
\and
 Daniel \v{S}tefankovi\v{c}\thanks{
Department of Computer Science, University of Rochester,
Rochester, NY 14627.  Research
supported by NSF grant CCF-1563757.}
 }

\date{July 5, 2020}

\begin{document}

\maketitle
\thispagestyle{empty}
\begin{abstract} 

We study the complexity of approximating the value of the independent set  polynomial $Z_G(\lambda)$
of a graph~$G$ with maximum degree $\Delta$ when the activity~$\lambda$ is a complex number.

When $\lambda$ is real,  
the complexity picture is well-understood, and is captured by two 
real-valued thresholds~$\lambda^*$ and $\lambda_c$, which depend on~$\Delta$ and satisfy
$0<\lambda^*<\lambda_c$. It is known that if $\lambda$ is a
real number in the interval $(-\lambda^*,\lambda_c)$
then there is an FPTAS for approximating $Z_G(\lambda)$ on 
graphs $G$ with maximum degree at most~$\Delta$.
On the other hand, if $\lambda$ is a real number outside of the (closed)
interval, then approximation is NP-hard. The key to establishing this picture was the interpretation of the thresholds $\lambda^*$ and $\lambda_c$ on the $\Delta$-regular tree.
The ``occupation ratio'' of a $\Delta$-regular tree~$T$ is
the contribution to $Z_T(\lambda)$ from independent sets containing the root of the tree,
divided by $Z_T(\lambda)$ itself.
This occupation ratio converges to a limit,
as the height of the tree grows,  if and only if $\lambda\in [-\lambda^*,\lambda_c]$.

Unsurprisingly, the case where $\lambda$ is complex is more challenging.  
It is known that there is an FPTAS when
$\lambda$ is a complex number  with norm at most $\lambda^*$
and  also when $\lambda$ is in a small strip surrounding the real interval $[0,\lambda_c)$.
However, neither of these results is believed to fully capture the truth about
when approximation is possible.  Peters and Regts identified the complex values of $\lambda$ for which 
the occupation  ratio of the $\Delta$-regular tree converges. These values carve a cardioid-shaped region $\LambdaD$ 
in the complex plane, whose boundary
  includes the critical points $-\lambda^*$ and $\lambda_c$. Motivated by the picture in the real case, they asked whether $\LambdaD$ marks the true approximability threshold for general complex values $\lambda$. 

Our main result shows that for every $\lambda$ outside of $\LambdaD$, 
the problem of approximating $Z_G(\lambda)$ on graphs $G$ with maximum degree at most~$\Delta$
is  indeed NP-hard.  In fact, when 
$\lambda$ is outside of $\LambdaD$ and  is not a positive real number, we give the stronger result
that approximating $Z_G(\lambda)$ is actually \#P-hard.
Further, on the negative real axis, when $\lambda < - \lambda^*$, we show that it is \#P-hard to even decide 
whether $Z_G(\lambda)>0$, resolving in the affirmative a conjecture of 
Harvey, Srivastava and Vondr\'ak. 

Our proof techniques are
based around tools from complex analysis --- specifically 
the study of iterative multivariate rational maps.
\end{abstract}

\newpage

\clearpage
 \setcounter{page}{1}
\section{Introduction}

The independent set polynomial is one of the most well-studied graph polynomials, arising in combinatorics and in computer science. 
It is also known in statistical physics as the partition function of the hard-core  model.
This paper studies the computational complexity of evaluating the polynomial approximately 
when a parameter, called the \emph{activity}, is complex.
For properties of this polynomial  in the complex plane, including connections to the 
Lov\'asz Local Lemma, see the work of Scott and Sokal~\cite{SS}.
For $\lambda \in \Complex$ and a graph $G$ 
the polynomial is defined as $Z_G(\lambda):=\sum_{I}\lambda^{|I|}$, where the sum ranges over all independent sets of $G$.
We will be interested in the problem of approximating $Z_G(\lambda)$ when the
maximum degree of $G$ is bounded.

When $\lambda$ is real, the complexity picture is well-understood.
For $\Delta\geq 3$, let $\mathcal{G}_\Delta$ be the set of graphs
with maximum degree at most~$\Delta$.
The complexity of
approximating $Z_G(\lambda)$ for $G\in \mathcal{G}_\Delta$ is captured
by two real-valued thresholds
$\lambda^*$ and $\lambda_c$ which depend on $\Delta$ and satisfy $0<\lambda^* < \lambda_c$.
To be precise,    $\lambda^*=\frac{(\Delta-1)^{\Delta-1}}{\Delta^{\Delta}}$ and $\lambda_c=\frac{(\Delta-1)^{\Delta-1}}{(\Delta-2)^{\Delta}}$.
The known  results are as follows.  
\begin{enumerate}
\item If $\lambda$ is in the interval  $-\lambda^*<\lambda<\lambda_c$, there  is an FPTAS for approximating $Z_G(\lambda)$ on graphs 
$G\in \mathcal{G}_\Delta$. For $0\leq\lambda<\lambda_c$, this follows from the work of Weitz \cite{Weitz}, while for $-\lambda^*<\lambda<0$
it follows  from the works of Harvey, Srivastava, and Vondr\'{a}k \cite{Piyush} and Patel and Regts \cite{PR}.
\item  If  $\lambda<-\lambda^*$ or $\lambda>\lambda_c$, it is $\NP$-hard to approximate $|Z_G(\lambda)|$ on graphs $G\in \mathcal{G}_\Delta$, even within an exponential factor. For  $\lambda>\lambda_c$, this follows from the work of Sly and Sun \cite{SlySun}, while for $\lambda<-\lambda^*$ it follows  from the work of Galanis, Goldberg, and \v{S}tefankovi\v{c} \cite{GGS}.
\end{enumerate}
The  key to establishing   this complexity characterisation was the following interpretation of the thresholds $\lambda^*$ and $\lambda_c$. 
Given a $\Delta$-regular tree $T$ of height $h$ with root~$\rho$, let $p_h$ denote the ``occupation ratio'' of the tree, 
which is given by $p_h=\frac{\sum_{I; \rho\in I}\lambda^{|I|}}{Z_{T}(\lambda)}$,
where   the sum ranges over  the independent sets of~$T$ that include the root~$\rho$. It turns out that the occupation ratio $p_h$ converges to a limit as $h\rightarrow\infty$ if and only if the activity $\lambda$ lies within the interval $[-\lambda^*,\lambda_c]$, so it turns out that the complexity of approximating
$Z_G(\lambda)$ for $G\in \mathcal{G}_\Delta$ depends on whether this  quantity converges.

Understanding the complexity picture in the case where $\lambda\in\Complex$ is more challenging.  
If $\lambda$ is a complex number with norm at most $\lambda^*$ 
then there is an FPTAS for
approximating $Z_G(\lambda)$  on graphs $G\in \mathcal{G}_\Delta$.
This is due to Harvey, Srivastava and Vondr{\'a}k
and to Patel and Regts~\cite{Piyush,PR}. 
More recently, Peters and Regts~\cite{Peters} showed the existence of an FPTAS when $\lambda$ is in a small strip surrounding the real interval $[0,\lambda_c)$.
However, neither of these results is believed to fully capture the truth about
when approximation is possible.  
Motivated by the real case, Peters and Regts~\cite{Peters} identified the values of $\lambda$ for which 
the occupation  ratio of the $\Delta$-regular tree converges (for $\Delta \geq 3$). These values carve a cardioid-shaped region $\LambdaD$ 
in the complex plane, whose boundary
  includes the critical points $-\lambda^*$ and $\lambda_c$.
  The definition of $\LambdaD$ is as follows  
       (see Figure~\ref{fig:cardioid})\footnote{Technically, the word ``cardioid'' refers to a curve which can be obtained by a point on the perimeter of a circle which is rolling around a fixed circle of the same radius. The region \eqref{eq:region} does not formally correspond to a ``cardioid'' in this sense, but its shape closely resembles a heart for all values of $\Delta\geq 3$, which justifies the (slight) abuse of terminology.}:
\begin{equation}\label{eq:region}
\LambdaD=\Big\{\lambda \in \mathbb{C}\,\Big|\,\exists z\in \Complex:\ |z|\leq 1/(\Delta-1),\ \lambda=\frac{z}{(1-z)^{\Delta}}\Big\}.
\end{equation}%
\begin{figure}[h]
\begin{center}
\scalebox{0.2}[0.2]{\includegraphics{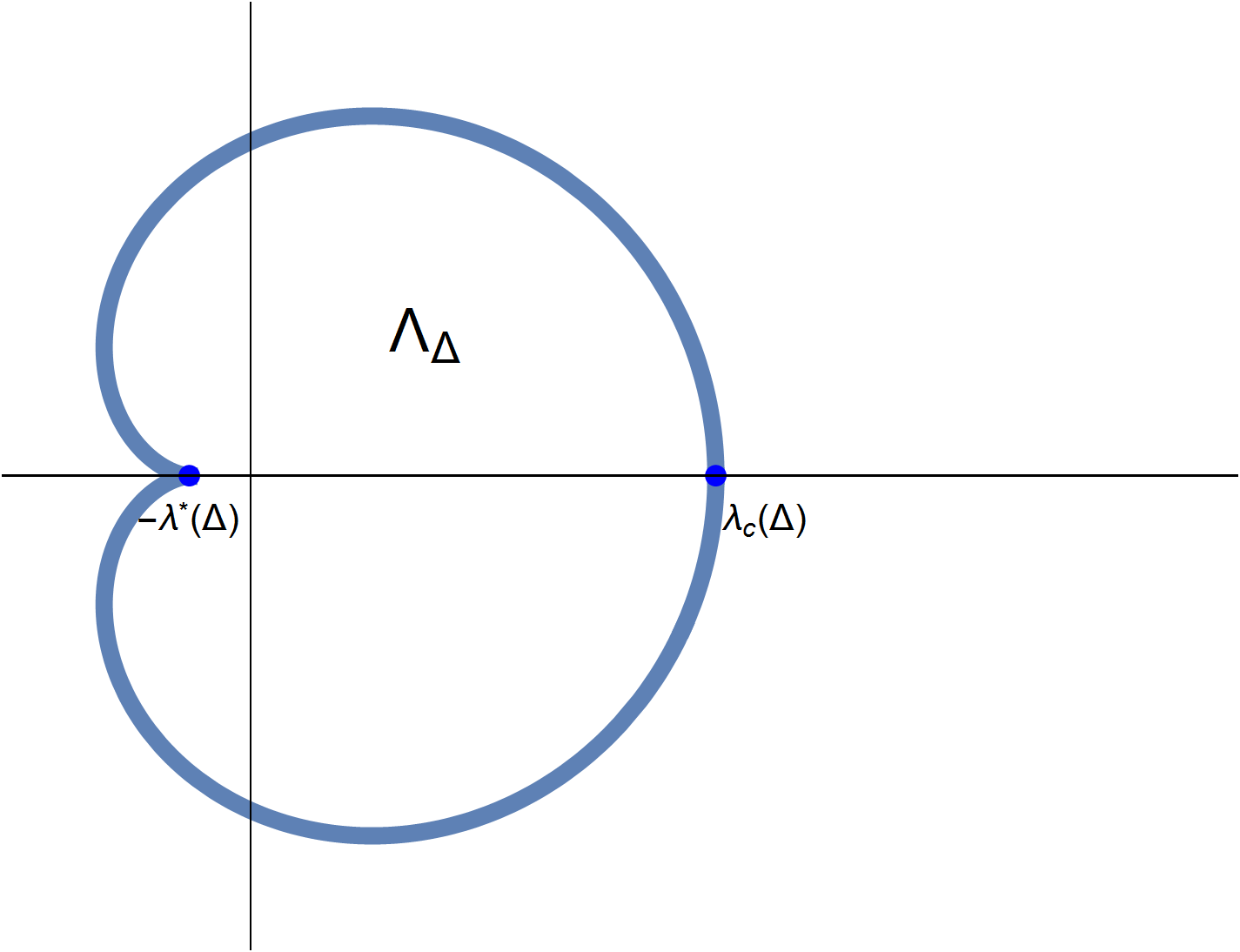}}
\end{center}
\caption{\label{fig:cardioid}The cardioid-shaped region $\LambdaD$ in the complex plane. We show that for all $\lambda\in \Complex\backslash (\LambdaD\cup \Reals_{\geq 0})$, approximating $Z_G(\lambda)$ is $\numP$-hard. Previously, it was known that the problem is $\NP$-hard on the real line in the intervals $\lambda<-\lambda^*$ and $\lambda>\lambda_c$. Note, we have  that the thresholds $-\lambda^*,\lambda_c$ belong to $\LambdaD$, by taking $z=\pm 1/(\Delta-1)$ in \eqref{fig:cardioid}.}
\end{figure}%
Peters and Regts showed that, for every $\lambda$ in the (strict) interior of $\LambdaD$, the occupation ratio of the $\Delta$-regular tree converges, and asked whether the region $\LambdaD$ marks the true approximability threshold for general complex values $\lambda$. 

Our main result shows that for every $\lambda$ outside of the region $\LambdaD$, 
the problem of approximating $Z_G(\lambda)$ on graphs $G \in \mathcal{G}_\Delta$
is  indeed NP-hard, thus answering \cite[Question 1]{Peters}.  In fact, when 
$\lambda$ is outside of $\LambdaD$ and  is not a positive real number, we establish the stronger result
that approximating $Z_G(\lambda)$ is actually \#P-hard. We do this by showing that an approximation algorithm for  $Z_G(\lambda)$ can be converted 
into  a polynomial-time algorithm for exactly counting independent sets.  
Further, on the negative real axis, when $\lambda < - \lambda^*$, we show that it is \#P-hard to even decide 
whether $Z_G(\lambda)>0$, resolving in the affirmative a conjecture of Harvey, Srivastava, and Vondr\'{a}k \cite[Conjecture 5.1]{Piyush}.

We need the following notation to formally state our results. Given
  a complex number   $x\in \Complex$, we use $|x|$ to denote its norm and $\Arg(x)$ to denote the principal value of its argument in the range $[0,2\pi)$. We also define $\arg(x) = \{ \Arg(x)+ 2 \pi j \mid j \in \mathbb{Z}\}$.  For $y,z\in \Complex$, we use
$d(y,z)$ to denote the Ziv distance between them~\cite{Ziv}, namely $d( y,z) = \frac{| y-z|}{\max(| y|,|z|)}$. We denote by $\CQ$ the set of complex numbers whose real and imaginary parts are rational numbers (see Definition~\ref{def:CQ}).

We consider the problems of multiplicatively approximating the norm of $Z_G(\lambda)$, 
additively approximating the argument of $Z_G(\lambda)$, and
approximating $Z_G(\lambda)$ by producing a complex number $\widehat{Z}$
such that the Ziv distance $d\big(\widehat{Z},Z_G(\lambda)\big)$ is small. 
  We start with  the following problem, which captures
the approximation of the norm of $Z_G(\lambda)$.
\prob{
$\FactorHardCore{K}$.} 
{ A bipartite graph $G$ with maximum degree at
most $\Delta$.} 
 { If $|Z_G(\lambda)|=0$  then the algorithm may output any rational number. Otherwise,
 it must output a  rational number $\widehat{N}$ such that
$\widehat{N}/K \leq
|Z_{G}(\lambda)|\leq K \widehat{N}$.}
Our first theorem shows that it is $\numP$-hard to approximate $|Z_G(\lambda)|$ on bipartite graphs of maximum degree $\Delta$ within a constant factor.

\begin{theorem}\label{thm:norm}
Let $\Delta\geq 3$ and $\lambda\in \CQ$ be such that $\lambda\not\in (\LambdaD \cup \Reals_{\geq 0})$. Then, 
 $\FactorHardCore{1.01}$  is $\numP$-hard.\end{theorem}

\begin{remarkno}The value ``$1.01$'' in the statement of Theorem~\ref{thm:norm} is
not important. In fact, for any fixed $\epsilon>0$ we can
use the theorem, together with a standard powering argument, to
show that  it is $\numP$-hard to approximate $|Z_G(\lambda)|$ 
within a factor of $2^{n^{1-\epsilon}}$. 
\end{remarkno}

The following problem captures the approximation of the argument of $Z_G(\lambda)$.  
  \prob{$\ArgHardCore{\rho}$.} 
 { A bipartite graph $G$ with maximum degree at
most $\Delta$.}
{If $Z_G(\lambda)=0$ then the algorithm may output any rational number. Otherwise, it must output  
a rational number $\widehat{A}$ such that, for some
$a \in  \arg(Z_{G}(\lambda))$,
$ |\widehat{A} - a | \leq \rho$.} 

Our second theorem shows that it is $\numP$-hard to approximate $\arg (Z_{G}(\lambda))$ on bipartite graphs of maximum degree $\Delta$ within an additive constant $\pi/3$. 
\begin{theorem}\label{thm:arg}
Let $\Delta\geq 3$ and $\lambda\in \CQ$ be such that $\lambda\not\in (\LambdaD \cup \Reals_{\geq 0})$. Then, $\ArgHardCore{\pi/3}$
 is $\numP$-hard.
\end{theorem}
 Theorem~\ref{thm:arg} also has the following immediate corollary for the case in which~$\lambda$ is a negative real number, resolving in the affirmative \cite[Conjecture 5.1]{Piyush}.
\begin{corollary}\label{cor:sign}
Let $\Delta\geq 3$ and 
$\lambda \in \QQ$  be such that
$\lambda<-\lambda^*$. Then, given as input a bipartite graph $G$ with maximum degree $\Delta$, it is $\numP$-hard to decide whether $Z_G(\lambda)>0$.
\end{corollary}

Theorems~\ref{thm:norm} and~\ref{thm:arg}
show as a corollary that it is \#P-hard to approximate $Z_G(\lambda)$ within small Ziv distance.
 
  \prob{$\ComplexHardCore$.} 
 { A  bipartite graph $G$ with maximum degree at
most $\Delta$. A positive integer $R$, in unary.}
{If $Z_G(\lambda)=0$ then the algorithm may output any complex number. Otherwise, it must output  
a  complex number $z$ such that
$d(z, Z_G(\lambda)) \leq   1/R$.}

\begin{corollary}\label{cor:Ziv}
Let $\Delta\geq 3$ and $\lambda\in \CQ$ be such that $\lambda\not\in (\LambdaD \cup \Reals_{\geq 0})$. Then, the problem 
$\ComplexHardCore$
 is $\numP$-hard.
\end{corollary}
 
Corollary~\ref{cor:Ziv} follows immediately from Theorem~\ref{thm:norm} using the fact (see \cite[Lemma 2.1]{ComplexIsing})
that $d(z',z) \leq \epsilon$ implies $|z'|/|z| \leq 1/(1-\epsilon)$.
This fact implies 
(see \cite[Lemma 2.2]{ComplexIsing})
that there is 
 a polynomial Turing reduction from  
$\FactorHardCore{K}$  to $\ComplexHardCore$.
Similarly,
Corollary~\ref{cor:Ziv} can also be proved using Theorem~\ref{thm:arg}.
To see this, note  that
for $\epsilon \leq 1/3$, $d(z',z) \leq \epsilon$ implies 
that there are 
$a\in \arg(z)$ and $a'\in \arg(z')$ such that
$| a-a'|\leq \sqrt{36\epsilon/11}$. This fact is proved in \cite[Lemma 2.1]{ComplexIsing}
 and it implies
 (see \cite[Lemma 2.2]{ComplexIsing})  that  there is a polynomial Turing reduction from   $\ArgHardCore{\rho}$ to $\ComplexHardCore$.

Note that our $\numP$-hardness results for $\lambda\in \CQ\backslash (\LambdaD\cup \Reals_{\geq 0})$  highlight a difference in complexity 
between this case and the case where $\lambda$ is a   rational satisfying $\lambda> \lambda_c$.
If $\lambda$ is a positive  rational then $Z_G(\lambda)$ can be efficiently approximated 
in polynomial time using an NP oracle, via the bisection technique of Valiant and Vazirani~\cite{VV}.
Thus, in that case approximation is NP-easy, and is unlikely to be $\numP$-hard. The techniques for proving hardness also differ in the two cases.

\subsection{Proof approach}
To prove our inapproximability results, we construct graph gadgets which, when  appended appropriately to a vertex, have the effect of altering the activity $\lambda$ to any complex activity $\lambda'$ that we wish, perhaps with some small error $\epsilon$. In fact, it is essential for our $\numP$-hardness results to be able to make the error $\epsilon$ exponentially small with respect to the number of the vertices in the graph  (see the upcoming Proposition~\ref{lem:complexexp} for details).

Interestingly,  our constructions are based on using tools from complex analysis for analysing the iteration of rational maps.   We start with the observation that $(\Delta-1)$-ary trees of height $h$  can be used to ``implement''  activities $\lambda'$ which correspond to the iterates of the  complex rational map $f:x\mapsto \frac{1}{1+\lambda x^{\Delta-1}}$. Crucially, we show that when $\lambda\notin \LambdaD$, all of the fixpoints of $f$ are repelling, i.e., applying the map $f$ at any point close to a fixpoint $\omega$ will push us away from the fixpoint. In the iteration of univariate complex rational maps, repelling fixpoints belong to the so-called Julia set of the map; a consequence of this is that iterating $f$ in a neighbourhood $U$ of a repelling fixpoint gives rise to a chaotic behaviour: after sufficiently many iterations, one ends up anywhere in the complex plane. 

This sounds promising, but how can we get close to a \emph{repelling} fixpoint of $f$ in the first place? In fact, we need to be able to create arbitrary points in a neighbourhood $U$ of a repelling fixpoint and iterating $f$ will not get us anywhere close (since the fixpoint is repelling).  The key is  to use a Fibonacci-type construction which requires analysing a more intricate multivariate version of the map $f$. Surprisingly, we can show that the iterates of the multivariate version converge to the fixpoint $\omega$ of the univariate $f$ with the smallest norm.  Using convergence properties of the multivariate map around $\omega$ (and some extra work), we obtain a family of (univariate) contracting\footnote{Let $\Phi: \Complex\rightarrow \Complex$ be a complex map. We say that $\Phi$ is contracting on a set $S\subseteq \Complex$ if there exists a  real number $M<1$ such that for all $x,y\in S$ it holds that $|\Phi(x)-\Phi(y)|\leq M|x-y|$.}  maps $\Phi_1,\hdots,\Phi_t$ and  a small neighbourhood $U$ around $\omega$ such that $U\subseteq \cup^t_{i=1}\Phi_i(U)$. The final step is to show that ``contracting maps that cover yield exponential precision''. To do this we first show  that, starting from any point in $U$, we can apply (some sequence of) $\Phi_1,\hdots,\Phi_t$ at most $poly(n)$ times to  implement any point in $U$ with precision $\exp(-\Omega(n))$. 
We then show that by iteratively applying  the univariate map $f$ and carefully tracking the distortion introduced, we 
can eventually implement any point in the complex plane with exponentially small error.

Section~\ref{sec:outline} gives a more detailed description of the proof approach, stating the key lemmas, and
explaining how they relate to each other. Section~\ref{sec:prelims} gives some preliminary definitions,
which are not necessary for understanding the proof outline in Section~\ref{sec:outline}, but are necessary once we start with the proofs.
 Section~\ref{sec:prelims} also imports some key facts from related works.
 Section~\ref{sec:cardioid} proves some key properties of the cardioid defined by~\eqref{eq:region}. 
 The main technical part of the paper is contained in Sections~\ref{sec:main1} and~\ref{sec:reductions}.
 Section~\ref{sec:main1} shows how to implement activities with exponential precision.
 Section~\ref{sec:reductions} shows how to obtain our inapproximability results, using these.
 Some of the technical proofs for Section~\ref{sec:main1} are deferred to Section~\ref{sec:mainpart}.

\subsection{New Developments}

After this paper was written, Bencs and Csikv\'ari~\cite{BC} have discovered a new zero-free region
inside the region $\LambdaD$. Using the algorithm of Patel and Regts~\cite{PR},
this gives  
an FPTAS for
approximating $Z_G(\lambda)$  on graphs $G\in \mathcal{G}_\Delta$, within this zero-free region. Further to this, Rivera-Letelier~\cite{JRL} and Buys~\cite{Buys} showed the existence of zeros inside the region $\LambdaD$, close to the boundary.

\section{Proof Outline}\label{sec:outline}
In this section, we give a more detailed outline of the proof of our results. We focus mainly on the case where $\lambda\in \CQ\backslash(\LambdaD\cup \Reals)$.
In Section~\ref{sec:realmod} we describe suitable modifications 
that will give us the ingredients needed for negative real values $\lambda\in \QQ\backslash\LambdaD$.

Let $\lambda\in \Complex$ and $G=(V,E)$ be an arbitrary graph. We denote by $\mathcal{I}_G$ the set of independent sets of $G$
(including the empty independent set). 
For a vertex $v\in V$, we will denote
\[\Zin_{G,v}(\lambda):=\sum_{I\in \mathcal{I}_G;\, v\in I}\lambda^{|I|},\quad \Zout_{G,v}(\lambda):=\sum_{I\in \mathcal{I}_G;\, v\notin I}\lambda^{|I|}.\]
Thus, $\Zin_{G,v}(\lambda)$ is the contribution to the partition function $Z_G(\lambda)$ from those independent sets $I\in \mathcal{I}_G$ such that $v\in I$; similarly, $\Zout_{G,v}(\lambda)$ is the contribution to $Z_G(\lambda)$ from those $I\in \mathcal{I}_G$ such that $v\notin I$. 
\begin{definition}\label{def:Gimplement}
Fix a complex number~$\lambda$ that is not~$0$. Given~$\lambda$,
the graph $G$  is said to \emph{implement} the activity $\lambda'\in \Complex$ with \emph{accuracy} $\epsilon>0$ if there is a vertex $v$ in $G$ such that $\Zout_{G,v}(\lambda)\neq 0$  and
\begin{enumerate}
\item \label{it:voneGde} $v$ has degree one in $G$, and
\item  \label{it:ps123}   $\displaystyle \Big|\frac{\Zin_{G,v}(\lambda)}{\Zout_{G,v}(\lambda)}-\lambda'|\leq \epsilon $.
\end{enumerate}
We call $v$ the terminal of $G$. If Item~\ref{it:ps123} holds with $\epsilon=0$, then $G$ is said to \emph{implement} the activity $\lambda'$.
\end{definition}

The key to obtaining our $\numP$-hardness results is to show that, given \emph{any} target activity $\lambda'\in \Complex$, we can construct in  polynomial time a bipartite graph $G$ that implements $\lambda'$ with exponentially small accuracy,
as a function of the size of~$\lambda'$. More precisely, we
use $\size{\lambda',\eps}$ to denote the number of bits needed to represent 
the complex number $\lambda' \in \CQ$ and the rational~$\epsilon$ (see Definition~\ref{def:CQ}).
The implementation that we need is captured by the following proposition. 
\newcommand{\statelemcomplexexp}{Let $\Delta\geq 3$ and $\lambda \in \CQ$ be such that  $\lambda\notin \LambdaD\cup \Reals$.

There is an algorithm which, on input $\lambda'\in \CQ$ and rational $\eps>0$, outputs in $poly(\size{\lambda',\eps})$ time a bipartite graph $G$ of maximum degree at most  $\Delta$ with terminal $v$ that implements $\lambda'$ with accuracy $\eps$. Moreover, the algorithm outputs the values $\Zin_{G,v}(\lambda),\Zout_{G,v}(\lambda)$.}
\begin{proposition}\label{lem:complexexp}
\statelemcomplexexp
\end{proposition}
Proposition~\ref{lem:complexexp} is extremely helpful in our reductions since it enables us to construct other gadgets very easily, e.g., equality gadgets that reduce the degree of a graph and gadgets  that can turn it into  a bipartite graph. The proofs of  Theorems~\ref{thm:norm} and~\ref{thm:arg} show how to use these gadgets to obtain $\numP$ hardness. In this proof outline, we focus on the  most difficult part which is the proof   of Proposition~\ref{lem:complexexp}.

To prove Proposition~\ref{lem:complexexp}, we will make use of the following multivariate map:
\begin{equation}\label{eq:funifmult}
(x_1,\hdots,x_{d})\mapsto \frac{1}{1+\lambda x_1\cdots x_d }, \mbox{ where } d:=\Delta-1.
\end{equation}
If, starting from 1, there is a sequence of operations \eqref{eq:funifmult} which ends with the value $x$, for the purposes of this outline, we will loosely say that ``we can generate the value $x$'' (the notion is formally defined in Definition~\ref{def:generates}). There is a simple correspondence between the values that we can generate and the activities that we can implement:  in Lemma~\ref{lem:hard-core-imp}, we show that if we can generate a value $x$, we can also implement the activity $\lambda x$ using a tree of maximum degree  $\Delta$.\footnote{\label{foot:generates}Note the extra factor of $\lambda$ when we pass to the implementation setting which  is to ensure the degree requirement in Item~\ref{it:voneGde} of Definition~\ref{def:Gimplement}; while the reader should not bother at this stage with this technical detail, the statements of our lemmas are usually about implementing activities and therefore have this extra factor $\lambda$ incorporated.} 

To get some insight about the map \eqref{eq:funifmult}, the first natural step is to look at the univariate case $x_1=\hdots=x_d=x$, where the map \eqref{eq:funifmult} simplifies into
\[f:x\mapsto \frac{1}{1+\lambda x^d}.\]
Even analysing the iterates of this map is a surprisingly intricate task; fortunately there is a rich theory concerning the iteration of complex rational maps which we can use (though much less is known in the multivariate setting!).  In the next section, we review the basic ingredients of the theory that we need, see \cite{Beardon, Milnor} for detailed accounts on the subject.

\subsection{Iteration of complex rational maps}\label{sec:iter}

We will use $\Riem=\Complex\cup \{\infty\}$ to denote the Riemann sphere (complex numbers with infinity). To handle $\infty$, it will be convenient to consider the \emph{chordal} metric $d(\cdot,\cdot)$ on the Riemann sphere $\Riem$, which is given for $z,w\in \Complex$ by
\[d(z,w)=\frac{2|z-w|}{(1+|z|^2)^{1/2}(1+|w|^2)^{1/2}}, \mbox{ and }
d(z,\infty)=\lim_{w\rightarrow\infty} d(z,w)=\frac{2}{(1+|z|^2)^{1/2}}.\]
Note that $d(z,w)$ is bounded by an absolute constant for all $z,w\in \Riem$.

Let $f:\Riem\rightarrow \Riem$ be a complex rational map, i.e., $f(z)=P(z)/Q(z)$ for some coprime polynomials $P,Q$. We define $f(\infty)$ as the limit of $f(z)$ when $z\rightarrow\infty$. The \emph{degree} of $f$ is the maximum of the degrees of $P,Q$.  A point $p\in \Complex$ is called a \emph{pole} of $f$ if $Q(p)=0$; when $p=\infty$, $p$ is a pole of $f$ if 0 is a pole of $f(1/z)$. 

Suppose that $z^*\in \Complex$ is a fixpoint of $f$, i.e., $f(z^*)=z^*$. The multiplier of $f$ at $z^*$ is given by $q=f'(z^*)$. If $z^*=\infty$, the multiplier of $f$ at $z^*$ is given by $1/f'(\infty)$. Depending on the value of $|q|$, the fixpoint $z^*$ is classified as follows: (i) \emph{attracting} if $|q|<1$, (ii) \emph{repelling} if $|q|>1$, and (iii) \emph{neutral} if $|q|=1$.  

For a non-negative integer $n\geq 0$, we will denote by $\fn{f}{n}$ the $n$-fold iterate of $f$ (for $n=0$, we let $\fn{f}{0}$ be the identity map). Given $z_0\in \Riem$, the sequence of points $\{z_n\}$ defined by $z_n=f(z_{n-1})=\fn{f}{n}(z_0)$ is called the \emph{orbit} of $z_0$.  

Given a rational map $f:\Riem\rightarrow \Riem$, we will be interested in the sensitivity of an orbit under small perturbations of the starting point. A point $z_0$ belongs to the \emph{Fatou} set if, for every $\epsilon>0$ there exists $\delta>0$ such that, for any point $z'$ with $d(z',z_0)\leq \delta$, it holds that $d(\fn{f}{n}(z'),\fn{f}{n}(z_0))\leq \eps$ for all positive integer $n$ (in other words, $z_0$ belongs to the Fatou set if the family of maps $\{\fn{f}{n}\}_{n\geq 1}$ is equicontinuous at $z_0$ under the chordal metric). A point $z_0$ belongs to the \emph{Julia} set if $z_0$ does not belong to the Fatou set (i.e., the Julia set is the complement of the Fatou set).  

\begin{lemma}[e.g., {\cite[Lemma 4.6]{Milnor}}]\label{lem:repelling}
Every repelling fixpoint belongs to the Julia set.
\end{lemma}

For $z\in \Riem$, the grand orbit $[z]$ is the set of points $z'$ whose orbit intersects the orbit of $z$, i.e., for every $z'\in [z]$, there exist integers $m,n\geq 0$ such that $\fn{f}{m}(z)=\fn{f}{n}(z')$. The exceptional set of the map $f$ is the set of points $z$ whose grand orbit $[z]$ is finite. As we shall see in the upcoming Lemma~\ref{lem:except}, the exceptional set of a rational map $f$ can have at most two points and, in our applications, it will in fact be empty. 

For $z_0\in \Complex$ and $r>0$, we use $B(z_0,r)$ to denote the open ball of radius $r$ around $z_0$. A set $U$ is a neighbourhood of $z_0$ if $U$ contains a ball $B(z_0,r)$ for some $r>0$. We will use the following fact.
\begin{theorem}[see, e.g., {\cite[Theorem 4.10]{Milnor}}]\label{thm:bazooka}
Let $f:\Riem\rightarrow \Riem$ be a complex rational map with exceptional set $E_f$. Let $z_0$ be a point in the Julia set and let $U$ be an arbitrary neighbourhood of $z_0$.  Then, the union of the forward images of $U$, i.e., the set $\bigcup_{n\geq 0} \fn{f}{n}(U)$, contains $\Riem\backslash E_f$.
\end{theorem}
Peters and Regts \cite{Peters} used a version of Theorem~\ref{thm:bazooka} to conclude the existence of trees $T$ and $\lambda$'s close to the boundary of $\LambdaD$ such that $Z_T(\lambda)=0$. We will use Theorem~\ref{thm:bazooka} as a tool to get our $\numP$-hardness results for any $\lambda$ outside the cardioid $\LambdaD$.

\subsection{A characterisation of the cardioid}\label{sec:tg4gt3}
To use the tools from the previous section, we will need to analyse the fixpoints of the map $f(z)=\frac{1}{1+\lambda z^{\Delta-1}}$. We denote by $\tLambdaD$ the following curve  (which is actually the boundary of the region $\LambdaD$ defined in \eqref{eq:region})\footnote{The fact that the curve $\tLambdaD$, as defined in \eqref{eq:regionbound}, is the boundary of the region $\LambdaD$ (defined in \eqref{eq:region}) follows from Lemma~\ref{lem:uniqueness}.
Lemma~\ref{lem:uniqueness} 
states that 
for every $\lambda\in \LambdaD$, there is a unique $z\in \mathbb{C}$ such that $|z|\leq 1/(\Delta-1)$ and $\lambda=\frac{z}{(1-z)^{\Delta}}$.
Thus, the function 
$g(z)=\frac{z}{(1-z)^{\Delta}}$ is holomorphic and injective on the open disc $U$ given by $|z|<1/(\Delta-1)$. By the open mapping theorem, we have that $g(U)$ is an open set. Moreover, $|g(z)|\leq |z|/(1-|z|)^{\Delta-1}$, so $g(U)$ is bounded.  Since $g$ extends injectively and continuously to the closed disc $|z|\leq 1/(\Delta-1)$, the boundary of $g(U)$ is given by the image of the circle $|z|=1/(\Delta-1)$. (Alternatively, one can also apply the domain invariance theorem to $g$.)}:
\begin{equation}\label{eq:regionbound}
\tLambdaD=\Big\{\lambda \in \mathbb{C}\,\big|\,\exists z\in \mathbb{C}:\ |z|= 1/(\Delta-1),\ \lambda=\frac{z}{(1-z)^{\Delta}}\Big\}.
\end{equation}
The following lemma is proved in Section~\ref{sec:cardioid}. 
\begin{lemma}\label{lem:fixpoints}
Let $\Delta\geq 3$ and consider the map $f(z)=\frac{1}{1+\lambda z^{\Delta-1}}$ for $\lambda \in\Complex$. Then,
\begin{enumerate}
\item For all $\lambda\in \LambdaD\backslash \tLambdaD$, $f$ has a unique attractive fixpoint; all other fixpoints are repelling.
\item For all $\lambda\in \tLambdaD$, $f$ has a unique neutral fixpoint; all other fixpoints are repelling.
\item For all $\lambda\notin \LambdaD$, all of the fixpoints of $f$ are repelling.
\end{enumerate}
\end{lemma}

\subsection{Applying the theory}\label{sec:vt4vv45665v}
We are now in a position to discuss in  detail   how to apply the tools of Section~\ref{sec:iter} and the result of Section~\ref{sec:tg4gt3}. Let $\lambda\in \CQ\backslash (\LambdaD\cup \Reals)$. By Lemma~\ref{lem:fixpoints}, all of the fixpoints of the map $f(z)=\frac{1}{1+\lambda z^{\Delta-1}}$ are repelling. By Lemma~\ref{lem:repelling}, all of the repelling fixpoints belong to the Julia set of the map, and therefore, by applying Theorem~\ref{thm:bazooka}, iteratively applying $f$ to a neighbourhood $U$ of a repelling fixpoint gives the entire complex plane.  Therefore, if we want to generate an arbitrary complex value $\lambda'\in \Complex$, it suffices to be able to generate values in a neighbourhood $U$ close to a repelling fixpoint of $f$. Of course, in our setting we will also need to do this efficiently, up to exponential precision. The following proposition is therefore the next important milestone. It formalises exactly what we need to show in order to be able to  prove Proposition~\ref{lem:complexexp}.

\newcommand{\statelemmainonetwoone}{
Let $\Delta\geq 3$ and $\lambda\in \CQ\setminus\Reals$, and set $d:=\Delta-1$. Let $\omega$ be the fixpoint of $f(x)=\frac{1}{1+\lambda x^{d}}$ with the smallest norm. There exists a rational $\rho>0$ such that the following holds.

There is a polynomial-time algorithm such that, on input $\lambda'\in B(\lambda\omega,\rho)\cap \CQ$ and rational $\epsilon>0$, outputs a bipartite graph $G$ of maximum degree at most  $\Delta$ with terminal $v$ that implements $\lambda'$ with accuracy $\epsilon$. Moreover, the algorithm outputs the values $\Zin_{G,v}(\lambda),\Zout_{G,v}(\lambda)$.
}
\begin{proposition}\label{lem:main121}
Let $\Delta\geq 3$ and $\lambda\in \CQ\setminus\Reals$, and set $d:=\Delta-1$. Let $\omega$ be the fixpoint of $f(x)=\frac{1}{1+\lambda x^{d}}$ with the smallest norm.\footnote{Note, by Lemma~\ref{lem:4frf46}, all the fixpoints of $f$ have different norms for $\lambda\in \CQ\setminus\Reals$, so $\omega$ is well-defined.} There exists a rational $\rho>0$ such that the following holds.

There is a polynomial-time algorithm such that, on input $\lambda'\in B(\lambda\omega,\rho)\cap \CQ$ and rational $\epsilon>0$, outputs a bipartite graph $G$ of maximum degree at most  $\Delta$ with terminal $v$ that implements $\lambda'$ with accuracy $\epsilon$. Moreover, the algorithm outputs the values $\Zin_{G,v}(\lambda),\Zout_{G,v}(\lambda)$.
\end{proposition}

To briefly explain why Proposition~\ref{lem:complexexp} follows from  Proposition~\ref{lem:main121}, we first show how to use  Proposition~\ref{lem:main121} to implement activities 
$\lambda x^*$ where $x^*$ is
close to a pole $p$ of $f$ (i.e., a point $p$ which  satisfies $1+\lambda p^d=0$). For some $r>0$, let $U$ be the ball $B(\omega,r)$ of radius $r$ around $\omega$. Using Theorem~\ref{thm:bazooka}, we find the first integer value of $N>0$ such that a pole $p^*$ belongs to $\fn{f}{N}(U)$; in fact, we can choose $r$ (see Lemma~\ref{lem:tranpole}) so that there exists a radius $r^*>0$ such that $B(p^*,r^*)\subseteq\fn{f}{N}(U)$. The idea of ``waiting till we hit the pole of $f$'' is that, up to this point, the iterates of $f$ satisfy Lipschitz inequalities, i.e., it can be shown that there exists a  real number $L>0$ such that $|\fn{f}{N}(x_1)-\fn{f}{N}(x_2)|\leq L|x_1-x_2|$ for all $x_1,x_2\in U$. 
Therefore, for any desired target $x^*\in B(p^*,r^*)$ we can find $w^*\in U$ such that $\fn{f}{N}(w^*)=x^*$.
We then
implement $\lambda w^*$ using  Proposition~\ref{lem:main121} with accuracy $\epsilon>0$.
Due to the Lipschitz inequality, this yields an implementation of $\lambda x^*$ with accuracy at most $\lambda L\epsilon$, i.e., just a constant factor distortion. 
Once we are able to create specified activities 
close to 
$\lambda p^*$ where $p^*$ is a pole
we move from there using the recurrence 
$f(x)=\frac{1}{1+\lambda x^d}$
and this enables us to  implement activities $\lambda'$ with large norm.
 After that, we use the implementations that we have  to implement activities $\lambda'$ with small norm and finally, $\lambda'$ with moderate value of $|\lambda'|$ as well. See the proof of Proposition~\ref{lem:complexexp} in Section~\ref{sec:complexexp} for more details.

\subsection{Chasing repelling fixpoints}\label{sec:chasing}

In this section, we focus on the proof of  Proposition~\ref{lem:main121}, whose proof 
(given in Section~\ref{sec:mainpart})
requires us to delve into the analysis of the multivariate map
\begin{equation*}\tag{\ref{eq:funifmult}}
(x_1,\hdots,x_{d})\mapsto \frac{1}{1+\lambda x_1\cdots x_d }, \mbox{ where } d:=\Delta-1.
\end{equation*} 
Recall, in the scope of proving  Proposition~\ref{lem:main121}, our goal is to generate points close to a repelling fixpoint of the map $f:x\mapsto \frac{1}{1+\lambda x^d}$. Since $\lambda$ is outside the cardioid region $\LambdaD$, the fixpoints of the map $f$ are repelling and therefore we cannot get close to any of them by just iterating $f$. 
Can the multivariate map  make it easier to get to
a fixpoint of $f$? The answer to the question is yes, as the following lemma asserts.

\begin{lemma}\label{hop2}
Let $\Delta\geq 3$ and $\lambda\in \Complex\setminus\Reals$, and set $d:=\Delta-1$. Let $\omega$ be the fixpoint of $f(x)=\frac{1}{1+\lambda x^{d}}$ with the smallest norm. For $k\geq 0$, let $x_k$ be the sequence defined by
\begin{equation}\label{rec22122}
x_0=x_1=\dots=x_{d-1}=1, \quad x_k = \frac{1}{1+\lambda \prod_{i=1}^d x_{k-i}}\quad\mbox{for $k\geq d$}.
\end{equation}
Then, the sequence $x_k$ is well-defined (i.e., the denominator of \eqref{rec22122} is nonzero for all $k\geq d$) and converges to the fixpoint $\omega$ as $k\rightarrow \infty$. Moreover, there exist infinitely many $k$ such that $x_k\neq \omega$.
\end{lemma}
Note, Lemma~\ref{lem:fixpoints} gurarantees that the fixpoint $\omega$ in Lemma~\ref{hop2} is repelling when $\lambda\in \Complex\backslash(\LambdaD\cup \Reals)$, so Lemma~\ref{hop2} indeed succeeds in getting us close to a repelling fixpoint 
in this case.  It is instructive  at this point to note that the sequence in \eqref{rec22122} corresponds to a Fibonacci-type tree construction $T_0,\dots,T_k$, where for $k\geq d$ tree $T_k$ consists of a root $r$ with subtrees $T_{k-d},\dots,T_{k-1}$ rooted at the children of $r$.
The trees $T_{k-d},\ldots,T_{k-1}$ generate the values $x_{k-d},\hdots,x_{k-1}$, respectively,
and the tree~$T_k$ generates the   value $x_k$.

A few remarks about the proof of Lemma~\ref{hop2} are in order. Analysing the behaviour of multivariate recurrences such as the one in \eqref{rec22122} is typically an extremely complicated task and the theory for understanding such recurrences appears to be still under development. Fortunately, the recurrence \eqref{rec22122} can be understood in a surprisingly simple way by using the linear recurrence $R_k$ defined by $R_0=\cdots=R_d=1$ and $R_{k+1}=R_{k}+\lambda R_{k-d}$ for all $k\geq d$, and observing that $x_k=R_k/R_{k+1}$ for all $k$. By interpreting $R_k$ as the independent set polynomial of a claw-free graph evaluated at $\lambda\in \Complex\backslash \Reals$, we obtain using a result of Chudnovsky and Seymour \cite{DBLP:journals/jct/ChudnovskyS07} that $R_k\neq 0$. The detailed proof of Lemma~\ref{hop2} can be found in Section~\ref{sec:f4ggefr}.

\subsection{Exponential precision via contracting maps that cover}\label{sec:gbeyn}
Lemma~\ref{hop2} resolves the intriguing task of getting close to a repelling fixpoint $\omega$ of the univariate map when $\lambda\in \Complex\backslash(\LambdaD \cup \Reals)$. But in the context of  Proposition~\ref{lem:main121} we need to accomplish far more: we need to be able to generate any point which is in  some (small) ball $U$ around the fixpoint $\omega$, with exponentially small error $\epsilon$. 

To do this, we will focus on a small ball $U$ around $\omega$, i.e.,  $U=B(\omega,\delta)$ for some sufficiently small $\delta>0$,  and we will examine how the multivariate map \eqref{eq:funifmult} behaves when $x_1,\hdots,x_d\in U$.  In particular, we show in Lemma~\ref{apx} that for any choice of $x_1,\hdots,x_d\in U$ it holds that 
\begin{equation}\label{eq:wdwcw}
\frac{1}{1+\lambda x_1\cdots x_d }\approx\omega+z\Big((x_1-\omega)+\hdots+(x_d-\omega)\Big), \mbox{ where $z$ satisfies $0<|z|<1$ and $z\in \Complex\backslash\Reals$.}
\end{equation}
To establish Equation~\eqref{eq:wdwcw} we choose $z=\omega-1$
so it is important (see Lemma~\ref{uiopa}) that 
$\omega$ satisfies $0<|\omega-1| < 1$.
Another important observation is that once we fix $x_1,\hdots,x_{d-1}\in B(\omega,\delta)$, the resulting map $\Phi$ is contracting with respect to the remaining argument $x_d$ (in the vicinity of $\omega$) --- see Lemma~\ref{derile} for a more detailed  treatment of this contraction.

The observation 
that $\Phi$ is contracting
will form the basis of our approach to iteratively reduce the accuracy with which we need to generate points (by going backwards): if we need to generate a desired $x\in U$ with error at most $\epsilon$ it suffices to be able to generate $\Phi^{-1}(x)$ with error at most $\epsilon/|z|>\epsilon$, i.e., to generate $x$ with  good accuracy, we only need to do the easier task of generating the point $\Phi^{-1}(x)$ with  
less restrictive accuracy.  The only trouble is that, if we use a single map $\Phi$, after a few iterations of the process the preimage $\Phi^{-1}(x)$ will eventually escape $U$. To address this, note that in the construction of the map $\Phi$ above, we had the freedom to choose arbitrary $x_1,\hdots,x_{d-1}\in B(\omega,\delta)$. We will make use of this freedom and, in particular, we will use a family of contracting maps $\Phi_1,\hdots,\Phi_t$ (for some large constant $t$) instead of a single map $\Phi$; the large number of maps will allow us to guarantee that for all $x\in U$, at least one of the preimages $\Phi_1^{-1}(x),\hdots, \Phi_t^{-1}(x)$ belongs in $U$, i.e., that the images $\Phi_1(U),\hdots, \Phi_t(U)$ cover $U$. We will discuss  in Section~\ref{sec:qzasw} how to obtain the maps $\Phi_1,\hdots, \Phi_t$,  but first let us formalise the above into the following lemma, which is the basis of our technique for making the error exponentially small.

\newcommand{\statelemprecision}{Let $z_0\in \CQ$, $r>0$ be a rational and $U$ be the ball $B(z_0,r)$. Further, suppose that   $\lambda_1',\hdots,\lambda_t'\in \CQ$ are such that the complex maps $\Phi_i:z\mapsto \frac{1}{1+\lambda_i' z}$ with $i\in[t]$ satisfy the following:
\begin{enumerate}
\item for each $i\in[t]$, $\Phi_i$ is contracting on  the ball $U$,
\item $U\subseteq \bigcup^{t}_{i=1}\Phi_i(U)$.
\end{enumerate}
There is an algorithm which, on input (i) a starting point $x_0\in U\cap \CQ$, (ii) a target $x\in U\cap \CQ$,  and (iii) a rational $\epsilon>0$, outputs in $poly(\size{x_0,x,\eps})$ time a number $\hat{x}\in U\cap \CQ$ and  a sequence $i_1,i_2,\hdots,i_k\in [t]$ such that  
\[\hat{x}=\Phi_{i_k}(\Phi_{i_{k-1}}(\cdots\Phi_{i_1}(x_0)\cdots))\mbox{ and }|\hat{x}-x|\leq \epsilon.\]}
\begin{lemma}\label{lem:precision}
\statelemprecision
\end{lemma}
The proof of Lemma~\ref{lem:precision} can be carried out along the lines we sketched above, see Section~\ref{sec:mainidea} for details.  In that section, we also pair Lemma~\ref{lem:precision} with a path construction which, given the sequence of indices  $i_1,\hdots,i_k$, returns a path of length $k$ that implements $\lambda \hat{x}$ (cf. footnote~\ref{foot:generates} for the extra factor of $\lambda$), see Lemma~\ref{lem:pathpath} for details.

\subsection{Constructing the  maps}\label{sec:qzasw}
We next turn to the last missing piece, which is to create the maps $\Phi_1,\hdots,\Phi_t$ which satisfy the hypotheses of  Lemma~\ref{lem:precision} in a ball $U=B(\omega,\delta)$ around the fixpoint $\omega$ for some small radius $\delta>0$ (note, we are free to make $\delta$ as small as we wish). The following (standard) notions of ``covering'' and ``density'' will be relevant for this section.
\begin{definition}
Let $U\subseteq \Complex$. A set $F\subseteq U$ is called an $\eps$-covering of $U$ if for every $x\in U$ there
exists $y\in F$ such that $|x-y|\leq\eps$. A set $F\subseteq U$ is called dense in $U$ if $F$ is an $\eps$-covering of $U$ for every $\epsilon>0$.
\end{definition}
We have already seen in Section~\ref{sec:gbeyn} that,  for arbitrary $x_1,\hdots,x_d\in U$, we have  
\begin{equation*}\tag{\ref{eq:wdwcw}}
\frac{1}{1+\lambda x_1\cdots x_d }\approx\omega+z\Big((x_1-\omega)+\hdots+(x_d-\omega)\Big), \mbox{ where $z$ satisfies $0<|z|<1$ and $z\in \Complex\backslash\Reals$.}
\end{equation*}
We also discussed that, if we fix arbitrary $x_1,\hdots,x_{d-1}\in U$, the resulting map $\Phi(x)=\frac{1}{1+(\lambda x_1\cdots x_{d-1})x}$ is contracting in $U$ for all sufficently small $\delta>0$, and therefore we can easily take care of the contraction properties that we need (in the context of Lemma~\ref{lem:precision}). The more difficult part is to control the preimage of the map $\Phi$. We show in Lemma~\ref{apxinv} that for $x,x_1,\hdots,x_{d-1}\in U$, it holds that 
\[\Phi^{-1}(x)=\frac{1}{\lambda x_1 \cdots x_{d-1}}\Big(\frac{1}{x}-1\Big)\approx \omega + \Big( \frac{x-\omega}{z} - \sum_{j=1}^{d-1} (x_j-\omega) \Big).\]
Therefore to ensure that $\Phi^{-1}(x)$ belonds to $U=B(\omega,\delta)$ we need to ensure that $x_1,\hdots x_{d-1}$ are such that
\begin{equation}\label{eq:55v21324}
\Big| \frac{x-\omega}{z} - \sum_{j=1}^{d-1} (x_j-\omega) \Big|<\delta/2.
\end{equation}
Note that by Lemma~\ref{hop2} we can generate points arbitrarily close to $\omega$ and hence we can make each of $x_2-\omega,\hdots, x_{d-1}-\omega$ so small that they are effectively negligible in  \eqref{eq:55v21324}; then, to be able to satisfy \eqref{eq:55v21324}, we need to be able to choose $x_1$ so that $|(x-\omega)/z-(x_1-\omega)|$ is small, say less than $\delta/4$. Since $|(x-\omega)/z|\leq \delta/|z|$, the key will therefore be to produce a $(\delta/4)$-covering of the  slightly enlarged ball $B(\omega, \delta/|z|)$. Then, we can take $x_1$ to be one of the points in the $(\delta/4)$-covering. 

We will in fact show the following slightly more general lemma, which guarantees that we can indeed generate the required points around $\omega$ for any desired precision $\epsilon>0$ provided that we choose $\delta$ small enough (and can therefore implement activities around $\lambda \omega$). Note that the lemma can be viewed as a ``relaxed" version of  Proposition~\ref{lem:main121} with much weaker guarantees.

\newcommand{\statelemmaopoptwo}{
Let $\Delta\geq 3$ and $\lambda\in \CQ\setminus\Reals$, and set $d:=\Delta-1$. Let $\omega$ be the fixpoint of $f(x)=\frac{1}{1+\lambda x^{d}}$ with the smallest norm. For any $\eps,\kappa>0$ there exists  a radius $\rho\in (0,\kappa)$ such that the following holds. 
For every $\lambda'\in B(\lambda\omega,\rho)$, there exists a tree $G$ of maximum degree at most  $\Delta$ that implements $\lambda'$ with accuracy $\rho\epsilon$.  
}
\begin{lemma}\label{opop2}
\statelemmaopoptwo
\end{lemma}

But how can we ``populate'' the vicinity of $\omega$, i.e., generate a covering of a ball $U=B(\omega,\delta)$? Lemma~\ref{hop2} only  shows that we can generate points arbitrarily close to $\omega$. The key once again is to use the multivariate map around $\omega$ and, in particular, the perturbation estimate in the r.h.s of \eqref{eq:wdwcw}. To focus on the displacement from $\omega$, we will use the transformation $a_i=x_i-\omega$ so that \eqref{eq:wdwcw} translates into the following operation
\[(a_1,\hdots,a_d)\mapsto z(a_1+\cdots+a_d),\]
i.e., if we have generated points which are displaced by $a_1,\hdots,a_d$ from $\omega$, we can also generate a point which is roughly displaced by $z(a_1+\cdots+a_d)$ from $\omega$; we will only need to apply the operation a finite number of times, so the error coming from \eqref{eq:wdwcw} will not matter critically and can be ignored in the following. We show in Lemma~\ref{getpo} that, using a sequence of such operations, we can generate points   of the form $\omega+z^{N(p)}p(z)$ where $p$ is an arbitrary polynomial with non-negative integer coefficients and $N(p)$ is a positive integer which is determined by the number of operations we used to create $p$. We further show in Lemma~\ref{polydense}, that for all $z\in \Complex\backslash \Reals$ with $|z|<1$, the values $p(z)$, as $p$ ranges over all
polynomials  with non-negative integer coefficients, form a dense set of $\Complex$. Therefore, to obtain Lemma~\ref{opop2}, we can choose an $\epsilon$-covering $F$ of the unit disc using a finite set of values $p(z)$ and set $\delta=z^N$ where $N=\max_{p\in F}{N(p)}$; then, we can generate the points $\omega+\delta p(z)$ for every $p\in F$, which form an $(\epsilon \delta)$-covering of the ball $U=B(\omega,\delta)$, yielding  Lemma~\ref{opop2}.
The full proof is in Section~\ref{sec:45g45g4b}.

\subsection{Fitting the pieces together and proof for the real case}\label{sec:realmod}

We briefly summarise the proof of  Proposition~\ref{lem:main121}. First, we get points close to a repelling fixpoint by showing Lemma~\ref{hop2} (discussed in Section~\ref{sec:chasing} and proved in Section~\ref{sec:f4ggefr}). Then, we bootstrap this into a moderately dense set of points around the fixpoint, yielding Lemma~\ref{opop2} (discussed in Section~\ref{sec:qzasw} and proved in Section~\ref{sec:45g45g4b}). Further, we bootstrap this into exponential precision around the fixpoint using Lemma~\ref{lem:precision} (discussed in Section~\ref{sec:gbeyn} and proved in Section~\ref{sec:tv45vef}). Finally, we propagate this exponential precision to the whole complex plane, therefore yielding  Proposition~\ref{lem:main121} (discussed in Section~\ref{sec:vt4vv45665v} and proved in Section~\ref{sec:complexexp}).

Finally, we mention the modifications needed for the real case when $\lambda<-\lambda^*$. The following proposition is the analogue of Proposition~\ref{lem:complexexp} and allows us to implement real activities with exponential precision.
\newcommand{\statelemrealexp}{Let $\Delta\geq 3$ and $\lambda \in \QQ$ be such that $\lambda<-\lambda^*$.  

There is an algorithm which, on input $\lambda',\eps\in \QQ$ with $\eps>0$, outputs in $poly(\size{\lambda',\eps})$ time a bipartite graph $G$ of maximum degree at most  $\Delta$ with terminal $v$ that implements $\lambda'$ with accuracy $\eps$. Moreover, the algorithm outputs the values $\Zin_{G,v}(\lambda),\Zout_{G,v}(\lambda)$.}
\begin{proposition}\label{lem:realexp}
\statelemrealexp
\end{proposition}

As in the complex case, we will need a moderately dense set of activities to get started, i.e., an analogue of Lemma~\ref{opop2}; here, our job is somewhat simplified (relative to the case where $\lambda\in \Complex\backslash \Reals$) since  we can use the  following result of \cite{GGS} .
\begin{lemma}[{\cite[Lemma 4]{GGS}}]\label{lem:realeps}
Let $\Delta\geq 3$ and $\lambda<-\lambda^*$.  Then, for every $\lambda'\in\mathbb{R}$, for every $\epsilon>0$, there exists a bipartite graph $G$ of maximum degree at most  $\Delta$ that implements $\lambda'$ with accuracy $\epsilon$. 
\end{lemma}
Note that Lemma~\ref{lem:realeps} does not control the size of the graph $G$ with respect to the accuracy $\epsilon$, so it does not suffice to prove Proposition~\ref{lem:realexp} on its own. In order to do this, we use the ``contracting maps that cover''  technique to get the exponential precision, i.e.,  the analogue of Lemma~\ref{lem:precision} restricted to the reals (see Lemma~\ref{lem:precisionreals}). The proof of Proposition~\ref{lem:realexp} is completed in Section~\ref{sec:realexp}.

Once the proofs of Propositions~\ref{lem:complexexp} and~\ref{lem:realexp} are in place, we give the proofs of our $\numP$-hardness results in Section~\ref{sec:reductions}.

\subsection{Dependencies between Lemmas}

The proofs of some theorems, propositions and lemmas   depend directly upon other
theorems, propositions and lemmas that are proved in this paper.
To help the reader keep track of this, we provide Table~\ref{table:dep}.
Note that Theorems~\ref{thm:norm} and~\ref{thm:arg} follow directly from Theorems~\ref{thm:BipMV} and~\ref{thm:MVtoSV}.

\begin{table}[h]
\caption{Dependencies Between Lemmas, etc.}
\centering
\begin{tabular}{| l | l |}
\hline\hline
Result & Depends directly on \\
\hline
\hline
Proposition~\ref{lem:complexexp} & Proposition~\ref{lem:main121}, Lemma \ref{lem:tranpole}\\ \hline
Lemma~\ref{lem:fixpoints} & Lemma~\ref{lem:uniqueness}\\ \hline
Proposition~\ref{lem:main121} & Lemmas~\ref{lem:precision}, \ref{lem:pathpath}, \ref{lem:perturbcontract}\\ \hline
Lemma~\ref{hop2} & Lemmas~\ref{lem:4frf46}, \ref{lele1}\\ \hline
Lemma~\ref{opop2} & Lemmas~\ref{hop2}, \ref{lem:hard-core-imp}, \ref{uiopa}, \ref{polydense}, \ref{getpo}, \ref{apx}\\ \hline
Proposition~\ref{lem:realexp} & Lemmas~\ref{lem:precisionreals},  \ref{lem:pathpath}, \ref{lem:numerics} \\ \hline
Lemma~\ref{lem:tranpole} & Lemma~\ref{lem:fixpoints} \\ \hline
Lemma~\ref{lem:equality} & Lemmas~\ref{lem:equalitygadget}, \ref{lem:minusoneeqgadget} \\ \hline
Theorem~\ref{thm:MV} & Propositions~\ref{lem:complexexp}, \ref{lem:realexp}, Lemma~\ref{lem:equality} \\ \hline
Theorem~\ref{thm:BipMV} & Theorem~\ref{thm:MV} \\ \hline
Theorem~\ref{thm:MVtoSV} & Propositions~\ref{lem:complexexp}, \ref{lem:realexp}, Lemma~\ref{lem:LB} \\ \hline
Lemma~\ref{uiopa} & Lemma~\ref{lem:4frf46} \\ \hline
Lemma~\ref{derile} & Lemmas~\ref{uiopa}, \ref{apx} \\ \hline
Lemma~\ref{lem:perturbcontract} & Lemmas~\ref{opop2}, \ref{uiopa},  \ref{apx}, \ref{apxinv}, \ref{derile} \\
\hline
\hline
\end{tabular}
\label{table:dep}
\end{table}

\section{Preliminaries}\label{sec:prelims}
\subsection{Implementing activities}

We recall the following definitions from \cite{GGS}, which we modify (slightly) here to account for complex activities.

\begin{definition}\label{def:implement}
Let $\Delta\geq 2$ be an integer and $\lambda\in \Cnz$. We say that $(\Delta,\lambda)$ \emph{implements} the activity $\lambda'\in\Complex$ if there is a \emph{bipartite} graph $G$ of maximum degree at most  $\Delta$ which implements the activity $\lambda'$. More generally, we say that $(\Delta,\lambda)$ implements a set of activities $S\subseteq \Complex$, if for every $\lambda'\in S$ it holds that $(\Delta,\lambda)$ implements $\lambda'$.
\end{definition}

Implementing activities allows us to modify the activity at a particular vertex $v$. As in \cite{GGS}, it will therefore be useful to consider the hard-core model with non-uniform activities. Let $G=(V,E)$ be a graph and $\lambdab=\{\lambda_v\}_{v\in V}$ be a complex vector, so that $\lambda_v$ is the activity of the vertex  $v\in V$. The hard-core partition function with activity vector $\lambdab$ is defined as
\[Z_G(\lambdab)=\sum_{I\in \mathcal{I}_G}\prod_{v\in I} \lambda_v.\]
Note that by setting all vertex activities equal to $\lambda$ we obtain the standard hard-core model with activity $\lambda$.  For a vertex $v\in V$, we define $\Zin_G(\lambdab)$ and $\Zout_G(\lambdab)$ for the non-uniform model analogously to $\Zin_G(\lambda)$ and $\Zout_G(\lambda)$ for the uniform model, respectively.

The following lemma is proved in \cite{GGS} for real values but the proof holds verbatim in the complex setting as well. The lemma connects the partition function $Z_G(\lambdab)$ of a graph $G$ with non-uniform activities to the  partition function $Z_{G'}(\lambda)$ of an augmented graph $G'$ with uniform activity $\lambda$ (the augmented graph $G'$ is obtained by sticking on each vertex $v$ of $G$ a graph $G_v$ which implements the activity $\lambda_v$). 
\begin{lemma}[{\cite[Lemma 5]{GGS}}]\label{lem:transf}
Let $\lambda\in \Cnz$, let $t\geq 1$ be an arbitrary integer and let $\lambda_1',\hdots,\lambda_t'\in \Complex$. Suppose that, for $j\in [t]$, the graph $G_j$ with terminal $v_j$ implements the activity $\lambda_j'$, and let  $C_j:=\Zout_{G_j,v_j}(\lambda)$. Then, the following holds for every graph $G=(V,E)$ and every activity vector $\lambdab=\{\lambda_v\}_{v\in V}$ such that $\lambda_v\in\{\lambda,\lambda_1',\hdots,\lambda_t'\}$ for every $v\in V$.

For $j\in[t]$, let $V_j:=\{v\in V\mid \lambda_v=\lambda_j'\}$. Consider the graph $G'$ obtained from $G$ by attaching, for every $j\in [t]$ and every vertex $v\in V_j$, a copy of the graph $G_j$ to the vertex $v$ and identifying the terminal $v_j$ with the vertex $v$. Then, for $C:=\prod^t_{j=1}C_j^{|V_j|}$, it holds that
\begin{gather}
Z_{G'}(\lambda)=C\cdot Z_G(\lambdab),\label{eq:uni}\\
\mbox{$\forall v\in V$: } \Zin_{G',v}(\lambda)=C\cdot \Zin_{G,v}(\lambdab), \qquad \Zout_{G',v}(\lambda)=C\cdot \Zout_{G,v}(\lambdab).\label{eq:uniinout}
\end{gather}
\end{lemma}
\begin{remark}\label{rem:observ}
As noted in \cite[Remark 6]{GGS}, the construction of $G'$ in Lemma~\ref{lem:transf} ensures that  the degree of every vertex $v$ in $G$ with $\lambda_v=\lambda$ maintains its degree, while  the degree of every other vertex $v$ in $G$ gets increased by one.  Also, if the graph $G$ is bipartite and the graphs $G_j$ are bipartite for all $j=1,\hdots,t$, then $G'$ is bipartite as well. 

These observations will ensure in later applications of Lemma~\ref{lem:transf} that we do not blow up the degree and that we preserve the bipartiteness of the underlying graph $G$.
\end{remark}

\subsection{Finding roots of polynomials}
To prove our $\numP$-hardness results,  we will sometimes need in our reductions to compute with  accuracy $\epsilon$ roots of polynomial equations with coefficients in $\CQ$. We will therefore need some basic results that these procedures can be carried out in polynomial time. To formalise the running time, we will use the following definition for the size of a number in $\CQ$.

\begin{definition}\label{def:CQ}
Let $\CQ$ be the set of complex numbers whose real and imaginary parts are rationals. Let $\alpha\in\CQ$ and write $\alpha=\displaystyle\frac{a}{b}+\mathrm{i}\frac{c}{d}$, where  $a,b,c,d$ are integers such that $\mathrm{gcd}(a,b)=1$, $\mathrm{gcd}(c,d)=1$. Then, the \emph{size} of $\alpha$, denoted by $\size{\alpha}$, is given by $1+\log(|a|+|b|+|c|+|d|)$.

For $\alpha_1,\hdots,\alpha_t\in \CQ$, we denote by $\size{\alpha_1,\hdots,\alpha_t}$ the total of the sizes of $\alpha_1,\hdots,\alpha_t$.
\end{definition}

We will need the following fact for finding roots of polynomials with complex coefficients.
\begin{lemma}[see, e.g., {\cite{NEFF}}]\label{fact:root}
There is an algorithm which, on input a coefficient list $c_0,\hdots,c_n\in \CQ$ 
with $c_n\neq 0$
and a rational $\epsilon>0$, outputs in time $poly(n,\size{c_0,\hdots,c_n,\epsilon})$ numbers $\hat{\rho}_1,\hdots,\hat{\rho}_n\in \CQ$ such that 
\[|\rho_1-\hat{\rho}_1|,\hdots,|\rho_n-\hat{\rho}_n|\leq \epsilon,\]
where ${\rho}_1,\hdots,{\rho}_n$ are the roots of the polynomial $P(x)=\sum^n_{i=0}c_ix^i$.
\end{lemma}

\subsection{Lower bounds on polynomials evaluated at algebraic numbers}
Let 
$P(x) = \sum_{i=0}^n a_i x^i$ 
be a polynomial with complex coefficients.
The (naive) \emph{height} of $P(x)$ is defined as
$H(P) = \max_i \{|a_i|\}$. 
In transcendental number theory, 
the height of a polynomial in~$x$ is used to give a lower bound on its value
when  $x$ is an algebraic number.  The simple version of the lower bound
 that we use (Lemma~\ref{lem:use} below)  is from \cite[Lemma 6.3]{ComplexIsing} 
but the proof  is entirely standard, and
the proof given in~\cite{ComplexIsing}
 is taken  from the proof of Theorem A.1 of Bugeaud's book~\cite{Bug04}. 

Recall, that the minimal polynomial for an algebraic number is the monic polynomial of minimum degree and rational coefficients which has $\alpha$ as a root; the degree of an algebraic number is the degree of its minimal polynomial. 
\begin{lemma}[Liouville's inequality]
\label{lem:use}Let $P(x)$ be an integer polynomial of degree~$n$ and $y\in \Complex$ be an algebraic number of degree~$d$.
Then either $p(y)=0$ or $|P(y)| \geq c_y^{-n} {(n+1)H(P)}^{-d+1}$,
where $c_y>1$ is an effectively computable constant that only depends on~$y$.
\end{lemma}

\subsection{Characterising the exceptional set}
We conclude this section with the following characterisation of the exceptional set (cf. Theorem~\ref{thm:bazooka}).
\begin{lemma}[see {\cite[Lemma 4.9]{Milnor}} and {\cite[Theorem 4.1.2]{Beardon}}]\label{lem:except}
Let $f:\Riem\rightarrow \Riem$ be a complex rational map of degree $\geq 2$ and let $E_f$ denote its exceptional set. Then, $|E_f|\leq 2$. Moreover, 
\begin{itemize}
\item if $E_f=\{\zeta\}$, then $\zeta$ is a fixpoint of $f$ with multiplier 0.
\item if $E_f=\{\zeta_1,\zeta_2\}$, then $\zeta_1,\zeta_2$ are fixpoints of $\fn{f}{2}$ with multiplier 0.
\end{itemize}
\end{lemma}

\section{Proving the properties of the cardioid}\label{sec:cardioid}
In this section, we prove Lemma~\ref{lem:fixpoints} which classifies the fixpoints of the map $f(z)=\frac{1}{1+\lambda z^{\Delta-1}}$ depending on the value of $\lambda$. 

We start by proving the following property of the region $\LambdaD$ defined in \eqref{eq:region}. 
\begin{lemma}\label{lem:uniqueness} 
Let $\Delta\geq 3$ be an integer. Then, for every $\lambda\in \LambdaD$, there is a unique $z\in \mathbb{C}$ such that $|z|\leq 1/(\Delta-1)$ and $\lambda=\frac{z}{(1-z)^{\Delta}}$.
\end{lemma}
\begin{proof}
Existence of $z\in \mathbb{C}$ with the required properties is immediate by the definition of $\LambdaD$. 

To show uniqueness, set $d:=\Delta-1$, and assume for the sake of contradiction  that there exist $x,y\in \Complex$ with $x\neq y$ such that $|x|,|y|\leq 1/d$ and $\lambda=\frac{x}{(1-x)^{d+1}}=\frac{y}{(1-y)^{d+1}}$. This gives 
\[x(1-y)^{d+1}-y(1-x)^{d+1}=0.\]
By expanding terms we obtain that
\[x(1-y)^{d+1}-y(1-x)^{d+1}=(x-y)-\sum^{d+1}_{k=2}\binom{d+1}{k}(-1)^{k}xy(x^{k-1}-y^{k-1}).\]
Since $x\neq y$ and $x(1-y)^{d+1}-y(1-x)^{d+1}=0$, we can factor out $(x-y)$ to obtain that 
\[M=1, \mbox{ where } M:=\sum^{d+1}_{k=2}\binom{d+1}{k}(-1)^{k}xy\sum^{k-2}_{j=0}x^{j}y^{(k-2)-j}.\]
We will first show that $\abs{M}=1$ only if $x,y$ are conjugate complex numbers satisfying $|x|=|y|=1/d$. Then, we will bootstrap the argument to show that $|M|=1$ further implies that $x,y\in \mathbb{R}$. Thus, we will get that $x=y$, contradicting that $x\neq y$.

The main observation is that since $|x|,|y|\leq 1/d$,  by the triangle inequality we have that 
\begin{align}
1=\abs{M}=\abs{\sum^{d+1}_{k=2}\binom{d+1}{k}(-1)^{k}xy\sum^{k-2}_{j=0}x^{j}y^{(k-2)-j}}&\leq \sum^{d+1}_{k=2}\binom{d+1}{k}|x||y|\sum^{k-2}_{j=0}|x|^{j}|y|^{(k-2)-j}\notag\\
&\leq \sum^{d+1}_{k=2}\binom{d+1}{k}\frac{k-1}{d^k}=1,\label{eq:Anormbound2}
\end{align}
where the last equality follows from subtracting the equalities
\begin{equation*}
\begin{aligned}
\sum^{d+1}_{k=2}\binom{d+1}{k}\frac{k}{d^k}&=(d+1)\sum^{d+1}_{k=2}\binom{d}{k-1}\frac{1}{d^k}=\frac{d+1}{d}\sum^{d}_{k=1}\binom{d}{k}\frac{1}{d^k}=\frac{d+1}{d}\bigg(\Big(1+\frac{1}{d}\Big)^d-1\bigg),\\
\sum^{d+1}_{k=2}\binom{d+1}{k}\frac{1}{d^k}&=\Big(1+\frac{1}{d}\Big)^{d+1}-1-\frac{d+1}{d}.
\end{aligned}
\end{equation*}
It follows that the inequality in  \eqref{eq:Anormbound2} must hold at equality, from where we obtain  that $|x|=|y|=1/d$. Further, since $\frac{x}{(1-x)^{d+1}}=\frac{y}{(1-y)^{d+1}}$, we obtain that $\abs{1-x}=\abs{1-y}$. Now, since $x\neq y$, from $|x|=|y|=1/d$ and $\abs{1-x}=\abs{1-y}$, we obtain that $x$ and $y$ are conjugate complex numbers with $|x|=|y|=1/d$. 

Since $x\neq y$ and $x,y$ are conjugates, the imaginary parts of both $x$ and $y$ must be nonzero. W.l.o.g., we may assume that the imaginary part of $x$ is positive, so we can set $x=\frac{1}{d}\emm^{\im \theta}$ and $y=\frac{1}{d}\emm^{-\im \theta}$ for some $\theta\in(0,\pi)$. For each $k=2,\hdots,d+1$, let 
$M_k:=xy\sum^{k-2}_{j=0}x^{j}y^{(k-2)-j}$. Observe that $M_k$ is a real number (since it  equals its conjugate) satisfying $|M_k|\leq (k-1)/d^{k}$. In fact, we have that 
\[M_k=xy\frac{x^{k-1}-y^{k-1}}{x-y}=\frac{\sin((k-1) \theta)}{d^k\sin \theta}\]
for each $k=2,\hdots,d+1$ and in particular $M_2=1/d^2$.
Using the triangle inequality again, we have that 
\[1=\abs{M}=\abs{\sum^{d+1}_{k=2}\binom{d+1}{k}(-1)^{k}M_k}\leq \sum^{d+1}_{k=2}\binom{d+1}{k}\frac{k-1}{d^k}=1,\]
so the inequality must hold at equality. For this to happen, and since the $M_k$'s are real numbers, it must be the case that there exists $s\in\{0,1\}$ so that $M_k=(-1)^{k+s}(k-1)/d^{k}$ for each $k=2,\hdots,d+1$. From $M_2=1/d^2$, we obtain that $s=0$. Note that $d\geq 2$, so $M_3=-2/d^3$ yields that $\sin(2\theta)=-2\sin(\theta)$. By the identity $\sin(2\theta)=2\sin\theta\cos\theta$, we obtain that at least one of $\sin\theta=0$ or $\cos\theta=-1$ must hold, which contradicts  that $\theta\in(0,\pi)$. 

Thus, it must be the case that $x=y$, concluding the proof of Lemma~\ref{lem:uniqueness}. 
\end{proof}
We are now ready to prove Lemma~\ref{lem:fixpoints}, which we restate here for convenience.  
\newcommand{\statelemfixpoints}{
Let $\Delta\geq 3$ and consider the map $f(z)=\frac{1}{1+\lambda z^{\Delta-1}}$ for $\lambda \in\Complex$. Then,
\begin{enumerate}
\item For all $\lambda\in \LambdaD\backslash \tLambdaD$, $f$ has a unique attractive fixpoint; all other fixpoints are repelling.
\item  For all $\lambda\in \tLambdaD$, $f$ has a unique neutral fixpoint; all other fixpoints are repelling.
\item For all $\lambda\notin \LambdaD$, all of the fixpoints of $f$ are repelling.
}
\begin{lemfixpoints}
Let $\Delta\geq 3$ and consider the map $f(z)=\frac{1}{1+\lambda z^{\Delta-1}}$ for $\lambda \in\Complex$. Then,
\begin{enumerate}
\item \label{it:interior} For all $\lambda\in \LambdaD\backslash \tLambdaD$, $f$ has a unique attractive fixpoint; all other fixpoints are repelling.
\item \label{it:boundary} For all $\lambda\in \tLambdaD$, $f$ has a unique neutral fixpoint; all other fixpoints are repelling.
\item \label{it:exterior} For all $\lambda\notin \LambdaD$, all of the fixpoints of $f$ are repelling.
\end{enumerate}
\end{lemfixpoints}
\begin{proof}
Let $\omega\in \mathbb{C}$ be an arbitrary fixpoint of $f$ so that 
\begin{equation}\label{eq:wfixp23}
\omega=\frac{1}{1+\lambda \omega^{\Delta-1}} \mbox{ or equivalently }\lambda \omega^{\Delta}=1-\omega.
\end{equation}
Let $q:=f'(\omega)$ be the multiplier of $f$ at $z=\omega$. We have that
\begin{equation}\label{eq:fmultiplier}
q=f'(\omega)=-\frac{(\Delta-1)\lambda  \omega^{\Delta-2}}{(1+\lambda \omega^{\Delta-1})^2}=-(\Delta-1)\lambda  \omega^\Delta=-(\Delta-1)(1-\omega),
\end{equation}
where in the latter two equalities we used \eqref{eq:wfixp23}. Let $\widehat{\omega}:=1-\omega$. Then, \eqref{eq:wfixp23} and \eqref{eq:fmultiplier} give
\begin{equation}\label{eq:rewrite34356}
\lambda=\frac{\widehat{\omega}}{(1-\widehat{\omega})^\Delta}, \quad |q|=(\Delta-1)\abs{\widehat{\omega}}.
\end{equation}
We are now ready to prove the lemma. Let $\omega_1,\hdots, \omega_t$ denote the distinct fixpoints of $f$ (note that $1\leq t\leq \Delta$) and for $i\in[t]$, let $q_i$ be the multiplier of $f$ at $\omega_i$. Further, let $\widehat{\omega_i}=1-\omega_i$. Then, \eqref{eq:rewrite34356} gives that for all $i\in[t]$ it holds that
\begin{equation}\label{eq:rewrite345}
\lambda=\frac{\widehat{\omega_i}}{(1-\widehat{\omega_i})^\Delta}, \quad |q_i|=(\Delta-1)\abs{\widehat{\omega_i}}.
\end{equation}

For $\lambda\in \LambdaD\backslash \tLambdaD$, there exists $j\in [t]$ such that $\abs{\widehat{\omega_{j}}}<1/(\Delta-1)$ (by the definition of the regions $\LambdaD$ and $\tLambdaD$) and hence $|q_{j}|<1$. By Lemma~\ref{lem:uniqueness}, for all $i\in [t]$ with $i\neq j$ it holds that  $\abs{\widehat{\omega_i}}>1/(\Delta-1)$ and hence $|q_{i}|>1$. This proves  Item~\ref{it:interior}.

For $\lambda\in \tLambdaD$, there exists $j\in [t]$ such that $\abs{\widehat{\omega_{j}}}=1/(\Delta-1)$ (by the definition of the region $\tLambdaD$) and hence $|q_{j}|=1$. Again, by Lemma~\ref{lem:uniqueness}, for all $i\in [t]$ with $i\neq j$ it holds that  $\abs{\widehat{\omega_i}}>1/(\Delta-1)$ and hence $|q_{i}|>1$. This proves  Item~\ref{it:boundary}.

For $\lambda\notin \LambdaD $, we have that every $\widehat{\omega_i}$ satisfies $\abs{\widehat{\omega_i}}>1/(\Delta-1)$ (by the definition of the region $\LambdaD$) and hence $|q_i|>1$. This proves Item~\ref{it:exterior}.

This concludes the proof of Lemma~\ref{lem:fixpoints}.
\end{proof}

We close this section with the following lemma that applies to all $\lambda \in \Complex\backslash \Reals$.
\begin{lemma}\label{lem:4frf46}
Let $\lambda\in\Complex\backslash \Reals$ and $\Delta\geq 3$. Then, the fixpoints of the map $f(z)=\frac{1}{1+\lambda z^{\Delta-1}}$ have (pairwise) distinct norms. 
\end{lemma}
\begin{proof}
For convenience, let $d:=\Delta-1$. The fixpoints are roots of $\lambda z^{d+1} + z - 1$. It will be convenient to reparameterize
$z=1/y$ and consider the roots of $\lambda + z^d - z^{d+1}$. (Note that for $y_1=1/z_1, y_2=1/z_2$ we have $|z_1|=|z_2|$ iff $|y_1|=|y_2|$.) 

For the sake of contradiction, suppose that $\lambda = y^{d+1} - y^{d}$
has two roots $y_1,y_2$ of the same norm, that is $|y_1|=|y_2|$. We have
$$|\lambda| = |y_1|^d (|y_1-1|) = |y_2|^d (|y_2-1|),$$
and since $|y_1|=|y_2|$ we conclude $|y_1-1| = |y_2-1|$.
Note that this means that $y_1,y_2$ lie on the intersection of
two circles---one centered at $0$ and one centered at $1$.
This means that $y_1$ and $y_2$ are conjugate (since the centers lie on
the real line and if two circles intersect in 2 points then the
points are symmetric about the segment connecting the centers) and thus
$$
\overline{\lambda} = \overline{y_1}^{d+1} - \overline{y_1}^d = y_2^{d+1} - y_2^d = \lambda,
$$
that is $\lambda\in\Reals$.
\end{proof}

\section{Impementing activities with exponential precision}\label{sec:main1}
In this section, we prove our main implementation results, Propositions~\ref{lem:complexexp} and~\ref{lem:realexp}. We start in Section~\ref{sec:mainidea} by proving Lemma~\ref{lem:precision} and its analogue for the real case that will help us obtain the exponential precision; we also give a path construction that will give us the desired implementations. Then, in Section~\ref{sec:realexp}, we give the proof of Proposition~\ref{lem:realexp} for the real case and, in Section~\ref{sec:complexexp}, the proof of Proposition~\ref{lem:complexexp} for the complex case. 

\subsection{Contracting maps that cover yield exponential precision}\label{sec:mainidea}
We first restate here Lemma~\ref{lem:precision} for convenience.
\begin{lemprecision}
\statelemprecision
\end{lemprecision}
The following is the exact analogue of Lemma~\ref{lem:precision} for the real case. For $z\in \Reals$ and $r>0$, we use $I(z,r)$ to denote the interval of length $2r$ centered at $z$.
\begin{lemma}\label{lem:precisionreals}
Let $z_0\in \mathbb{Q}$ and $r>0$ be a rational. Further, suppose that   $\lambda_1',\hdots,\lambda_t'\in \mathbb{Q}$ are such that the real maps $\Phi_i:z\mapsto \frac{1}{1+\lambda_i' z}$ with $i\in[t]$ satisfy the following:
\begin{enumerate}
\item for each $i\in[t]$, $\Phi_i$ is contracting on  the interval $I(z_0,r)$,
\item $I(z_0,r)\subseteq \bigcup^{t}_{i=1}\Phi_i(I(z_0,r))$.
\end{enumerate}

There is an algorithm which, on input (i) a starting point $x_0\in I(z_0,r)\cap \mathbb{Q}$, (ii) a target $x\in I(z_0,r)\cap \mathbb{Q}$,  and (iii) a rational $\epsilon>0$, outputs in $poly(\size{x_0,x,\eps})$ time a number $\hat{x}\in I(z_0,r)\cap \mathbb{Q}$ and  a sequence $i_1,i_2,\hdots,i_k\in [t]$ such that  
\[\hat{x}=\Phi_{i_k}(\Phi_{i_{k-1}}(\cdots\Phi_{i_1}(x_0)\cdots))\mbox{ and }|\hat{x}-x|\leq \epsilon.\]
\end{lemma}
\begin{proof}[Proof of Lemmas~\ref{lem:precision} and~\ref{lem:precisionreals}]
We focus on proving Lemma~\ref{lem:precision}, the proof of Lemma~\ref{lem:precisionreals} being almost identical; one only needs to replace the ball $B(z_0,r)$ with the interval $I(z_0,r)$ in the following argument.

Since the maps $\Phi_i$ are contracting on the ball $B(z_0,r)$ for all $i\in [t]$, there exists a  real number $M<1$ such that for all $i\in[t]$ and all $x,y\in B(z_0,r)$, it  holds that
\begin{equation}\label{eq:contraction123}
|\Phi_i(x)-\Phi_i(y)|\leq M|x-y|.
\end{equation}
W.l.o.g., we assume that $M\in \mathbb{Q}$.

Let $x_0,x\in B(z_0,r)\cap \CQ$ and $\epsilon\in (0,1)$. We are now going to describe a procedure that produces a point that is at distance at most $\eps$ from $x$ in time polynomial in $\size{x_0,x,\eps}$.

\phantom{AA} GetPoint$(x,\epsilon)$\\
\phantom{AAAA} if $\eps\geq |x-x_0|$, then return $x_0$\\
\phantom{AAAA} else\\
\phantom{AAAAAA} let $i\in[t]$ be such that $x\in \Phi_i(B(z_0,r))$\\
\phantom{AAAAAA} $y\leftarrow$ GetPoint$(\Phi_i^{-1}(x), \epsilon/M)$ \\
\phantom{AAAAAA} return $\hat{x}=\Phi_i(y)$ 

\vskip 0.2cm\noindent
Note that, in each recursive call of the procedure GetPoint$(\cdot,\cdot)$, the second parameter increases by a factor $1/M$ and hence the number $N$ of recursive calls is bounded by $1+\log_{1/M} (2r/\eps)=O(\size{\eps})$. Moreover, we can find $i\in[t]$ such that $x\in \Phi_i(B(z_0,r))$ in time polynomial in $\size{x,\eps}$ (since we can compute $\Phi_i^{-1}(x)$ for each $i\in [t]$ and check whether $|\Phi_i^{-1}(x)-z_0|\leq r$). Finally, note that, in each recursive call of the procedure GetPoint$(\cdot,\cdot)$, the first parameter  is always from $B(z_0,r)$ since $x\in \Phi_i(B(z_0,r))$.

The correctness of the algorithm is proved by induction on the number $N$ of recursive calls. In the base case $N=0$, we have $\eps\geq |x-x_0|$ and the procedure returns $x_0$ which is at distance at most $\eps$ from $x$. For the inductive step
we have $|y - \Phi_i^{-1}(x)|\leq \epsilon/M$ and hence by \eqref{eq:contraction123}
\[|\Phi_i(y) - x| \leq M|y-\Phi_i^{-1}(x)|\leq \epsilon.\]
It remains to observe that we can modify the procedure GetPoint$(x,\epsilon)$ so that it also returns the desired sequence $i_1,i_2,\hdots,i_k$.
\end{proof}

We will pair our applications of Lemmas~\ref{lem:precision} and~\ref{lem:precisionreals}  with the following path construction.
\begin{lemma}\label{lem:pathpath}
Fix $\lambda \in \CQ$.
Let $x_0,\lambda_1',\hdots,\lambda_t'\in \CQ\backslash\{0\}$ and, for $i\in[t]$, consider the maps $\Phi_i(z)=\tfrac{1}{1+\lambda_i'z}$ for $z\neq -\tfrac{1}{\lambda_i'}$. Let $i_1,i_2,\hdots,i_k\in [t]$ be a sequence such that $\Phi_{i_k}(\Phi_{i_{k-1}}(\cdots\Phi_{i_1}(x_0)\cdots))=\hat{x}$ for some $\hat{x}\in \Complex$.

Let $P$ be a path of length $k+1$ whose vertices are labelled as $v_0,v_1,\hdots,v_k, v_{k+1}$. Let $\lambdab$ be the activity vector on $P$ given by
\[\lambda_{v_0}=(1-x_0)/x_0, \quad \lambda_{v_j}=\lambda_{i_j}' \mbox{ for $j\in[k]$}, \quad \lambda_{v_{k+1}}=\lambda.\]
Then, it holds that $\Zout_{P,v_{k+1}}(\lambdab)\neq 0$ and $\frac{\Zin_{P,v_{k+1}}(\lambdab)}{\Zout_{P,v_{k+1}}(\lambdab)}=\lambda\hat{x}$. Moreover, there is an algorithm that, on input $x_0,\lambda_1',\hdots,\lambda_t',k$, computes the quantities $\Zin_{P,v_{k+1}}(\lambdab)$ and $\Zout_{P,v_{k+1}}(\lambdab)$ in time polynomial in $k$ and $\size{x_0,\lambda_1',\hdots,\lambda_t'}$.
\end{lemma}
\begin{proof}
For $j=0,1,\hdots, k$, denote by $P_j$ the path of length $j$ induced by the vertices $v_0,v_1,\hdots, v_j$; the activity of a vertex $v_i$ in $P_j$ is equal to the activity of the vertex $v_i$ in $P$.  To simplify notation, we will drop the activity vector $\lambdab$ from notation, i.e., we will write $\Zin_{P_j,v_j}, \Zout_{P_j,v_j}, Z_{P_j}$ instead of $\Zin_{P_j,v_j}(\lambdab), \Zout_{P_j,v_j}(\lambdab), Z_{P_j}(\lambdab)$.

We will show by induction on $j$ that
\begin{equation}\label{eq:g4vc3c}
\frac{\Zout_{P_j,v_{j}}}{Z_{P_j}}=x_j, \mbox{ where } x_j=\Phi_{i_j}(\Phi_{i_{j-1}}(\cdots\Phi_{i_1}(x_0)\cdots)).
\end{equation}
For $j=0$, we  have 
\begin{equation}\label{eq:evrttvve}
\Zout_{P_0,v_{0}}=1,\quad Z_{P_0}=1+\lambda_0=1/x_0,
\end{equation}
and therefore \eqref{eq:g4vc3c} holds. Assume that \eqref{eq:g4vc3c} holds for some $j$ in $0,\hdots, k-1$, we show that it holds for $j+1$ as well. Note that $x_{j+1}=\Phi_{i_{j+1}}(x_j)$ and therefore $x_j\neq -1/\lambda_{v_{j+1}}'$, i.e.,  $\Zout_{P_j,v_{j}}/Z_{P_j}\neq -1/\lambda_{v_{j+1}}'$. We have 
\begin{equation}\label{eq:evrttvve1}
\begin{aligned}
&\Zout_{P_{j+1},v_{j+1}}=Z_{P_j},\\ 
&Z_{P_{j+1}}=\Zin_{P_{j+1},v_{j+1}}+\Zout_{P_{j+1},v_{j+1}}=\lambda_{v_{j+1}}\Zout_{P_{j},v_{j}}+Z_{P_j}\neq 0
\end{aligned}
\end{equation}
and therefore
\begin{equation*}
\frac{\Zout_{P_{j+1},v_{j+1}}}{Z_{P_{j+1}}}=\frac{1}{1+\lambda_{v_{j+1}}\frac{\Zout_{P_{j},v_{j}}}{Z_{P_j}}}=\Phi_{i_{j+1}}(x_j)=x_{j+1}.
\end{equation*}
This finishes the proof of \eqref{eq:g4vc3c} for all $j=0,1,\hdots, k$. To conclude the proof, note that 
\begin{equation}\label{eq:evrttvve2}
\begin{aligned}
\Zin_{P,v_{k+1}}&=\lambda_{v_{k+1}} \Zout_{P_k,v_k}=\lambda \Zout_{P_k,v_k},\\ 
\Zout_{P,v_{k+1}}&=Z_{P_k}\neq 0,
\end{aligned}
\end{equation} 
so that $\frac{\Zin_{P,v_{k+1}}}{\Zout_{P,v_{k+1}}}=\lambda\hat{x}$. Finally, note that using \eqref{eq:evrttvve}, \eqref{eq:evrttvve1}, \eqref{eq:evrttvve2}, we can also compute $\Zin_{P,v_{k+1}}(\lambdab)$ and $\Zout_{P,v_{k+1}}(\lambdab)$ in time polynomial in $k$ and $\size{x_0,\lambda_1',\hdots,\lambda_t'}$. This completes the proof of Lemma~\ref{lem:pathpath}.
\end{proof}

\subsection{Proof of Proposition~\ref{lem:realexp} (real case)}\label{sec:realexp}

To prove Proposition~\ref{lem:realexp}, we will bootstrap Lemma~\ref{lem:realeps} to obtain implementations that control logarithmically the size of $G$ in terms of the desired accuracy $\epsilon$. The following technical lemma will allow us to use Lemma \ref{lem:precisionreals}.
\begin{lemma}\label{lem:numerics}
Let $\lambda_1^*=-1/4, \lambda_2^*=-6/25$, and $I$ be the interval $[7/4,11/6]$. Let $\eta=10^{-10}$, then the following holds for all $\lambda_1'\in [\lambda_1^*-\eta,\lambda_1^*+\eta]$ and $\lambda_2'\in [\lambda_2^*-\eta,\lambda_2^*+\eta]$.

The maps $\Phi_1(x)=\frac{1}{1+\lambda_1'x}$ and $\Phi_2(x)=\frac{1}{1+\lambda_2'x}$ are contracting on the interval $I$ and, moreover, $I\subseteq \Phi_1(I)\cup \Phi_2(I)$.
\end{lemma}
\begin{proof}
For all $\lambda_1'\in [\lambda_1^*-\eta,\lambda_1^*+\eta]$ and $\lambda_2'\in [\lambda_2^*-\eta,\lambda_2^*+\eta]$, we have that $\Phi_1,\Phi_2$ are increasing, while $\Phi_1',\Phi_2'$ are decreasing. Therefore
\begin{gather*}
\Phi_1(I)=[\Phi_1(7/4),\Phi_1(11/6)],\quad  \Phi_2(I)=[\Phi_2(7/4),\Phi_2(11/6)]\\
\Phi_1'(I)=[\Phi_1'(11/6),\Phi_1'(7/4)],\quad  \Phi_2'(I)=[\Phi_2'(11/6),\Phi_2'(7/4)]
\end{gather*}
By calculating the relevant function values, we obtain that
\[[1.78,1.84]\subseteq \Phi_1(I),\quad [1.73,1.78] \subseteq \Phi_2(I), \quad \max_{x\in I}\big\{|\Phi_1'(x)|,|\Phi_2'(x)|\big\}\leq 9/10,\]
and the lemma follows.
\end{proof}

We are now ready to give the proof of Proposition~\ref{lem:realexp}, which we restate here for convenience.
\begin{lemrealexp}
\statelemrealexp
\end{lemrealexp}
\begin{proof}[Proof of Proposition~\ref{lem:realexp}]
Let $I=[7/4,11/6]$ be the interval of Lemma~\ref{lem:numerics}. The main idea of the proof is to use the maps in Lemma~\ref{lem:numerics} in combination with Lemma~\ref{lem:precisionreals} to get a subinterval of $I$ where we can get exponentially accurate implementations. Then, we will propagate this exponential accuracy to the whole real line by using implementations from Lemma~\ref{lem:realeps}.

We next specify the activities that we will need to implement with constant precision via Lemma~\ref{lem:numerics} (later, these activities will be combined with Lemma~\ref{lem:precisionreals} to get the exponential precision). Let $\lambda_1^*,\lambda_2^*,\eta, I$ be as in Lemma~\ref{lem:numerics}. Moreover, let $x_0^*$ be the midpoint of the interval $I$ and let $\lambda_0^*=(1-x_0^*)/x_0^*$, $\lambda_3^*=-1/x_0^*$. Finally, let $\lambda_4^*=\max\{10^{12},10^{12}|\lambda|\}$. (Note that all of these are rationals.)

By Lemma~\ref{lem:realeps}, for $i=0,1,\hdots,4$, there exists a bipartite graph $G_i$ of maximum degree $\Delta$ which implements $\lambda_i^*$ with accuracy $\eta$. Let $v_i$ be the terminal of $G_i$ and set $\lambda_i'=\frac{\Zin_{G_i,v_i}(\lambda)}{\Zout_{G_i,v_i}(\lambda)}$; we have that $\lambda_i'\in [\lambda_i^*-\eta,\lambda_i^*+\eta]$.  Moreover, let $x_0=1/(1+\lambda_0')$ and note that $|x_0-x_0^*|\leq 10^{-5}$, so that $x_0\in I$. Also, let $x_3=-1/\lambda_3'$ and note that $|x_3-x_0^*|\leq 10^{-5}$, so that the interval $I^*=[x_3-10^{-3},x_3+10^{-3}]$ is a subinterval of  $I$.  Finally, we have that 
\begin{equation}\label{eq:y0choice0}
|\lambda_4'|\geq |\lambda_4^*|-\eta\geq |\lambda_4^*|-|\lambda_4^*|/2\geq \max\{10^{11},10^{11}|\lambda|\}.
\end{equation}

\vskip 0.2cm

Suppose that we are given inputs $\lambda',\eps\in \QQ$ with $\eps>0$  and we want to output in $poly(\size{\lambda',\eps})$ time a bipartite graph of maximum degree $\Delta$ that implements $\lambda'$ with accuracy $\epsilon$. Clearly, we may assume that $\eps\in (0,1)$. The algorithm has three cases depending on the value of $|\lambda'|$.
\vskip 0.2cm
\noindent \textbf{Case I (large $\mathbf{|\lambda'|})$: $|\lambda'|\geq \max\{10^4|\lambda|,1\}$.}  Let $x^*$ be a rational such that
\begin{equation}\label{eq:3fgg56yy35tg}
\frac{\lambda}{1+\lambda_3' x^*}=\lambda', \mbox{ so that } x^*=\frac{1}{\lambda_3'}\Big(\frac{\lambda}{\lambda'}-1\Big).
\end{equation}
Recall that $x_3=-1/\lambda_3'$, so using the assumption $|\lambda'|\geq 10^4|\lambda|$ and that $|\lambda_3'|\geq |\lambda_3^*|-\eta\geq 1/10$, we have that
\[|x^*-x_3|=\frac{|\lambda|}{|\lambda'\lambda_3'|}\leq 10^{-3},\]
It follows that $x^*$ belongs to the interval  $I$.

  Note that Lemma~\ref{lem:numerics} guarantees that the maps $\Phi_1(x)=\frac{1}{1+\lambda_1'x}$ and $\Phi_2(x)=\frac{1}{1+\lambda_2'x}$ satisfy the hypotheses of Lemma~\ref{lem:precisionreals} with $z_0=x_0^*$ and $r$ equal to half the length of the interval $I$. Therefore, using the algorithm of Lemma~\ref{lem:precisionreals}, on input $x_0, x^*$ and $\epsilon'=\epsilon \cdot \min\big\{\frac{|\lambda|}{2|\lambda'\lambda_3'|},\frac{|\lambda|}{2|(\lambda')^2\lambda_3'|},1\big\}$, we obtain in time $poly(\size{x_0,x^*,\eps'})=poly(\size{\lambda',\eps})$ a number $\hat{x}$ and a sequence $i_1,\hdots,i_k\in\{1,2\}$ such that 
\begin{equation}\label{eq:vrvwc43rf2}
\hat{x}=\Phi_{i_k}(\Phi_{i_{k-1}}(\cdots\Phi_{i_1}(x_0)\cdots))\mbox{ and }|\hat{x}-x^*|\leq \epsilon'.
\end{equation}
Using \eqref{eq:3fgg56yy35tg} and \eqref{eq:vrvwc43rf2}, we have by the triangle inequality that 
\begin{equation}\label{eq:erf3f5f}
|1+\lambda_3' \hat{x}|\geq |1+\lambda_3' x^*|-|\lambda_3' (\hat{x}-x^*)|\geq \frac{|\lambda|}{|\lambda'|}-|\lambda_3' \epsilon'|\geq \frac{|\lambda|}{2|\lambda'|} ,
\end{equation} 
and therefore
\begin{equation}\label{eq:4bhb5bb4rfv}
\begin{aligned}
\Big|\frac{\lambda}{1+\lambda_3' \hat{x}}-\lambda'\Big|&=\Big|\frac{\lambda}{1+\lambda_3' \hat{x}}-\frac{\lambda}{1+\lambda_3' x^*}\Big|=  \frac{|\lambda|\,|\lambda_3'|\, |\hat{x}-x^*|}{|1+\lambda_3'x^*|\cdot |1+\lambda_3'\hat{x}|}\\
&=\frac{|\lambda'\lambda_3'|\, |\hat{x}-x^*|}{|1+\lambda_3'\hat{x}|}\leq \frac{2|(\lambda')^2\lambda_3'|}{|\lambda|}\, |\hat{x}-x^*|\leq \epsilon,
\end{aligned}
\end{equation}
where in the last equality we used \eqref{eq:3fgg56yy35tg}, in the second to last inequality we used \eqref{eq:erf3f5f}, and in the last inequality we used \eqref{eq:vrvwc43rf2} and the choice of $\epsilon'$.

Now,  let $P$ be a path of length $k+2$ whose vertices are labelled as $v_0,v_1,\hdots,v_{k+1}, v_{k+2}$. Let $\lambdab$ be the activity vector on $P$ given by
\[\lambda_{v_0}=\lambda_0', \quad \lambda_{v_j}=\lambda_{i_j}' \mbox{ for $j\in[k]$}, \quad \lambda_{v_{k+1}}=\lambda_3', \quad \lambda_{v_{k+2}}=\lambda .\]
Then, by Lemma~\ref{lem:pathpath}, it holds that $\Zout_{P,v_{k+2}}(\lambdab)\neq 0$ and $\frac{\Zin_{P,v_{k+2}}(\lambdab)}{\Zout_{P,v_{k+2}}(\lambdab)}=\frac{\lambda}{1+\lambda_3' \hat{x}}$; moreover, we can also compute the values $\Zin_{P,v_{k+2}}(\lambdab),\Zout_{P,v_{k+2}}(\lambdab)$ in time $poly(k, \size{\lambda_0',\lambda_1',\lambda_2'})=poly(\size{\lambda',\eps})$. Since the bipartite graphs $G_0, G_1, G_2,G_3$ implement the activities $\lambda_0',\lambda_1',\lambda_2',\lambda_3'$, respectively, we obtain by applying Lemma~\ref{lem:transf} (to the path $P$ with activity vector $\lambdab$) a bipartite graph $G'$ with maximum degree $\Delta$ and terminal $v_{k+2}$ such that 
\begin{equation}\label{eq:f3f43fg454}
\Zin_{G',v_{k+2}}(\lambda)= C\cdot \Zin_{P,v_{k+2}}(\lambdab),\qquad \Zout_{G',v_{k+2}}(\lambda)=C\cdot \Zout_{P,v_{k+2}}(\lambdab),
\end{equation}
where $C=\prod^3_{i=1}\big(\Zout_{G_i,v_i}(\lambda)\big)^{|\{j\in\{0,\hdots, k+2\}\mid \lambda_{v_j}=\lambda_i'\}}$. We conclude that
\[\frac{\Zin_{G',v_{k+2}}(\lambda)}{\Zout_{G',v_{k+2}}(\lambda)}=\frac{\Zin_{P,v_{k+2}}(\lambdab)}{\Zout_{P,v_{k+2}}(\lambdab)}=\frac{\lambda}{1+\lambda_3' x^*}.\]
Combining this with \eqref{eq:4bhb5bb4rfv}, we obtain that $G'$ with terminal $v_{k+2}$ is a bipartite graph of maximum degree $\Delta$ which implements $\lambda'$ with accuracy $\epsilon$. Moreover, using \eqref{eq:f3f43fg454}, we can also compute the values $\Zin_{G',v_{k+2}}(\lambda),\Zout_{G',v_{k+2}}(\lambda)$.

\vskip 0.2cm

\noindent \textbf{Case II (small $\mathbf{|\lambda'|}$): $|\lambda'|\leq \min\{10^{-5},10^{-5}|\lambda|\}$}. We first assume that $\lambda'\neq 0$. 

Let $\hat{\lambda}$ be such that
\begin{equation}\label{eq:ffg55gf3}
\frac{\lambda}{1+\hat{\lambda}}=\lambda', \mbox{ so that } \hat{\lambda}=\frac{\lambda}{\lambda'}-1.
\end{equation}
Using the assumption $|\lambda'|\leq \min\{10^{-5},10^{-5}|\lambda|\}$ and the triangle inequality, we have that
\[|\hat{\lambda}|\geq \frac{|\lambda|}{|\lambda'|}-1\geq \frac{|\lambda|}{2|\lambda'|}\geq \max\{10^{4}|\lambda|,1\}.\] 
Therefore, by Case I, we can construct a bipartite graph $G$ with terminal $v$ that  implements $\hat{\lambda}$ with accuracy $\epsilon'=\epsilon \cdot \min\big\{\frac{|\lambda|}{2|\lambda'|},\frac{|\lambda|}{4|\lambda'|^2},1\big\}$ in time $poly(\size{\hat{\lambda},\epsilon'})=poly(\size{\lambda',\epsilon})$. Moreover, we can compute the values $\Zin_{G,v}(\lambda)$ and $\Zout_{G,v}(\lambda)$ at the same time.  Let $\lambda''= \frac{\Zin_{G,v}(\lambda)}{\Zout_{G,v}(\lambda)}$, so that $|\lambda''-\hat{\lambda}|\leq \epsilon'$. Using this and \eqref{eq:ffg55gf3}, we have
\begin{equation}\label{eq:everevev52514}
|1+\lambda'' |\geq  |1+\hat{\lambda}|-|(\lambda''-\hat{\lambda})|\geq \frac{|\lambda|}{|\lambda'|}-\epsilon'\geq \frac{|\lambda|}{2|\lambda'|},
\end{equation}
and therefore 
\begin{equation}\label{eq:g45g45}
\begin{aligned}
\Big|\frac{\lambda}{1+\lambda'' }-\lambda'\Big|&=\Big|\frac{\lambda}{1+\lambda'' }-\frac{\lambda}{1+\hat{\lambda}}\Big|=\frac{|\lambda |\,|\lambda''-\hat{\lambda}|}{|1+\hat{\lambda}|\cdot |1+\lambda''|}\\
&=\frac{|\lambda' ||\lambda''-\hat{\lambda}|}{|1+\lambda''|}\leq \frac{2|\lambda'|^2}{|\lambda|}|\lambda''-\hat{\lambda}|\leq \epsilon.
\end{aligned}
\end{equation}
where in the last equality we used \eqref{eq:ffg55gf3}, in the second to last inequality we used \eqref{eq:everevev52514}, and in the last inequality we used $|\lambda''-\hat{\lambda}|\leq \epsilon'$ and the choice of $\epsilon'$.

Now,  let $G'$ be the bipartite graph obtained from $G$ by adding a new vertex $u$ whose single neighbour is the terminal $v$ of $G$. Then, we have that
\begin{equation}\label{eq:f3f43vrever}
\Zin_{G',u}(\lambda)= \lambda\Zout_{G,v}(\lambda),\qquad \Zout_{G',u}(\lambda)=Z_G(\lambda)=\Zin_{G,v}(\lambda)+\Zout_{G,v}(\lambda).
\end{equation}
 We conclude that
\[\frac{\Zin_{G',u}(\lambda)}{\Zout_{G',u}(\lambda)}=\frac{\lambda\Zout_{G,v}(\lambda)}{\Zin_{G,v}(\lambda)+\Zout_{G,v}(\lambda)}=\frac{\lambda}{1+\lambda''}.\]
Combining this with \eqref{eq:g45g45}, we obtain that $G'$ with terminal $u$ is a bipartite graph of maximum degree $\Delta$ which  implements $\lambda'$ with accuracy $\epsilon$. Moreover, using \eqref{eq:f3f43vrever}, we can also compute the values $\Zin_{G',u}(\lambda),\Zout_{G',u}(\lambda)$.

To finish this case, it remains to argue for $\lambda'=0$. Then, for $\epsilon''=\min\{\epsilon,10^{-5},10^{-5}|\lambda|\}$, we can use the preceding method to  implement the activity $\epsilon''/2\neq 0$ with accuracy $\epsilon''/2$ in time $poly(\size{\epsilon''})=poly(\size{\epsilon})$. The implemented activity $\lambda''$ satisfies by the triangle inequality $|\lambda''|\leq \epsilon$; hence, we have implemented the desired activity $\lambda'=0$ with accuracy $\epsilon$.  

\vskip 0.2cm

\noindent \textbf{Case III (moderate $\mathbf{|\lambda'|}$): $\min\{10^{-5},10^{-5}|\lambda|\}<|\lambda'|<\max\{10^{4}|\lambda|,1\}$}. Let $x^*$ be a rational such that 
\begin{equation}\label{eq:y1starbounds}
\frac{\lambda}{1+\lambda_4'  x^*}=\lambda',\mbox{ so that }x^*=\frac{1}{\lambda_4'}\Big(\frac{\lambda}{\lambda'}-1\Big).
\end{equation}
Using the assumption $10^{-5}|\lambda|<|\lambda'|$ and  $|\lambda_4'|\geq \max\{10^{11},10^{11}|\lambda|\}$ (cf. \eqref{eq:y0choice0}), we have that $|x^*|\leq \min\{10^{-5}/|\lambda|, 10^{-5}\}$. Let $\epsilon'=\epsilon\cdot \min\big\{\frac{1}{10^{10}|\lambda_4'|},\frac{|\lambda|}{10^{10}|\lambda_4'|}\big\}$.  Then, by the algorithm for Case II, we can implement the activity $\lambda x^*$ with precision $\epsilon'$ in time $poly(\size{\lambda x^*,\epsilon'})=poly(\size{\lambda',\epsilon})$. That is, we can construct a bipartite graph $G$ of maximum degree at most $\Delta$ with terminal $v$ such that, for  $\lambda'':=\Zin_{G,v}(\lambda)/\Zout_{G,v}(\lambda)$, it holds that
\begin{equation}\label{eq:y1bounds1244235}
\big|\lambda''- \lambda x^*\big|\leq \epsilon'.
\end{equation}
Now, using \eqref{eq:y1starbounds} and \eqref{eq:y1bounds1244235}, we have by the triangle inequality that
\begin{equation}\label{eq:4r45v4vv2c23}
|1+ \lambda_4'(\lambda''/\lambda)|\geq |1+ \lambda_4' x^*|-|\lambda_4' (x^*-\lambda''/\lambda)|\geq \frac{1}{10^4}- \frac{1}{10^{10}}\geq \frac{1}{10^5},
\end{equation}
and therefore
\begin{equation}\label{eq:v4v4v35g3f2ftgb}
\begin{aligned}
\Big|\frac{\lambda}{1+ \lambda_4'(\lambda''/\lambda)}-\lambda'\Big|&=\Big|\frac{\lambda}{1+ \lambda_4'(\lambda''/\lambda)}-\frac{\lambda}{1+ \lambda_4'x^*}\Big|=\frac{|\lambda|\, |\lambda'_4|\, |x^*-(\lambda''/\lambda)|}{|1+ \lambda_4'x^*|\cdot |1+ \lambda_4'(\lambda''/\lambda)|}\\
&= \frac{|\lambda'|\, |\lambda'_4|\, |x^*-(\lambda''/\lambda)|}{|1+ \lambda_4'(\lambda''/\lambda)|}\leq 10^9 |\lambda'_4| |\lambda x^*-\lambda''|\leq \epsilon.
\end{aligned}
\end{equation}
where in the last equality we used \eqref{eq:y1starbounds}, in the second to last inequality we used \eqref{eq:4r45v4vv2c23} and $|\lambda'|/|\lambda|\leq 10^4$, and in the last inequality we used \eqref{eq:y1bounds1244235} and the choice of $\epsilon'$.

Recall that $G_4$ is a bipartite graph of maximum degree $\Delta$, with terminal $v_4$, which implements the activity $\lambda_4'$. Let $G'$ be the bipartite graph obtained by taking  a copy of $G_4$ and $G$ and identifying the terminals $v_4,v$ into a single vertex which we label $u'$ (note that $G'$  has maximum degree $\Delta$ as well since $v_4$ and $v$ have degree one in $G_4$ and $G$, respectively). Then, 
\begin{equation}\label{eq:f3f3553d4t}
\Zin_{G',u'}(\lambda)=\frac{1}{\lambda}\Zin_{G_4,v_4}(\lambda)\Zin_{G,u}(\lambda), \quad \Zout_{G',u'}(\lambda)= \Zout_{G_4,v_4}(\lambda)\Zout_{G,v}(\lambda),\\
\end{equation}
Consider the graph $G''$ obtained from $G'$ by adding a new vertex $u''$ whose
single neighbour is the vertex $u'$. Then, we have that
\begin{equation}\label{eq:f3f43vrever423}
\Zin_{G'',u''}(\lambda)= \lambda\Zout_{G',u'}(\lambda),\qquad \Zout_{G'',u''}(\lambda)=Z_{G'}(\lambda)=\Zin_{G',u'}(\lambda)+\Zout_{G',u'}(\lambda).
\end{equation}
Using this in conjuction with \eqref{eq:f3f3553d4t}, we conclude that
\[\frac{\Zin_{G'',u''}(\lambda)}{\Zout_{G'',u''}(\lambda)}=\frac{\lambda\Zout_{G',u'}(\lambda)}{\Zin_{G',u'}(\lambda)+\Zout_{G',u'}(\lambda)}=\frac{\lambda}{1+\frac{\Zin_{G',u}(\lambda)}{\Zout_{G',u}(\lambda)}}=\frac{\lambda}{1+ \lambda'_4(\lambda''/\lambda)}.\]
Combining this with \eqref{eq:v4v4v35g3f2ftgb}, we obtain that $G''$ is a bipartite graph of maximum degree at most $\Delta$ with terminal $u''$ which implements $\lambda'$ with accuracy $\epsilon$. Moreover, using \eqref{eq:f3f3553d4t} and \eqref{eq:f3f43vrever423}, we can also compute the values $\Zin_{G'',u''}(\lambda),\Zout_{G'',u''}(\lambda)$.

\vskip 0.2cm

This completes the three different cases of the algorithm, thus completing the proof of Proposition~\ref{lem:realexp}.
\end{proof}

\subsection{Proof of Proposition~\ref{lem:complexexp}  (complex case)}\label{sec:complexexp}

In this section, we prove Proposition~\ref{lem:complexexp} assuming  Proposition~\ref{lem:main121} (the proof of the latter is given in Section~\ref{sec:mainpart}).
Note that  Proposition~\ref{lem:main121} applies for all $\lambda\in \CQ\setminus\Reals$. By restricting our attention to $\lambda\notin \LambdaD\cup \Reals$ and using the theory of Section~\ref{sec:iter}, we obtain the following.
\begin{lemma}\label{lem:tranpole}
Let $\Delta\geq 3$ and $\lambda\in \CQ\setminus\Reals$ be such that $\lambda\notin \LambdaD$, and set $d:=\Delta-1$. Let $\omega$ be the fixpoint of $f(z)=\frac{1}{1+\lambda z^{d}}$ with the smallest norm, and $p_1,\hdots, p_{d}$ be the poles of $f$. Then, for any  real number $\eta>0$, there exist:
\[\mbox{(i) an integer $N\geq 1$, (ii)  a pole $p^*\in\{p_1,\hdots, p_d\}$, (iii) rationals $L>0$ and $r,r',r^*\in(0,\eta)$ with $r'<r$}\] 
such that all of the following hold:
\begin{enumerate}
\item \label{it:ball} $B(p^*,r^*)\subseteq f^{N}(B(\omega,r'))$, 
\item \label{it:rfv} $p_1,\hdots,p_d\notin \bigcup^{N-1}_{n=0}\fn{f}{n}\big(B(\omega,r)\big)$, and 
\item \label{it:lip} For all $x_1,x_2\in B(\omega,r)$, it holds that $|f^{N}(x_1)-f^{N}(x_2)|\leq L|x_1-x_2|$.
\end{enumerate}
\end{lemma}
\begin{proof}
Consider an arbitrary $\eta>0$. Note that $\omega\neq p_1,\hdots,p_d$, so there is no loss of generality in assuming that $\eta\leq |\omega-p_1|, \hdots, |\omega-p_d|$.  Let $U'$ be the open ball $B(\omega,\eta/10)$ and note that our assumption on $\eta$ ensures that $p_1,\hdots,p_d\notin U'$.

Since $\lambda\notin \LambdaD\cup \Reals$, we have by Lemma~\ref{lem:fixpoints} that $\omega$ is a repelling fixpoint of $f$ and therefore $\omega$ belongs to the Julia set of $f$ (by Lemma~\ref{lem:repelling}). Moreover, using Lemma~\ref{lem:except}, we have that the exceptional set of $f$ is empty.\footnote{To see this, let $E_f$ denote the exceptional set of $f$. By Lemma~\ref{lem:except}, we have that a necessary condition for a point $x\in \Complex$ to be in $E_f$ is that either $f'(x)=0$ or $f'(f(x))f'(x)=0$, which gives $x=0$ as the only possible point.  However, since $f(0)=1$, we have that $x=0$ cannot be a fixpoint of either $f$ or $\fn{f}{2}$ and therefore, by Lemma~\ref{lem:except}, there is no point $x\in \Complex$ that belongs to $E_f$. Similarly, we have that $\infty\notin E_f$ since $x=\infty$ is not a fixpoint of either $f$ or $\fn{f}{2}$ (by $f(\infty)=0$ and $f(0)=1$), proving that $E_f$ is empty.} Therefore, by Theorem~\ref{thm:bazooka}, it holds that $\bigcup^{\infty}_{n=0}\fn{f}{n}(U')=\Riem$.   Let $N'$ be the smallest  integer such that one of the poles $p_1,\hdots,p_d$ belongs to $\fn{f}{N'}(U')$,  i.e., $N'$ satisfies
\begin{equation}\label{eq:Nchoice13}
\{p_1,\hdots,p_d\}\cap \Big(\bigcup^{N'-1}_{n=0}\fn{f}{n}(U')\Big)=\emptyset\mbox{ and } \{p_1,\hdots,p_d\}\cap \fn{f}{N'}(U')\neq \emptyset,
\end{equation}
Note that $N'\geq 1$ since $p_1,\hdots,p_d\not\in U'$. 

To prove the lemma, it will be important for us to ensure that the pole of $f$ which belongs to $\fn{f}{N'}(U')$ does not sit on the boundary of any of the sets $\fn{f}{0}(U'),\fn{f}{1}(U'),\hdots,\fn{f}{N'}(U')$. To achieve this, we will enlarge a little bit the ball $U'$ as follows. Let $\mathcal{P}$ be the union of the poles of the functions $\fn{f}{1},\hdots,\fn{f}{N'}$. Note that $\mathcal{P}$ is a finite set, therefore we can specify a radius $r\in(\eta/10,\eta)$ so that the boundary $\partial U$ of the open ball $U=B(\omega,r)$ is disjoint from $\mathcal{P}$ (i.e., $\partial U\cap \mathcal{P}=\emptyset$). Since $U'\subseteq U$ and $p_1,\hdots,p_d\not\in U$, we conclude from \eqref{eq:Nchoice13} that there exists a positive integer $N\leq N'$ such that 
\begin{equation}\label{eq:Nchoig54g5}
\{p_1,\hdots,p_d\}\cap \Big(\bigcup^{N-1}_{n=0}\fn{f}{n}(U)\Big)=\emptyset\mbox{ and } \{p_1,\hdots,p_d\}\cap \fn{f}{N}(U)\neq \emptyset,
\end{equation}
i.e., $N\leq N'$ is the first integer such that a pole of $f$ belongs to $\fn{f}{N}(U)$. Let $p^*\in \{p_1,\hdots,p_d\}$ be an  arbitrary pole of $f$ such that $p^*\in \fn{f}{N}(U)$. We claim that 
\begin{equation}\label{eq:4tv4b445}
\mbox{for all $n\in \{0,\hdots, N\}$, $p^*$ does not lie on the boundary of $\fn{f}{n}(U)$.}
\end{equation} 
Indeed, observe that $U$ is open and $\fn{f}{n}$ is holomorphic on $U$ for all $n=0,\hdots, N$ since $\fn{f}{n-1}(U)$ does not contain any pole of $f$. Therefore, by the open mapping theorem, we have that 
\begin{equation}\label{eq:3rf554354g5}
\mbox{$\fn{f}{n}(U)$ is an open set  for all $n=0,\hdots,N$}.
\end{equation}
Since $p^*\in \fn{f}{N}(U)$, this already shows that $p^*$ does not lie on the boundary of $\fn{f}{N}(U)$. For $n=0,\hdots,N-1$, we obtain from \eqref{eq:3rf554354g5}  and the open mapping theorem that $p^*$ lies on the boundary of  $\fn{f}{n}(U)$ only if $p^*$ lies on the boundary of $\fn{f}{n}(\partial U)$, i.e., $\fn{f}{n}(\partial U)$ contains a pole of $f$ and so $\partial U$ contains a pole of $\fn{f}{n+1}$. 
In turn, this would imply that $\partial U\cap \mathcal{P}\neq \emptyset$, which is excluded by the choice of the radius $r$ of $U$. This proves \eqref{eq:4tv4b445}.

We are now ready to prove Items~\ref{it:ball} and~\ref{it:rfv} of the lemma. Namely, from \eqref{eq:Nchoig54g5}, we have that  $p_1,\hdots,p_d\notin \bigcup^{N-1}_{n=0}\fn{f}{n}(B(\omega,r))$, which proves Item~\ref{it:rfv}. For Item~\ref{it:ball}, note that from $p^*\in \fn{f}{N}(U)$ we obtain that there exists $x^*\in U=B(\omega,r)$ such that $\fn{f}{N}(x^*)=p^*$. Since $U$ is an open ball, let $r'$ be a rational such that $|x^*-\omega|<r'<r$ and consider the open ball $B(\omega,r')$. Then, we have that  $p^*\in \fn{f}{N}(B(\omega,r'))$ (since $x^*\in B(\omega,r')$). We also have by the open mapping theorem that $\fn{f}{N}(B(\omega,r'))$ is open, therefore there exists rational $r^*\in (0,\eta)$ such that $B(p^*,r^*)\subseteq \fn{f}{N}(B(\omega,r'))$, thus proving Item~\ref{it:ball}.

It remains to prove Item~\ref{it:lip}, which essentially follows from the fact $\bigcup^{N-1}_{n=0}\fn{f}{n}(U)$ does not contain any poles of $f$ and (crude) Lipschitz arguments. In particular, note that, by \eqref{eq:Nchoig54g5} and \eqref{eq:3rf554354g5}, there exists a rational $\delta>0$ such that 
\begin{equation}\label{eq:lowpole12}
|x-p_1|,\hdots, |x-p_d|\geq \delta \mbox{ for all } x\in\bigcup^{N-1}_{n=0}\fn{f}{n}(U).
\end{equation}
Let $L_0=1$ and define $L_{n+1}=L_n\frac{d (|\omega|+L_n r)^{d-1}}{|\lambda| \delta^{2d}}$ for $n=1,\hdots,N-1$. For $n=0,\hdots,N$, we will show by induction that  
\begin{align}
&\mbox{for all $x_1,x_2\in U$ it holds that $|\fn{f}{n}(x_1)-\fn{f}{n}(x_2)|\leq L_n|x_1-x_2|$},\label{eq:lipschitz111}
\end{align}
which clearly proves Item~\ref{it:lip} by taking $L$ to be any rational $>L_N$. The base case $n=0$ is trivial,  so assume that \eqref{eq:lipschitz111} holds  for some non-negative integer $n\leq N-1$. Then, since $\omega$ is a fixpoint of $f$, we have that $\fn{f}{n}(\omega)=\omega$ and therefore $\fn{f}{n}(U)\subseteq B(\omega,L_n r)$, i.e.,
\begin{equation}\label{eq:vfever12}
|x|\leq |\omega|+L_n r \mbox{ for all } x\in \fn{f}{n}(U).
\end{equation} 
Moreover, since $n\leq N-1$, we have by \eqref{eq:lowpole12}  that $|x-p_1|,\hdots, |x-p_d|\geq \delta$ and hence, by factoring $1+\lambda x^d=\lambda(x-p_1)\cdots(x-p_d)$, we obtain
\begin{equation}\label{eq:vfever13}
\big|1+\lambda x^{d}\big|=|\lambda(x-p_1)\cdots(x-p_d)|\geq |\lambda| \delta^{d}\mbox{ for all } x\in \fn{f}{n}(U).
\end{equation}
Let $x_1,x_2\in U$ and set $z_1=\fn{f}{n}(x_1)$, $z_2=\fn{f}{n}(x_2)$, so that the inductive hypothesis translates into 
\begin{equation}\label{eq:f554b555b}
|z_1-z_2|\leq L_n|x_1-x_2|.
\end{equation}  We then obtain that 
\begin{align*}
|\fn{f}{n+1}(x_1)-\fn{f}{n+1}(x_2)|&=|f(z_1)-f(z_2)|=\frac{\big|\lambda\big(z_1^{d}-z_2^{d}\big)\big|}{\big|\big(1+\lambda z_1^{d}\big)\big(1+\lambda z_2^{d}\big)\big|}\\
&=|z_1-z_2| \cdot \frac{|\lambda|\cdot \big|z_1^{d-1}+z_1^{d-2}z_2+\cdots+z_2^{d-1}\big|}{\big|\big(1+\lambda z_1^{d}\big)\big(1+\lambda z_2^{d}\big)\big|}\\
&\leq L_{n+1} |x_1-x_2|,
\end{align*}
where in the last inequality we used \eqref{eq:vfever12}, \eqref{eq:vfever13}, and \eqref{eq:f554b555b}. This finishes the proof of \eqref{eq:lipschitz111}, and therefore the proof of Item~\ref{it:lip}. This concludes the proof of Lemma~\ref{lem:tranpole}.
\end{proof}

We are now ready to prove Proposition~\ref{lem:complexexp}, which we restate here for convenience.
\begin{lemcomplexexp}
\statelemcomplexexp
\end{lemcomplexexp}
\begin{proof}[Proof of Proposition~\ref{lem:complexexp}]
For convenience, set $d:=\Delta-1$. Let $\omega$ be the fixpoint of $f(x)=\frac{1}{1+\lambda x^{d}}$ with the smallest norm, and $p_1,\hdots, p_{d}$ be the poles of $f$. Let $\rho>0$ be the constant in  Proposition~\ref{lem:main121}. Then, by Lemma~\ref{lem:tranpole} (applied with $\eta=\rho/|\lambda|$), there exist a positive  integer $N$, rationals $L>0$ and $r,r',r^*\in(0,\eta)$ with $r'<r$, and a pole $p^*\in \{p_1,\hdots, p_d\}$ such that  
\begin{gather}
\mbox{$B(p^*,r^*)\subseteq f^{N}(B(\omega,r'))$, \ \ \  $p_1,\hdots,p_d\notin \bigcup^{N-1}_{n=0}\fn{f}{n}(B(\omega,r))$,}\label{eq:f5g53423}\\
\mbox{for all $x_1,x_2\in B(\omega,r)$, it holds that $|f^{N}(x_1)-f^{N}(x_2)|\leq L|x_1-x_2|$.}\label{eq:4b32f2c24f}
\end{gather}
We may assume that $r^*$ is sufficiently small so that, for all poles $p\in \{p_1,\hdots,p_d\}$ which are different than $p^*$ it holds that 
\begin{equation}\label{eq:v4vrccww}
|x-p|\geq \delta \mbox{ for all } x\in B(p^*,r^*),
\end{equation} 
where $\delta>0$ is a sufficiently small constant. Moreover, since $p^*$ is a pole of $f$, we have that $1+\lambda (p^*)^d=0$, so there exists a unique integer $k\in\{0,1,\hdots, d-1\}$ so that $p^*=\frac{1}{|\lambda|^{1/d}}\emm^{\im\theta+2\pi\im (k/d)}$, where $\theta=\frac{1}{d}(\pi-\Arg(\lambda))$. Since $k$ is an integer depending on $\lambda$ but not on the inputs $\lambda'$ or 
$\epsilon$, the value of $k$ which specifies $p^*$ among the poles of $f$ may be used by the algorithm.

Now, suppose that we are given inputs $\lambda'\in \CQ$ and rational $\eps>0$.  We want to output in time $poly(\size{\lambda',\epsilon})$ a bipartite graph of maximum degree $\Delta$ that implements $\lambda'$ with accuracy $\epsilon$. Clearly, we may assume that $\eps\in (0,1)$.   Let $M>0$ be a rational so that $M>2/(r^* \delta^d)$. The algorithm has three cases depending on the value of $|\lambda'|$, namely:
\begin{equation*}
\mbox{Case I: $|\lambda'|\geq M$}, \quad \mbox{Case II: $|\lambda'|\leq |\lambda|/(M+1)$}, \quad \mbox{Case III: $\lambda/(M+1) <|\lambda'|<M$}.
\end{equation*}
Note that since $\lambda,\lambda'\in \CQ$ and $M\in \QQ$, the algorithm can decide in time $poly(\size{\lambda,\lambda',M})=poly(\size{\lambda'})$ which of the three cases applies.\footnote{E.g., by squaring the inequalities, the radical of the norm goes away and the algorithm has just to compare rational numbers.}
\vskip 0.2cm
\noindent \textbf{Case I (large $\mathbf{|\lambda'|})$: $|\lambda'|\geq M$.}  The rough outline of the proof is to first specify and implement an activity $\lambda w$ for some appropriate $w$ whose main property is that $\lambda/(1+\lambda (f^{N}(w))^{d})$ is $\epsilon$-close to $\lambda'$; then, we will show how to implement the activity $\lambda/(1+\lambda (f^{N}(w))^{d})$ by using an appropriate tree construction.

We begin by specifying $w$. We first claim that there exists a unique $x^*\in B(p^*,r^*/2)$ such that
\begin{equation}\label{eq:3fgg56yy3}
\frac{\lambda}{1+\lambda (x^*)^{d}}=\lambda'.
\end{equation}
Indeed, for all $x$ such that $|x-p^*|=r^*/2$  it holds that (using \eqref{eq:v4vrccww} and the choice of $M$)
\begin{equation}\label{eq:4f4f4}
|1+\lambda x^{d}|=|\lambda|\cdot |x-p_1|\cdots |x-p_d|\geq |\lambda|r^*\delta^d/2> |\lambda|/|\lambda'|,
\end{equation}
so by Rouch\'{e}'s theorem\footnote{Rouch\'{e}'s theorem says that for any functions $f$ and $g$ that are holomorphic inside a region~$B$ surrounded by 
a simple closed contour $\partial B$, if $|g(x)| < |f(x)|$ on $\partial B$, then $f$ and $f+g$ have the same number of roots inside~$B$.} we have that the polynomial $1+\lambda x^{d}-\frac{\lambda}{\lambda'}$ has the same number of roots as the polynomial $1+\lambda x^{d}$ in the ball $B(p^*,r^*/2)$; the roots of the latter polynomial are precisely the poles $p_1,\hdots, p_d$ and therefore, by \eqref{eq:v4vrccww},  exactly one of those  lies in the ball $B(p^*,r^*/2)$ (namely $p^*$). This establishes the existence and uniqueness of $x^*\in B(p^*,r^*/2)$ satisfying \eqref{eq:3fgg56yy3}. 

By the first part of \eqref{eq:f5g53423}, we obtain that there exists $w^*\in  B(\omega,r')$ such that $f^{N}(w^*)=x^*$, i.e., 
\begin{equation}\label{eq:6445g5h4}
\frac{\lambda}{1+\lambda \big(f^{N}(w^*)\big)^{d}}=\lambda'.
\end{equation}
A fact we will use later is that 
\begin{equation}\label{eq:bttb45b1s12}
|f^{N}(w^*)-p^*|=|x^*-p^*|> \tau, \mbox{ where } \tau>0\mbox{ is such that }d \tau (|p^*|+\tau)^{d-1}<  1/|\lambda'|.
\end{equation} 
To see this, note that for all $x$ such that $|x-p^*|\leq \tau$ we have 
\begin{align*}
|1+\lambda x^{d}|&\leq |1+\lambda (p^*)^d|+|\lambda|\cdot|x^{d}- (p^*)^d|=|\lambda|\cdot \big|x^{d}- (p^*)^d\big|&
\\&\leq |\lambda|\cdot |x-p^*|\cdot\Big|\sum^{d-1}_{j=0}x^j(p^*)^{d-1-j}\Big| \leq d |\lambda| \tau (|p^*|+\tau)^{d-1}< |\lambda|/|\lambda'|,
\end{align*}
and therefore, by \eqref{eq:3fgg56yy3}, it must be the case that $|x^*-p^*|> \tau$, thus proving \eqref{eq:bttb45b1s12}. Note, we can compute $\tau\in \QQ$ satisfying \eqref{eq:bttb45b1s12} in time $poly(\size{\lambda'})$. Let 
\[\hat{\epsilon}:=\min\Big\{\frac{r-r'}{3},\frac{r^*}{4L},\frac{\tau}{4L},\frac{1}{4d L|\lambda'|(|p^*|+r^*)^{d-1}},\frac{\epsilon}{4d L|(\lambda')^2|(|p^*|+r^*)^{d-1}}, 1\Big\},\] 
and let $\epsilon'\in (0,\hat{\epsilon})$ be a rational  with $\size{\epsilon'}=poly(\size{\lambda',\epsilon})$. 

Note that $\fn{f}{N}(z)$ is a rational function of degree $d^N$ and, in fact, we can write it as  $\fn{f}{N}(z)=\frac{P(z)}{Q(z)}$ where $P(z),Q(z)$ are polynomials with coefficients in $\CQ$. Therefore,  we can rewrite \eqref{eq:6445g5h4} as a polynomial equation in terms of $w^*$, whose degree is at most $d^N$ (note that this is independent of $\lambda'$ and
$\epsilon$)  and whose coefficients have polynomial size in terms of $\size{\lambda'}$. Let $w_1,\hdots, w_t$ denote the roots of the polynomial. Using Lemma~\ref{fact:root}, we can compute $\hat{w}_1,\hdots, \hat{w}_t\in \CQ$ such that $|w_i-\hat{w_i}|\leq \epsilon'/2$  for all $i\in [t]$. In particular,  there exists $j\in [t]$ such that $|\hat{w}_j-w^*|\leq \epsilon'/2$ and therefore $\hat{w}_j\in B(\omega,r'+\epsilon'/2)$. Moreover, by applying \eqref{eq:4b32f2c24f} for $x_1=\hat{w}_j$ and $x_2=w^*$, we have
\[\big|\fn{f}{N}(\hat{w}_j)-x^*\big|=\big|\fn{f}{N}(\hat{w}_j)-\fn{f}{N}(w^*)\big|\leq L|\hat{w}_j-w^*|\leq L\epsilon'/2.\]
Therefore, by trying\footnote{Note that $\omega$ is the root of the polynomial $x+\lambda x^{d+1}-1$ with the smallest norm and therefore, using Lemma~\ref{fact:root}, we can compute $\hat{\omega}\in \CQ$ such that $|\hat{\omega}-\omega|\leq \epsilon'/6$  in time $poly(\size{\epsilon'})$. Similarly, $x^*$ is the unique root of the polynomial $1+\lambda x^{d}-\frac{\lambda}{\lambda'}$ in the ball $B(p^*,r^*/2)$ and therefore we can compute $\hat{x}^*\in \CQ$ such that $|\hat{x}^*-x^*|\leq L\epsilon'/6$  in time $poly(\size{\lambda',\epsilon'})$. Then, for $i\in[t]$, we check whether $|\hat{w}_i-\hat{\omega}|\leq 2\epsilon'/3$ and $|\fn{f}{N}(\hat{w}_i)-\hat{x}^*|\leq 2L\epsilon'/3$; the check must pass for at least one $\hat{w}\in\{\hat{w}_1,\hdots,\hat{w}_t\}$,  and using the triangle inequality we obtain that $\hat{w}$ satisfies \eqref{eq:ttbbg4643}. }  all of $\hat{w}_1,\hdots, \hat{w}_t$, we can specify $\hat{w}\in \{\hat{w}_1,\hdots, \hat{w}_t\}$ in time $poly(\size{\lambda',\epsilon})$  such that
\begin{equation}\label{eq:ttbbg4643}
\hat{w}\in B(\omega,r'+\epsilon'), \quad \big|\fn{f}{N}(\hat{w})-x^*\big|\leq L\epsilon'.
\end{equation}
Since $\epsilon'\leq (r-r')/3$ we have that $\hat{w}\in B(\omega,r)$. Recall that $r<\rho/|\lambda|$, so by the algorithm of  Proposition~\ref{lem:main121}, we can construct a bipartite graph $G$ of maximum degree $\Delta$ with terminal $v$ that  implements $\lambda \hat{w}$ with accuracy $\lambda \epsilon'$ in time $poly(\size{\lambda \hat{w},\lambda\epsilon'})=poly(\size{\lambda',\epsilon})$. Moreover, we can compute the values $\Zin_{G,v}(\lambda)$ and $\Zout_{G,v}(\lambda)$ at the same time.  

Let $w$ be such $\lambda w= \frac{\Zin_{G,v}(\lambda)}{\Zout_{G,v}(\lambda)}$, so that $|w-\hat{w}|\leq \epsilon'$. Since $r'+2\epsilon'<r$, we obtain using  \eqref{eq:ttbbg4643} that $w\in B(\omega,r)$. Further, by applying \eqref{eq:4b32f2c24f} for $x_1=w$ and $x_2=\hat{w}$, we get that $\big|\fn{f}{N}(w)-\fn{f}{N}(\hat{w})\big|\leq L\epsilon'$ and therefore by the triangle inequality and \eqref{eq:ttbbg4643} we have
\begin{equation}\label{eq:v4454hb64g}
\big|\fn{f}{N}(w)-\fn{f}{N}(w^*)\big|\leq 2L \epsilon'.
\end{equation}
We will next show that 
\begin{equation}\label{eq:4bhb5bb4}
\Big|\frac{\lambda}{1+\lambda \big(f^{N}(w)\big)^{d}}-\lambda'\Big|\leq \epsilon.
\end{equation}
Since $\fn{f}{N}(w^*)=x^*$ and $|x^*|\leq |p^*|+r^*/2$, we can conclude (using \eqref{eq:v4454hb64g}) that $\big|\fn{f}{N}(w^*)\big|, \big|\fn{f}{N}(w)\big|\leq |p^*|+r^*$. In turn, this gives
\begin{equation}\label{eq:t4b4tb4bnhur4}
\begin{aligned}
\big|\big(\fn{f}{N}(w)\big)^{d}-\big(\fn{f}{N}(w^*)\big)^{d}\big|&=\big|\fn{f}{N}(w)-\fn{f}{N}(w^*)\big|\cdot \Big|\sum^{d-1}_{j=0}\big(\fn{f}{N}(w)\big)^{j}\,\big(\fn{f}{N}(w^*)\big)^{d-1-j}\Big|\\
&\leq 2d L(|p^*|+r^*)^{d-1}\epsilon'\leq 1/(2|\lambda'|),
\end{aligned}
\end{equation}
where in the last inequality we used that $\epsilon'\leq \frac{1}{4d L|\lambda'|(|p^*|+r^*)^{d-1}}$. From \eqref{eq:6445g5h4}, \eqref{eq:t4b4tb4bnhur4} and the triangle inequality, we obtain that
\begin{equation}\label{eq:g5665g5nnn}
\begin{aligned}
\Big|1+\lambda \big(\fn{f}{N}(w)\big)^{d}\Big|&\geq \Big|1+\lambda \big(\fn{f}{N}(w^*)\big)^{d}\Big|-\Big|\lambda\big(\fn{f}{N}(w)\big)^{d}-\lambda\big(\fn{f}{N}(w^*)\big)^{d}\Big| 
\geq \frac{|\lambda|}{2|\lambda'|}.
\end{aligned}
\end{equation} 
and therefore
\begin{equation*}
\begin{aligned}
\Big|\frac{\lambda}{1+\lambda \big(f^{N}(w)\big)^{d}}-\lambda'\Big|&=  \frac{|\lambda|^2\cdot\big|\big(\fn{f}{N}(w)\big)^{d}-\big(\fn{f}{N}(w^*)\big)^{d}\big|}{\big|1+\lambda \big(f^{N}(w)\big)^{d}\big|\cdot \big|1+\lambda \big(f^{N}(w^*)\big)^{d}\big|}\\
&=\frac{|\lambda'\lambda|\cdot\big|\big(\fn{f}{N}(w)\big)^{d}-\big(\fn{f}{N}(w^*)\big)^{d}\big|}{\big|1+\lambda \big(f^{N}(w)\big)^{d}\big|}\leq 4d L|\lambda'|^2(|p^*|+r^*)^{d-1}\epsilon'\leq \epsilon,
\end{aligned}
\end{equation*}
where in the last equality we used \eqref{eq:6445g5h4}, in the second to last inequality we used \eqref{eq:t4b4tb4bnhur4} and \eqref{eq:g5665g5nnn}, and in the last inequality we used the choice of $\epsilon'$. This concludes the proof of \eqref{eq:4bhb5bb4}. To ensure that certain partition functions are non-zero, we will need the following additional fact for $w$, namely that 
\begin{equation}\label{eq:rg4g45nonz}
\fn{f}{n}(w)\neq p_1,\hdots,p_d \mbox{ for all } n=0,1,\hdots, N.
\end{equation}
For $n=0,1,\hdots, N-1$, this just follows from \eqref{eq:f5g53423} and the fact that $w\in B(\omega,r)$. For $n=N$, we have from \eqref{eq:v4454hb64g} and $f^{N}(w^*)=x^*$  that $f^{N}(w)$ is within distance $2L\epsilon'\leq r^*/2$ from $x^*\in B(p^*,r^*/2)$ and therefore $f^{N}(w)\in B(p^*,r^*)$. This implies that for all poles $p\neq p^*$ it holds that $f^{N}(w)\neq p$ (cf. \eqref{eq:v4vrccww}). For the pole $p^*$, we have from \eqref{eq:v4454hb64g} that $f^{N}(w)$ is within distance $2L\epsilon'\leq \tau/2$ from $f^{N}(w^*)$ and therefore, using \eqref{eq:bttb45b1s12}, we have that $|f^{N}(w)-p^*|\geq \tau/2>0$, i.e., $f^{N}(w)\neq p^*$. This finishes the proof of \eqref{eq:rg4g45nonz}.

In light of \eqref{eq:4bhb5bb4}, we next focus on implementing the activity $\lambda/(1+\lambda (f^{N}(w))^{d})$ using a bipartite graph of maximum degree $\Delta$. For $h= 0,1,\hdots,N$, let $T_h$ denote the $d$-ary tree of height $h$ and denote the root of the tree by $u_h$. Let $G_h$ be the bipartite graph of maximum degree $\Delta$ obtained from $T_h$ by taking, for each leaf $l$ of $T_h$, $d$ distinct copies of the graph $G$ (which implements $\lambda w$) and identifying $l$ with the $d$ copies of the terminal $v$ of $G$. Then, using that $\lambda w=\Zin_{G,v}(\lambda)/\Zout_{G,v}(\lambda)$, we have
\begin{equation}\label{eq:bt454b51244}
\begin{gathered}
\Zin_{G_0,u_0}(\lambda)=\lambda\big(\Zin_{G,v}(\lambda)/\lambda\big)^{d}, \qquad \Zout_{G_0,u_0}(\lambda)=\big(\Zout_{G,v}(\lambda)\big)^{d},\\
Z_{G_0}(\lambda)=\Zin_{G_0,u_{0}}(\lambda)+\Zout_{G_0,u_{0}}(\lambda)=\big(\Zout_{G,v}(\lambda)\big)^{d}(1+\lambda w^d).
\end{gathered}
\end{equation}
\eqref{eq:rg4g45nonz} ensures that $w \neq p_1,..,p_d$ so we have $1 + \lambda w^d \neq 0$. We therefore obtain from \eqref{eq:bt454b51244} that
\begin{equation}\label{eq:btb45b3f}
Z_{G_0}(\lambda)\neq 0 \mbox{ and }\frac{\Zout_{G_0,u_0}(\lambda)}{Z_{G_0}(\lambda)}=\frac{\big(\Zout_{G,v}(\lambda)\big)^{d}}{\big(\Zout_{G,v}(\lambda)\big)^{d}\big(1+\lambda w^d\big)}=f(w).
\end{equation}
Further, for $h=1,\hdots, N$ it holds that
\begin{equation}\label{eq:um76h6h6r}
\begin{gathered}
\Zin_{G_h,u_h}(\lambda)=\lambda\big(\Zout_{G_{h-1},u_{h-1}}(\lambda)\big)^{d}, \quad \Zout_{G_h,u_h}(\lambda)=\big(Z_{G_{h-1}}(\lambda)\big)^{d},\\
Z_{G_h}(\lambda)=\Zin_{G_h,u_h}(\lambda)+\Zout_{G_h,u_h}(\lambda)=\big(Z_{G_{h-1}}(\lambda)\big)^{d}\left(1+\lambda \frac{\big(\Zout_{G_{h-1},u_{h-1}}(\lambda)\big)^d}{\big(Z_{G_{h-1}}(\lambda)\big)^d}\right).
\end{gathered}
\end{equation}
We will show by induction that for all $h=0,1,\hdots,N$ it holds that
\begin{equation}\label{eq:mm4545mg4f5}
Z_{G_h}(\lambda)\neq 0\mbox{ and } \frac{\Zout_{G_h,u_h}(\lambda)}{Z_{G_h}(\lambda)} = \fn{f}{h+1}(w).
\end{equation}
For $h=0$, this is just \eqref{eq:btb45b3f}. Assume that it holds for $h-1$; we have by \eqref{eq:um76h6h6r} and the induction hypothesis that 
\[Z_{G_h}(\lambda)=\big(Z_{G_{h-1}}(\lambda)\big)^{d}\Big(1+\lambda (\fn{f}{h}(w)\big)^d\Big)\neq 0,\]
where the disequality follows from $Z_{G_{h-1}}(\lambda)\neq0$ and \eqref{eq:rg4g45nonz}. We therefore obtain that
\[\frac{\Zout_{G_h,u_h}(\lambda)}{Z_{G_h}(\lambda)}=\frac{1}{1+\lambda \Big(\frac{\Zout_{G_{h-1},u_{h-1}}(\lambda)}{Z_{G_{h-1}}(\lambda)}\Big)^d}=\frac{1}{1+\lambda \big(\fn{f}{h}(w)\big)^d}=\fn{f}{h+1}(w),\]
completing the induction and the proof of \eqref{eq:mm4545mg4f5}. For $h=N$, \eqref{eq:mm4545mg4f5} gives  that $Z_{G_N}(\lambda)\neq 0$ and 
\begin{equation}\label{eq:retb4hr6}
\frac{\Zout_{G_{N},u_{N}}(\lambda)}{Z_{G_{N}}(\lambda)}=\fn{f}{N+1}(w)=f(\fn{f}{N}(w))=\frac{1}{1+\lambda \big(\fn{f}{N}(w)\big)^d}.
\end{equation}

Consider the graph $G'$ obtained from $G_{N}$ by adding a new vertex $u'$ whose
single neighbour is the vertex $u_{N}$. Then, we have that
\begin{equation}\label{eq:brb6r6uu3254}
\Zin_{G',u'}(\lambda)= \lambda\Zout_{G_{N},u_{N}}(\lambda),\qquad \Zout_{G',u'}(\lambda)=Z_{G_{N}}(\lambda)\neq 0.
\end{equation}
Using this in conjuction with \eqref{eq:retb4hr6}, we conclude that
\[\frac{\Zin_{G',u'}(\lambda)}{\Zout_{G',u'}(\lambda)}=\frac{\lambda\Zout_{G_{N},u_{N}}(\lambda)}{Z_{G_{N}}(\lambda)}=\frac{\lambda}{1+\lambda \big(\fn{f}{N}(w)\big)^d}.\]
From this and \eqref{eq:4bhb5bb4}, we obtain that $G'$ is a bipartite graph of maximum degree $\Delta$ with terminal $u'$ which implements $\lambda'$ with accuracy $\epsilon$. Moreover, using \eqref{eq:bt454b51244}, \eqref{eq:um76h6h6r} and \eqref{eq:brb6r6uu3254}, we can also compute the values $\Zin_{G',u'}(\lambda),\Zout_{G',u'}(\lambda)$.

\vskip 0.2cm

\noindent The remaining cases of the algorithm (Cases II and III below) are almost identical to Cases II and III of the algorithm in Proposition~\ref{lem:realexp} for the real case, so we focus on the main differences (which mostly amount to modifying the upper and lower bounds for the relevant quantities). To align  with the notation there, let $G_4$ be a bipartite graph of maximum degree at most $\Delta$ with terminal $v_4$ that implements a constant activity $\lambda_4'$ with 
\begin{equation}\label{eq:56g56g4g543ww}
|\lambda_4'|>(M+1)(M+2), \mbox{ where } \lambda_4'=\frac{\Zin_{G_4,v_4}(\lambda)}{\Zout_{G_4,v_4}(\lambda)}.
\end{equation}
Note, this implementation can be done using the algorithm for Case I. We next  give the details of the algorithm for the remaining cases. 
\vskip 0.2cm

\noindent \textbf{Case II (small $\mathbf{|\lambda'|}$): $|\lambda'|\leq |\lambda|/(M+1)$}. We first assume that $\lambda'\neq 0$. Let $\hat{\lambda}$ be such that
\begin{equation}\label{eq:n56n6j6j}
\frac{\lambda}{1+\hat{\lambda}}=\lambda', \mbox{ so that } \hat{\lambda}=\frac{\lambda}{\lambda'}-1.
\end{equation}
Let $\hat{\epsilon}=\epsilon \cdot \min\big\{\frac{|\lambda|}{2|\lambda'|},\frac{|\lambda|}{2|\lambda'|^2},1\big\}$ and let $\epsilon'$ be a rational less than $\hat{\epsilon}$ so that $\size{\epsilon'}=poly(\size{\lambda',\epsilon})$.  

Using the assumption $|\lambda'|\leq |\lambda|/(M+1)$ and the triangle inequality, we have that $|\hat{\lambda}|\geq \frac{|\lambda|}{|\lambda'|}-1\geq M$.  Therefore, by Case I, we can construct in time $poly(\size{\lambda',\eps})$ a bipartite graph $G$ of maximum degree $\Delta$ with terminal $v$ that  implements $\hat{\lambda}$ with accuracy $\epsilon'$, and we can compute the values $\Zin_{G,v}(\lambda)$ and $\Zout_{G,v}(\lambda)$ at the same time.  Let $\lambda''= \frac{\Zin_{G,v}(\lambda)}{\Zout_{G,v}(\lambda)}$, so that $|\lambda''-\hat{\lambda}|\leq \epsilon'$. Using this and \eqref{eq:n56n6j6j}, we have
\begin{equation}\label{eq:b556n45y54hj}
|1+\lambda'' |\geq  |1+\hat{\lambda}|-|(\lambda''-\hat{\lambda})|\geq \frac{|\lambda|}{|\lambda'|}-\epsilon'\geq \frac{|\lambda|}{2|\lambda'|},
\end{equation}
and therefore 
\begin{equation}\label{eq:6h6h43gt3}
\begin{aligned}
\Big|\frac{\lambda}{1+\lambda'' }-\lambda'\Big|=\frac{|\lambda' ||\lambda''-\hat{\lambda}|}{|1+\lambda''|}\leq \frac{2|\lambda'|^2}{|\lambda|}|\lambda''-\hat{\lambda}|\leq \epsilon.
\end{aligned}
\end{equation}
Now,  let $G'$ be the bipartite graph obtained from $G$ by adding a new vertex $u$ whose single neighbour is the terminal $v$ of $G$. Then, just as in Case II of Proposition~\ref{lem:realexp} we can conclude that  $G'$ with terminal $u$ is a bipartite graph of maximum degree $\Delta$ which  implements $\lambda'$ with accuracy $\epsilon$, and we can also compute the values $\Zin_{G',u}(\lambda),\Zout_{G',u}(\lambda)$. The case $\lambda'=0$ can be handled using the above technique by implementing the activity $\epsilon''/2\neq 0$ with accuracy $\epsilon''/2$, where $\epsilon''>0$ is a rational less than $\min\{\epsilon,|\lambda|/(M+1)\}$ such that $\size{\epsilon''}=poly(\size{\epsilon})$, see Case II of the proof of Proposition~\ref{lem:realexp} for more details.

\vskip 0.2cm

\noindent \textbf{Case III (moderate $\mathbf{|\lambda'|}$): $|\lambda|/(M+1)<|\lambda'|<M$}. Let $x^*$ be such that 
\begin{equation}\label{eq:y6b6y54g5}
\frac{\lambda}{1+\lambda_4'  x^*}=\lambda',\mbox{ so that }x^*=\frac{1}{\lambda_4'}\Big(\frac{\lambda}{\lambda'}-1\Big).
\end{equation}
Let $\hat{\epsilon}=\epsilon\cdot \min\{\frac{|\lambda|^2}{2M|\lambda_4'|},\frac{|\lambda|^2}{2M^2|\lambda_4'|},1\}$ and let $\epsilon'$ be a rational less than $\hat{\epsilon}$ so that $\size{\epsilon'}=poly(\size{\lambda',\epsilon})$.

Using the assumption $|\lambda'|>|\lambda|/(M+1)$ and  $|\lambda_4'|> (M+1)(M+2)$ (see \eqref{eq:56g56g4g543ww}), we have that $|x^*|\leq 1/(M+1)$.  By the algorithm for Case II, we can implement the activity $\lambda x^*$ with precision $\epsilon'$ in time $poly(\size{\lambda x^*, \epsilon'})=poly(\size{\lambda', \epsilon})$, i.e., we can construct a bipartite graph $G$ of maximum degree at most $\Delta$ with terminal $v$ such that, for  $\lambda'':=\Zin_{G,v}(\lambda)/\Zout_{G,v}(\lambda)$, it holds that
\begin{equation}\label{eq:mbtjbt65}
\big|\lambda''- \lambda x^*\big|  \leq \epsilon'.
\end{equation}
Now, using \eqref{eq:y6b6y54g5} and \eqref{eq:mbtjbt65}, we have by the triangle inequality that
\begin{equation}\label{eq:g4g4g}
|1+ \lambda_4'(\lambda''/\lambda)|\geq |1+ \lambda_4' x^*|-|\lambda_4' (x^*-\lambda''/\lambda)|\geq \frac{|\lambda|}{M}- \frac{|\lambda_4'|\epsilon'}{|\lambda|}\geq \frac{|\lambda|}{2M},
\end{equation}
and therefore
\begin{equation}\label{eq:v4v4v35g3f2f}
\begin{aligned}
\Big|\frac{\lambda}{1+ \lambda_4'(\lambda''/\lambda)}-\lambda'\Big|&=\Big|\frac{\lambda}{1+ \lambda_4'(\lambda''/\lambda)}-\frac{\lambda}{1+ \lambda_4'x^*}\Big|=\frac{|\lambda|\, |\lambda'_4|\, |x^*-(\lambda''/\lambda)|}{|1+ \lambda_4'x^*|\cdot |1+ \lambda_4'(\lambda''/\lambda)|}\\
&= \frac{|\lambda'|\, |\lambda'_4|\, |x^*-(\lambda''/\lambda)|}{|1+ \lambda_4'(\lambda''/\lambda)|}\leq 2M^2 \frac{|\lambda'_4|}{|\lambda|^2} |\lambda x^*-\lambda''|\leq \epsilon.
\end{aligned}
\end{equation}
where in the last equality we used \eqref{eq:y6b6y54g5}, in the second to last inequality we used \eqref{eq:g4g4g} and $|\lambda'|\leq M$, and in the last inequality we used \eqref{eq:mbtjbt65} and the choice of $\epsilon'$.

Recall from \eqref{eq:56g56g4g543ww} that $G_4$ is a bipartite graph of maximum degree $\Delta$  with terminal $v_4$ which implements the activity $\lambda_4'$. Let $G'$ be the bipartite graph obtained by taking  a copy of $G_4$ and $G$ and identifying the terminals $v_4,v$ into a single vertex which we label $u'$. Further, consider the graph $G''$ obtained from $G'$ by adding a new vertex $u''$ whose single neighbour is the vertex $u'$. Then, just as in Case III of Proposition~\ref{lem:realexp}, using \eqref{eq:v4v4v35g3f2f} we can conclude that $G''$ is a bipartite graph of maximum degree at most $\Delta$ with terminal $u''$ which implements $\lambda'$ with accuracy $\epsilon$, and we can also compute the values $\Zin_{G',u}(\lambda),\Zout_{G',u}(\lambda)$.

\vskip 0.2cm

This completes the three different cases of the algorithm, thus completing the proof of Proposition~\ref{lem:complexexp}.
\end{proof}

\section{ \#P-hardness}\label{sec:reductions}

In order to prove Theorems~\ref{thm:norm} and~\ref{thm:arg} 
we first prove \#P-hardness of multivariate versions of our problems.
Instead of insisting that every vertex has activity~$\lambda$, we allow the
activities to be drawn from  
the set
\begin{equation}\label{eq:Llam}
\mathcal{L}(\lambda) = \{\lambda,-\lambda-1,-1,1\}.
\end{equation}

 \prob{
$\FactorMVHardCore{K}$.} { A  graph $G$ with maximum degree at
most $\Delta$. An activity vector $\lambdab=\{\lambda_v\}_{v\in V}$, such that, for each $v\in V$, $\lambda_v \in \mathcal{L}(\lambda)$.
For every vertex $v\in V$ with $\lambda_v\neq \lambda$, the degree of~$v$ in~$G$ must be at most~$2$.} 
{   If $|Z_G(\lambdab)|=0$  then the algorithm may output any rational number. Otherwise,
 it must output a  rational number $\widehat{N}$ such that
$\widehat{N}/K \leq
|Z_{G}(\lambdab)|\leq K \widehat{N}$. }  
 
  \prob{
$\ArgMVHardCore{\rho}$.} {A graph $G=(V,E)$ with maximum degree at
most $\Delta$. An activity vector $\lambdab=\{\lambda_v\}_{v\in V}$ such that, for each $v\in V$, $\lambda_v \in \mathcal{L}(\lambda)$.
For every vertex $v\in V$ with $\lambda_v\neq \lambda$, the degree of~$v$ in~$G$ must be at most~$2$.
 }
{ If $Z_G(\lambdab)=0$ then the algorithm may output any rational number. Otherwise, it must output  
a rational number $\widehat{A}$ such that, for some
$a\in \arg(Z_{G}(\lambdab))$,
$ |\widehat{A} - a| \leq \rho$.  }

\subsection{Reducing the degree using equality gadgets}

Given a graph $B=(V,E)$ and 
two subsets $T_{\textin}$ and $T$ of the vertex set $V$
satisfying $T_{\textin}\subseteq T$, let 
$\mathcal{I}_{B,T,T_{\textin}}$ denote the set of independent sets $I$ of~$B$ 
such that $I\cap T=T_{\textin}$.
Let $Z_{B,T,T_{\textin}}(\lambda) = 
\sum_{I\in \mathcal{I}_{B,T,T_{\textin}}}\lambda^{|I|}$.
We use similar notation when activities are non-uniform.
We start by introducing ``equality'' gadgets.

\begin{figure}[h]
\begin{center}
{
\tikzset{lab/.style={circle,draw,inner sep=0pt,fill=none,minimum size=5mm}} 
\begin{tikzpicture}[xscale=1,yscale=1]
\draw (0,0) node[lab, label={below:$\lambda_{u_i}= \lambda$}] (1) {$u_i$};
\draw (0,-3) node[lab, label={below:$\lambda_{v_i}=1$}] (2) {$v_{i}$};
\draw (6,-1.5) node[lab, label={below:$\lambda_{z_i}=\lambda$}] (5) {$z_i$};
\draw (3,-0.75) node[lab, label={above:$\lambda_{x_i}=-\lambda-1$}] (3) {$x_i$};
\draw (3,-2.25) node[lab, label={below:$\lambda_{y_i}=-\lambda-1$}] (4) {$y_i$};
\draw (1) -- (3) -- (5);
\draw (2) -- (4) -- (5);
\draw (3,-4.5) node {$B_i$ with activity vector $\lambdab$};
\begin{scope}[shift={(10,0)},xscale=0.8,yscale=0.8]
\draw (-3,0.75) node[lab, label={below:$\lambda_{s_i}= \lambda$}] (0a) {$s_i$};
\draw (-3,-3.75) node[lab, label={below:$\lambda_{t_i}= 1$}] (0b) {$t_i$};
\draw (0,0) node[lab, label={below:$\lambda_{u_i}= -1$}] (1) {$u_i$};
\draw (0,-3) node[lab, label={below:$\lambda_{v_{i}}= -1$}] (2) {$v_{i}$};
\draw (6,-1.5) node[lab, label={below:$\lambda_{z_i}=1$}] (5) {$z_i$};
\draw (3,-0.75) node[lab, label={above:$\lambda_{x_i}=-1$}] (3) {$x_i$};
\draw (3,-2.25) node[lab, label={below:$\lambda_{y_i}=-1$}] (4) {$y_i$};
\draw (0a) -- (1) -- (3) -- (5);
\draw (0b) -- (2) -- (4) -- (5);
\draw (1,-5.7) node {$B'_i$ with activity vector $\lambdab'$};
\end{scope}
\end{tikzpicture}
}
\end{center}
\caption{The binary equality gadgets $B_i$  
and $B'_i$  used in the proof of Lemma~\ref{lem:equality}.   }
\label{fig:equal}
\end{figure}
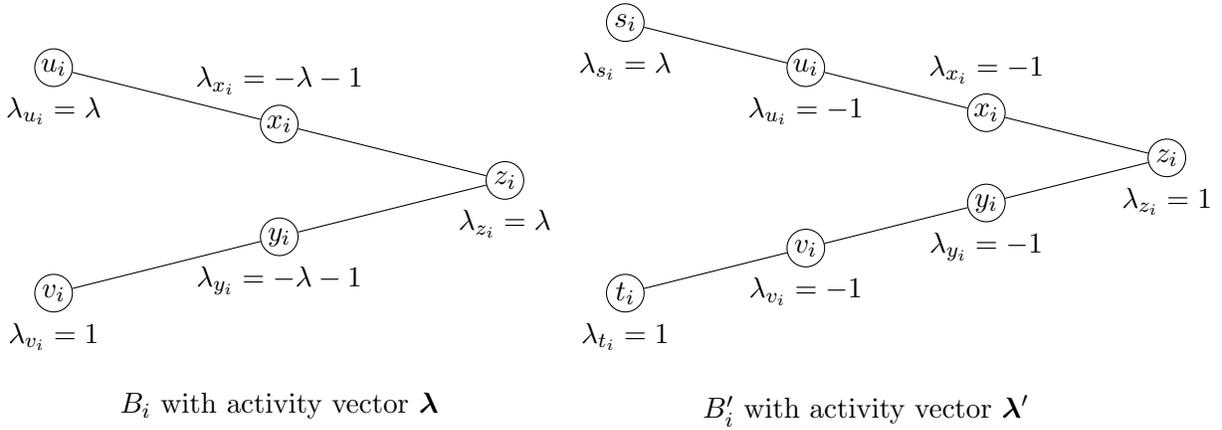 
\begin{lemma}\label{lem:equalitygadget}
Let $\lambda\in \Complex$ be such that $\lambda\neq -1,0$.  Let $B_i$ and $B'_i$ be the graphs in Figure~\ref{fig:equal} with activity vectors $\lambdab$ and $\lambdab'$, respectively, and set $T_i=\{u_i,v_i\}$, $T_i'=\{s_i,t_i\}$. Then,
$$Z_{B_i,T_i,\{u_i\}}(\lambdab) = Z_{B_i,T_i,\{v_{i}\}}(\lambdab)=0, \quad Z_{B_i,T_i,\emptyset}(\lambdab) = Z_{B_i,T_i,T_i}(\lambdab)=C:=\lambda(1+\lambda)\neq 0,$$
and
$$ Z_{B'_i,T'_i,\{s_i\}}(\lambdab') = Z_{B'_i,T'_i,\{{t}_i\}}(\lambdab') =0, \quad Z_{B'_i,T'_i,\emptyset}(\lambdab') =\frac{1}{\lambda} Z_{B'_i,T'_i,T'_i}(\lambdab')=1.$$
\end{lemma}
\begin{proof}
By enumerating the independent sets of $B_i$, we have 
\begin{equation}\label{eq:tv465b6b4}
\begin{aligned}
&Z_{B_i,T_i,\{u_i\}}(\lambdab)=\lambda_{u_i}(1+\lambda_{y_i}+\lambda_{z_i}),\qquad   Z_{B_i,T_i,\{v_{i}\}}(\lambdab) = \lambda_{v_i}(1+\lambda_{x_i}+\lambda_{z_i}),\\
&Z_{B_i,T_i,\emptyset}(\lambdab) = 1+ \lambda_{x_i}+\lambda_{y_i}+\lambda_{z_i}+ \lambda_{x_i}\lambda_{y_i},\quad Z_{B_i,T_i,T_i}(\lambdab)=\lambda_{u_i}\lambda_{v_i}(1+\lambda_{z_i}).
\end{aligned} 
\end{equation}
Observe that $B_i'$ is obtained from $B_i$ by ``appending'' the vertices $s_i$ and $t_i$. Therefore,
\begin{equation}\label{eq:45g44r2}
\begin{aligned}
Z_{B_i',T_i',\emptyset}(\lambdab')& = Z_{B_i,T_i,\emptyset}(\lambdab')+Z_{B_i,T_i,\{u_i\}}(\lambdab')+Z_{B_i,T_i,\{v_i\}}(\lambdab')+Z_{B_i,T_i,T_i}(\lambdab'),\\
Z_{B_i',T_i',\{s_i\}}(\lambdab')&=\lambda_{s_i}\big(Z_{B_i,T_i,\emptyset}(\lambdab')+Z_{B_i,T_i,\{v_i\}}(\lambdab')\big),\\
Z_{B_i',T_i',\{t_{i}\}}(\lambdab')&= \lambda_{t_i}\big(Z_{B_i,T_i,\emptyset}(\lambdab')+Z_{B_i,T_i,\{u_i\}}(\lambdab')\big),\\
Z_{B_i',T_i',T_i'}(\lambdab')&=\lambda_{s_i}\lambda_{t_i}Z_{B_i,T_i,\emptyset}(\lambdab').
\end{aligned} 
\end{equation}
Using \eqref{eq:tv465b6b4}, we have that 
\[Z_{B_i,T_i,\emptyset}(\lambdab')=1,\quad Z_{B_i,T_i,\{u_i\}}(\lambdab')=Z_{B_i,T_i,\{v_i\}}(\lambdab')=-1, \quad Z_{B_i,T_i,T_i}(\lambdab')=2.\]
Plugging this into \eqref{eq:45g44r2} concludes the proof of the lemma.
\end{proof}

The following lemma defines the gadgets for the case $\lambda=-1$.

\begin{lemma}\label{lem:minusoneeqgadget}
Let $\lambda=-1$. Let $B_i$ be the path of length six with endpoints $u_i,v_i$, with activity vector $\lambdab$, where every vertex has activity $\lambda=-1$, apart from the endpoint $v_i$  
   which has activity $+1$. Let $B_i'$ be the graph in Figure~\ref{fig:equal} with the activity vector $\lambdab'$ given there. Set $T_i=\{u_i,v_i\}$ and $T_i'=\{s_i,t_i\}$. Then,
$$Z_{B_i,T_i,\{u_i\}}(\lambdab) = Z_{B_i,T_i,\{v_{i}\}}(\lambdab)=0, \quad Z_{B_i,T_i,\emptyset}(\lambdab) = Z_{B_i,T_i,T_i}(\lambdab)=C:=1,$$
and
$$Z_{B'_i,T'_i,\{s_i\}}(\lambdab') = Z_{B'_i,T'_i,\{t_i\}}(\lambdab') =0, \quad Z_{B'_i,T'_i,\emptyset}(\lambdab') =\frac{1}{\lambda} Z_{B'_i,T'_i,T'_i}(\lambdab')=1.$$
\end{lemma}
\begin{proof}
Note that both graphs are paths of length six.  We can therefore use the formula in \eqref{eq:45g44r2}. The lemma therefore follows by just making the substitutions.
\end{proof}

The following lemma shows how to use the equality gadgets  to 
replace high-degree vertices with equivalent subgraphs made up of low-degree vertices.
For a graph $G$ and a vertex $v$ in $G$, we denote by $d_v(G)$ the number of neighbours of $v$.
\begin{lemma}\label{lem:equality}
Let $\lambda\in \Complex_{\neq 0}$. 
If $\lambda=-1$, let $C=1$; otherwise, let 
$C= \lambda(1+\lambda)\neq 0$.

Suppose that $G=(V,E)$ is a graph with an activity vector $\lambdab=\{\lambda_v\}_{v\in V}$ such that, for
every vertex $v\in V$, we have $\lambda_v \in \{1,\lambda\}$.
Let $U_1 = \{ v\in V \mid \lambda_v = 1\}$ and $U_\lambda = \{v\in V \mid \lambda_v=\lambda\}$. Consider an arbitrary set $S\subseteq V$.
Then, there is a set $S'$ of vertices (distinct from $V$)
and a graph $G'=(V',E')$ with an activity vector $\lambdab'=\{\lambda_v'\}_{v\in V'}$ such  that
$$Z_G(\lambdab) = \frac{Z_{G'}(\lambdab')}
{
\prod_{v\in S\cap U_1} C^{d_v(G)} 
\prod_{v\in S\cap U_\lambda} C^{d_v(G)-1}}.$$
 Furthermore,
\begin{itemize}
\item $V' = S' \cup (V\setminus S)$ and $|S'|$ is at most  $8 |V|^2$.
\item Every vertex $v\in S'$ has $d_v(G')\leq 3$ and $\lambda'_v \in \mathcal{L}(\lambda).$
If $\lambda'_v\neq \lambda$ then $d_v(G') \leq 2$.
\item Every vertex  in $V(G)\setminus S$ 
has $d_v(G)=d_v(G')$ and $\lambda_v'=\lambda_v$.
\end{itemize}
 \end{lemma}
\begin{proof}
To prove the lemma, we can assume that $|S|=1$.
(To prove the lemma for a larger set $S$ we just repeatedly apply the 
singleton-set version to each of the vertices in~$S$.)
So let $S = \{v\}$ and, for convenience, let $d=d_v(G)$ be the degree of $v$ in $G$. There are two cases.
 
\begin{itemize}
\item {\bf Case 1. $v\in U_1$.} Assume first that $d\neq 1$. In this case, $S'$ will be the union of the vertices in the gadgets $B_1,\ldots,B_{d}$ by identifying vertex $v_i$ of $B_i$ with vertex $u_{i+1}$ of $B_{i+1}$ for each $i=1,\hdots, d$ (we will use the conventions that $u_{d+1}\equiv u_1$, $B_{d+1}\equiv B_1$, and  $B_{0}\equiv B_d$).
To construct $G'$ from $G$ we  
replace $v$ with the union of
 these gadgets.
If a vertex $w$ is in exactly one of these gadgets, then the activity $\lambda'_w$ will be inherited from the gadget. Also, $w$ will have no other neighbours in $G'$ (other than the neighbours in its gadget).
Now, for $i=1,\hdots,d$, the vertex $u_i$ is in two gadgets, namely $B_{i}$ and $B_{i-1}$.
In addition to having its gadget neighbours, this vertex $u_i$ will be connected to the $i$'th neighbour of~$v$ in~$G$.
Then we will  set $\lambda'_{u_i} = \lambda$, since this is the product of the activities inherited from~$B_{i}$ and~$B_{i+1}$.
Note from Lemma~\ref{lem:equalitygadget} (and Lemma~\ref{lem:minusoneeqgadget} in the case $\lambda=-1$) that, in the resulting graph $G'$, in any
independent set of non-zero weight, either  
all $v_i$'s are in the independent set, or all $v_i$'s are missing. 
In each of these cases, we get a factor of $C^{d}$ in the partition function. The construction for $d=1$ is analogous, i.e., we replace $v$ with the gadget $B_1$, but now we do not do any identification of vertices and we use $u_1$ to connect to the neighbour of $v$ in $G$; further, every vertex in $B_1$ retains its activity in $G'$. As before, we conclude that we get a factor of $C$ in the partition function.

 \item {\bf Case 2. $v\in U_\lambda$.}\quad
 Assume first that $d\neq 1$. In this case, $S'$ will be the union of the gadgets
 $B_1,\ldots,B_{d-1},B'_{d}$. We will further identify vertex $v_i$ of $B_i$ with vertex $u_{i+1}$ for each $i=1,\hdots, d-2$; for $i=d-1$, we will identify vertex $v_{d-1}$  of $B_{d-1}$ with vertex $s_d$ of $B_d'$ and vertex $t_d$ of $B_d'$ with vertex $u_1$ of $B_1$.
To construct $G'$ from $G$ we  
we will replace $v$ with the union of
these gadgets.
The assignment of activities is the same as  in Case 1; the only difference is in the construction of the graph $G'$ where now, for $i=1,\hdots, d-1$, the $i$'th neighbour of $v$ connects to the vertex $u_i$ while the $d$'th neighbour of $v$ connects to the vertex $s_d$. From Lemma~\ref{lem:equalitygadget} (and Lemma~\ref{lem:minusoneeqgadget} in the case $\lambda=-1$), we have that, in the resulting graph $G'$, in any
independent set of non-zero weight, either  
all $v_i$'s are in the independent set, or all $v_i$'s are missing. 
If they are all in, then we get a factor of $C^{d-1}  \lambda$ in the partition function; otherwise, we get a factor of $C^{d-1} $. The construction for $d=1$ is analogous, i.e., we replace $v$ with the gadget $B_1'$, but now we do not do any identification of vertices and we use $s_1$ to connect to the neighbour of $v$ in $G$; further, every vertex in $B_1'$ retains its activity in $G'$. As before, we conclude that in any
independent set of non-zero weight, either  
all $v_i$'s are in the independent set (contributing a factor of $\lambda$), or all $v_i$'s are missing (contributing a factor of $1$). 
\end{itemize}
This concludes  the proof of Lemma~\ref{lem:equality}.
\end{proof}
  
  \subsection{\#P-hardness of the multivariate problem}

\begin{theorem}\label{thm:MV}
Let $\Delta\geq 3$ and  
$\lambda\in \CQ$ be a complex number such that
$\lambda\not\in (\LambdaD \cup \Reals_{\geq 0})$.  Then, for $ K  =1.02$,
$\FactorMVHardCore{K}$ is $\numP$-hard. Also, for $\rho = 9\pi/24$, 
$\ArgMVHardCore{\rho}$ is $\numP$-hard.
\end{theorem}

\begin{proof}
Counting  the number of independent sets of an input graph is a well-known $\numP$-hard problem. We will reduce this to both problems. To do this, let $H$ be  an $n$-vertex graph, and let $N=Z_{H}(1)$ denote the number of independent sets of~$H$. Our goal is to 
show how to use an oracle for 
$\FactorMVHardCore{K}$
or
$\ArgMVHardCore{\rho}$ to
compute $N$. Let $\epsilon = 1/50$ and $\eta = 1/35$.

 Consider any   rational number $M$ in the range $ 0 \leq M \leq 2^n$.
Let $J_M$ be a graph with terminal $v$ 
and maximum degree at most~$\Delta$
that implements $-M$ with accuracy~$\epsilon$; note, by applying Proposition~\ref{lem:complexexp} and Proposition~\ref{lem:realexp} (for complex and real $\lambda$, respectively), there is an algorithm to construct $J_M$ in time $poly(\size{M,\epsilon})=poly(\size{M})$ that also outputs the exact values of $\Zin_{J_M,v}(\lambda),\Zout_{J_M,v}(\lambda)$. Let
\begin{equation}\label{eq:searchtwo}
\lambda_M := {\Zin_{J_M,v}(\lambda)}/{\Zout_{J_M,v}(\lambda)}, \mbox{ and note that } \Zout_{J_M,v}(\lambda)\neq 0 \mbox{ and } |\lambda_M+M|\leq \epsilon.
\end{equation}
Let $G_M$ be the graph 
formed by taking the disjoint union of $H$ and $J_M$ and
attaching the terminal $v$ of $J_M$ to every vertex in~$H$.
Let $\lambdab_M$ be the activity vector for $G_M$ obtained by setting the activities of vertices originally belonging to $H$ equal to 1 and the activities of vertices originally belonging to $J_M$ equal to $\lambda$.
Note that 
\begin{equation}\label{eq:4gb564h4vvw}
Z_{G_M}(\lambdab_M) = \Zin_{J_M,v}(\lambda) + \Zout_{J_M,v}(\lambda) Z_H(1)=\Zin_{J_M,v}(\lambda) + \Zout_{J_M,v}(\lambda) N.
\end{equation}

Let $G'_M$ and $\lambdab'_M$ be the graph 
and activity vector
constructed by applying Lemma~\ref{lem:equality} to $G_M$
with $S = V(H) \cup \{v\}$. 
Note that the size of $G_M$
is at most a polynomial in $\size{M}$ so
 (from Lemma~\ref{lem:equality})
 the size of $G'_M$ is also at most a polynomial in $\size{M}$.
From Lemma~\ref{lem:equality}, we have that
\begin{equation}\label{eq:6bh56n6h7}
Z_{G_M}(\lambdab_M) = Z_{G'_M}(\lambdab'_M)/W_H,\mbox{ where } W_H:=C^{n-1} 
\prod_{v\in V(H)} C^{d_v(H)},
\end{equation}
and $C$ is the constant in Lemma~\ref{lem:equality}.
Furthermore, every vertex $v$ 
of $G'_M$
has degree at most~$\Delta$ and every vertex $v$ of $G'_M$
with $\lambda'_v\neq \lambda$ has
degree at most~$2$ in~$G'_M$.
Thus, $G'_M$ is a valid input to 
 $\FactorMVHardCore{K}$ and $\ArgMVHardCore{\rho}$. Moreover, combining \eqref{eq:4gb564h4vvw},\eqref{eq:6bh56n6h7} and dividing through by $ \Zout_{J_M,v}(\lambda) $,
we have 
\begin{equation}
\label{eq:searchone}
\lambda_M +    N=\frac{Z_{G'_M}(\lambdab'_M)}{W_H\, \Zout_{J_M,v}(\lambda)}=:f_M.
\end{equation}
 
\noindent  {\bf Part one: } \#P-hardness of $\FactorMVHardCore{1.02}$.
   
By the triangle inequality (in the form $\big||a|-|b|\big|\leq |a+b|$),  we have
$$\big||f_M|- |M-N|\big|\leq |\lambda_M + N+M-N|= |\lambda_M + M| \leq \epsilon,$$
where in the last inequality we used \eqref{eq:searchtwo}. Therefore, $|N-M|-\epsilon \leq |f_M| \leq |N-M|+\epsilon$.

Consider $M$ so that
 $|N-M|\geq1$. Then
$|f_M|$ is not $0$. 
From the definition of~$f_M$ in~\eqref{eq:searchone}, 
this means that
$|Z_{G'_M}(\lambdab'_M)|\neq 0$.
Using an oracle for $\FactorMVHardCore{$1.02$}$ 
we can produce
an estimate for $|Z_{G'_M}(\lambdab'_M)|$ within a factor of $1.02$.
We can also obtain 
an estimate of the value 
$|W_H  \Zout_{J_M,v}(\lambda) |$ within a factor of $(1+7\eta/8)/1.02=1.025/1.02$, since $\Zout_{J_M,v}(\lambda)$ is
output by 
the algorithm from
Proposition~\ref{lem:complexexp} and Proposition~\ref{lem:realexp}. Combining these,  we obtain an estimate $\hat{f}_M$ for
$|f_M|$ 
satisfying
$(1-\eta) |f_M| \leq   \hat{f}_M \leq   (1+\eta) |f_M|$.
We now use the binary search technique of \cite{ComplexIsing}.

 The invariant that we will maintain is that we have an interval 
 $[M_{\text{start}},M_{\text{end}}]$ 
 of real numbers
 with $M_{\text{start}} \leq N \leq M_{\text{end}}$.
 Initially, $M_{\text{start}} =  0$ and $M_{\text{end}}=2^n$.
 Let $\ell=M_{\text{end}}-M_{\text{start}}$.
 If $\ell < 1$ then there is only one integer 
 between $M_{\text{start}}$ and $M_{\text{end}}$, so the value of~$N$ is known.
 
 Suppose $\ell\geq 1$. 
 For $i\in\{0,\ldots, 8\}$,
 let $M_i = M_{\text{start}} + i \ell/ 8$.
 For $i\in\{0,\ldots, 7\}$,
 let $s_i$ be the sign (positive, negative, or zero) of $\hat{f}_{M_i} - \hat{f}_{M_{i+1}}$.
 
 First, consider $i\in \{0,\ldots, 7\}$ and suppose 
 $N\geq M_{i+2}$.
 Then,
  \begin{align*}
 \hat{f}_{M_i} - \hat{f}_{M_{i+1}} &\geq
 (1-\eta) |f_{M_i}| - (1+\eta) | f_{M_{i+1}}|\\
 &\geq (1-\eta)(N-M_i - \epsilon) - (1+\eta) (N-M_i-\ell/ 8+\epsilon) \\
 &=  (1+\eta)\ell/ 8- 2 \eta (N-M_i)  - 2 \epsilon\\
 &\geq  (1+\eta)\ell/ 8 - 2 \eta  \ell  - 2 \epsilon.
 \end{align*}
 This is positive since $2\eta < (1+\eta)/16$ and  
 $\epsilon \leq \eta \leq \eta \ell$,  so $s_i$ is positive.
 Similarly, if $N \leq M_{i-1}$ then
  \begin{align*}
 \hat{f}_{M_{i}} - \hat{f}_{M_{i+1}} &\leq
 (1+\eta) |f_{M_i}| - (1-\eta) | f_{M_{i+1}}|\\
 &\leq (1+\eta)(M_i-N + \epsilon) - (1-\eta) (M_i-N+\ell/ 8-\epsilon) \\
 &=  -(1-\eta)\ell/ 8 + 2 \eta ( M_i-N)  + 2 \epsilon\\
   &\leq  -(1-\eta)\ell/ 8 + 2 \eta  \ell  + 2 \epsilon,
 \end{align*} 
 so $s_i$ is negative.
  
Now,  consider $i^*$ so that
 $M_{i^*} \leq N \leq M_{i^*+1}$.
Then $s_0,\ldots, s_{i^*-2}$ are plus
and $s_{i^*+2},\ldots, s_{ 7}$ are minus.
So we have a (possibly non-empty) sequence of pluses followed by three
arbitrary signs followed by a (possibly non-empty) sequence of minuses.
If there are  three minuses in a row at the end of the sequence,
we must have  
$N\leq M_{ 7}$ (otherwise  $s_5$ would have been a plus).
So we can shrink the interval by
redefining $M_{\text{end}}$ to be $M_{ 7}$.
Otherwise, there are three pluses in a row at the beginning of the sequence.
This means that $N\geq M_1$ (otherwise $s_2$ would have been negative). So, in this case, we can shrink the interval by
redefining $M_{\text{start}}$ to be $M_1$.
Either way, the interval shrinks to $ 7/ 8$ of its original size, so we can recurse
on the new interval; after at most $poly(n)$ steps, we will have $M_{\text{end}}-M_{\text{start}}<1$, which gives us the exact value of $N$.

\noindent  {\bf Part two: } \#P-hardness of $\ArgMVHardCore{9 \pi/24}$.
   
 Consider any   rational number $M$ in the range $ 0 \leq M \leq 2^n$.   
 We will show that, if $N> M+1/7$ then
there is an $a\in \arg(f_M)$ such that
$-\pi/12 <  a <   \pi/12$.
Also, if $N<M-1/7$ then
there is an $a\in \arg(f_M)$ such that
  $\pi - \pi/12 <  a  <\pi + \pi/12$.

To prove these claims, let $x_M = \lambda_M + M$ so by~\eqref{eq:searchtwo},
$|x_M| \leq \epsilon$.  
Then by~\eqref{eq:searchone},
$f_M = \lambda_M + N = x_M + N-M $.
Suppose $N>M+1/7$ and
consider $\theta \in   \arg(f_M) = \arg(x_M + N-M)$.
For concreteness (by adding integer multiples of $2\pi$ if necessary),
suppose that $\theta$ is in the range $[-\pi,\pi)$.
Then $ \tan(\theta)\leq \frac{\epsilon}{N-M} \leq 7 \epsilon$.
But $\tan(\pi/12) > 0.26 > 7\epsilon$.
So 
$\theta \leq \pi/12$.  
Similarly, $\theta \geq -\pi/12$.
The case $M>N+1/7$  is similar (restricting $\theta$ to $[0,2\pi)$).

Now suppose 
$|N-M|>1/7$. 
Since $f_M = x_M + N - M$, we have $f_M\neq 0$.
Since we can compute the value of   
$W_H \Zout_{J_M,v}(\lambda)$ exactly,
we can also compute   $\Arg(W_H \Zout_{J_M,v}(\lambda) )$  
within $\pm \pi/48$.
Using
an oracle for   $\ArgMVHardCore{9\pi/24}$ 
 with  input $(G'_M,\lambdab'_M)$
 we thus  
obtain an estimate $\hat{A}_M$
such that, for some $a\in \arg(f_M)$,
 $|\hat{A}_M -  a| \leq 9\pi/24+\pi/48=19\pi/48$.
 
As in Part One, we now do binary search.
Again, the invariant that 
 we will maintain is that we have an interval 
 $[M_{\text{start}},M_{\text{end}}]$ 
 of real numbers
 with $M_{\text{start}} \leq N \leq M_{\text{end}}$.
 Initially, $M_{\text{start}} =  0$ and $M_{\text{end}}=2^n$.
 Let $\ell=M_{\text{end}}-M_{\text{start}}$.
 If $\ell < 1$ then there is only one integer 
 between $M_{\text{start}}$ and $M_{\text{end}}$, so the value of~$N$ is known. 

 Suppose $\ell\geq 1$. 
 For $i\in\{0,\ldots, 6\}$,
 let $M_i = M_{\text{start}} + i \ell/ 6$.
 Let $s_i$ be  minus if 
 there is an integer~$j$ such that 
  $\pi/2 < \hat{A}_{M_i} + 2 \pi j  < 3\pi/2$.
  Let  $s_i$ be plus if 
  there is an integer~$j$ such that
  $-\pi/2 < \hat{A}_{M_i} + 2 \pi j < \pi/2$.
  If neither of these occurs, then $s_i$ is undefined.
 
If $N\leq M_{i-1}$ then $s_i$ is  minus.
If $N\geq  M_{i+1}$ then $s_i$ is  plus.
So $s_0,\ldots,s_6$ consists of a (possibly empty) sequences of  pluses followed by one unknown value
followed by a (possibly empty) sequence of  minuses.

Suppose that  $s_0$, $s_1$ and $s_2$ are all  plus.
Then $N\geq M_1$ (otherwise $s_2$ would be  minus).
Otherwise, $s_4$, $s_5$ and $s_6$ are all minus.
In this case, $N\leq M_5$ (otherwise $s_4$ would be plus). 
Either way, we can shrink the interval to $5/6$ of its original length, so, as in Part one, we can discover the value of $N$.
\end{proof}

\subsection{Restricting to bipartite graphs}\label{sec:MVtobip}

In this section, we reduce the problems $\FactorMVHardCore{K}$ and
$\ArgMVHardCore{\rho}$ to multivariate versions whose inputs are further restricted to bipartite graphs. To do this, we will need to enlarge slightly the set of activities that a vertex can have. Namely, let
\begin{equation}\label{eq:Lbiplam}
\Lbip(\lambda) := \mathcal{L}(\lambda)\cup \{-2,-1/4\}= \{\lambda,-\lambda-1,-1,1,-2,-1/4\}.
\end{equation}
We will consider the following problems.
 \prob{
$\FactorMVBipHardCore{K}$.} { A bipartite graph $G$ with maximum degree at
most $\Delta$. An activity vector $\lambdab=\{\lambda_v\}_{v\in V}$, such that, for each $v\in V$, $\lambda_v \in \Lbip(\lambda)$.
For every vertex $v\in V$ with $\lambda_v\neq \lambda$, the degree of~$v$ in~$G$ must be at most~$2$.} 
{   If $|Z_G(\lambdab)|=0$  then the algorithm may output any rational number. Otherwise,
 it must output a  rational number $\widehat{N}$ such that
$\widehat{N}/K \leq
|Z_{G}(\lambdab)|\leq K \widehat{N}$. }  
 \prob{
$\ArgMVBipHardCore{\rho}$.} {A bipartite graph $G=(V,E)$ with maximum degree at
most $\Delta$. An activity vector $\lambdab=\{\lambda_v\}_{v\in V}$ such that, for each $v\in V$, $\lambda_v \in \Lbip(\lambda)$.
For every vertex $v\in V$ with $\lambda_v\neq \lambda$, the degree of~$v$ in~$G$ must be at most~$2$.
 }
{ If $Z_G(\lambdab)=0$ then the algorithm may output any rational number. Otherwise, it must output  
a rational number $\widehat{A}$ such that, for some $a\in \arg(Z_{G}(\lambdab))$,
$ |\widehat{A} - a| \leq \rho$.  } 

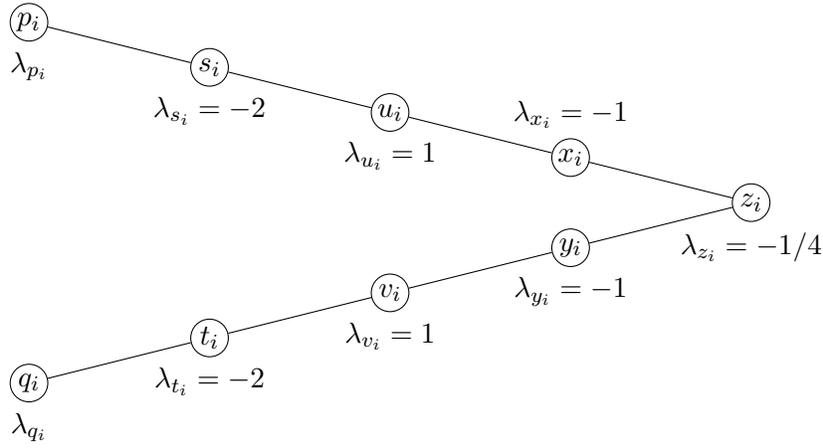
\begin{figure}[h]
\begin{center}
{
\tikzset{lab/.style={circle,draw,inner sep=0pt,fill=none,minimum size=5mm}} 
\begin{tikzpicture}[xscale=0.8,yscale=0.8]
\draw (-6,1.5) node[lab, label={below:$\lambda_{p_i}$}] (m1a) {$p_i$};
\draw (-6,-4.5) node[lab, label={below:$\lambda_{q_i}$}] (m1b) {$q_i$};
\draw (-3,0.75) node[lab, label={below:$\lambda_{s_i}= -2$}] (0a) {$s_i$};
\draw (-3,-3.75) node[lab, label={below:$\lambda_{t_i}= -2$}] (0b) {$t_i$};
\draw (0,0) node[lab, label={below:$\lambda_{u_i}= 1$}] (1) {$u_i$};
\draw (0,-3) node[lab, label={below:$\lambda_{v_{i}}= 1$}] (2) {$v_{i}$};
\draw (6,-1.5) node[lab, label={below:$\lambda_{z_i}=-1/4$}] (5) {$z_i$};
\draw (3,-0.75) node[lab, label={above:$\lambda_{x_i}=-1$}] (3) {$x_i$};
\draw (3,-2.25) node[lab, label={below:$\lambda_{y_i}=-1$}] (4) {$y_i$};
\draw (m1a) -- (0a) -- (1) -- (3) -- (5);
\draw (m1b) -- (0b) -- (2) -- (4) -- (5); 
\end{tikzpicture}
}
\end{center}
\caption{The  gadget $B_i''$  with activity vector $\lambdab''$ used in the proof of Lemma~\ref{lem:bipartitegadget}.   }
\label{fig:bip}
\end{figure} 
\begin{lemma}\label{lem:bipartitegadget}
Let $B_i''$ be the graph in Figure~\ref{fig:bip} with activity vector $\lambdab$, and set $T_i''=\{p_i,q_i\}$. Then, for any value of the activities $\lambda_{p_i},\lambda_{q_i}\in \Complex$, it holds that
$$Z_{B_i,T_i'',\emptyset}(\lambdab'') = 1, \quad Z_{B_i'',T_i'',\{p_i\}}(\lambdab'') =\lambda_{p_i},\quad  Z_{B_i,T_i'',\{q_{i}\}}(\lambdab'')=\lambda_{q_i}, \quad Z_{B_i'',T_i'',T_i''}(\lambdab'')= 0.$$
\end{lemma}
\begin{proof}
Note that $B_i''$ is obtained from the path $B_i'$ in Figure~\ref{fig:equal} by appending the vertices $p_i,q_i$, and in turn $B_i'$ is obtained from the path $B_i$ in Figure~\ref{fig:equal} by appending the vertices $s_i,t_i$. As in Lemma~\ref{lem:equalitygadget}, we denote $T_i'=\{s_i,t_i\}$ and $T_i=\{u_i,v_i\}$. Completely analogously to \eqref{eq:45g44r2}, we have  
\begin{equation}\label{eq:45g44r2b}
\begin{aligned}
Z_{B_i'',T_i'',\emptyset}(\lambdab'')& = Z_{B_i',T_i',\emptyset}(\lambdab'')+Z_{B_i',T_i',\{s_i\}}(\lambdab'')+Z_{B_i',T_i',\{t_i\}}(\lambdab'')+Z_{B_i',T_i',T_i'}(\lambdab''),\\
Z_{B_i'',T_i'',\{p_i\}}(\lambdab'')&=\lambda_{p_i}\big(Z_{B_i',T_i',\emptyset}(\lambdab'')+Z_{B_i',T_i',\{t_i\}}(\lambdab'')\big),\\
Z_{B_i'',T_i'',\{q_{i}\}}(\lambdab'')&= \lambda_{q_i}\big(Z_{B_i',T_i',\emptyset}(\lambdab'')+Z_{B_i',T_i',\{s_i\}}(\lambdab'')\big),\\
Z_{B_i'',T_i'',T_i''}(\lambdab'')&=\lambda_{p_i}\lambda_{q_i}Z_{B_i',T_i',\emptyset}(\lambdab').
\end{aligned} 
\end{equation}
To compute the r.h.s. in \eqref{eq:45g44r2b} we will use \eqref{eq:45g44r2}  which expresses the relevant quantities in terms of the gadget  $B_i$ in Figure~\ref{fig:equal}. Using \eqref{eq:tv465b6b4}  and \eqref{eq:45g44r2}, we therefore obtain
\begin{align*}
&Z_{B_i,T_i,\emptyset}(\lambdab'')=-1/4,& &Z_{B_i,T_i,\{u_i\}}(\lambdab'')=Z_{B_i,T_i,\{v_i\}}(\lambdab'')=-1/4,& Z_{B_i,T_i,T_i}(\lambdab'')=3/4,\\
&Z_{B_i',T_i',\emptyset}(\lambdab'')=0,& &Z_{B_i',T_i',\{s_i\}}(\lambdab')=Z_{B_i',T_i',\{t_i\}}(\lambdab'')=1,& Z_{B_i',T_i',T_i'}(\lambdab'')=-1.
\end{align*}
Plugging this into \eqref{eq:45g44r2b} concludes the proof of the lemma.
\end{proof}

\begin{theorem}\label{thm:BipMV}
Let $\Delta\geq 3$ and  
$\lambda\in \CQ$ be a complex number such that
$\lambda\not\in (\LambdaD \cup \Reals_{\geq 0})$.  

Then, for $ K  =1.02$,
$\FactorMVBipHardCore{K}$ is $\numP$-hard. Also, for $\rho = 9\pi/24$, 
$\ArgMVBipHardCore{\rho}$ is $\numP$-hard.
\end{theorem}
\begin{proof}
Let $G=(V,E)$ be an $n$-vertex graph with maximum degree at most $\Delta$. Suppose that $\lambdab=\{\lambda_v\}_{v\in V}$ is an activity vector for~$G$ such that,
for each $v\in V$, $\lambda_v \in \mathcal{L}(\lambda)$. Let $e_1,\hdots, e_m$ be an arbitrary enumeration of the edges of $G$ and suppose that $e_i=\{p_i,q_i\}$, where $p_i$ and $q_i$ are vertices of $G$.  Let $H$ be the bipartite graph constructed from~$G$ by  
replacing every edge $e_i$ of $G$ with the gadget $B_i''$ from Figure~\ref{fig:bip} and denote by $\lambdab'$ the resulting activity vector on $H$ (every vertex originally in $G$ retains its activity in $H$). Observe that every vertex activity in $H$ is from the set $\Lbip(\lambda)$. Moreover, $Z_G(\lambdab)= Z_H(\lambdab')$. The result therefore follows from Theorem~\ref{thm:MV}.
\end{proof}

We will need the following lower bound on $Z_G(\lambdab)$, which follows from Lemma~\ref{lem:use}. 
\begin{lemma}
\label{lem:LB}
Suppose that $\lambda\in \CQ$.
Then,  there exists a  rational $C_\lambda>1$ such that the following holds. For any $n$-vertex graph $G=(V,E)$
and any activity vector $\lambdab=\{\lambda_v\}_{v\in V}$ such that $\lambda_v\in \Lbip(\lambda)$ for all $v\in V$, it holds that either $Z_G(\lambdab)=0$ or
$|Z_G(\lambdab)| > C_\lambda^{-n}$.
\end{lemma}

\begin{proof} 
Let $\lambda_1=-\lambda-1$, $\lambda_2=-1$, $\lambda_3=1$, $\lambda_4=-2$ and $\lambda_5=-1/4$, so that $\Lbip(\lambda)=\{\lambda,\lambda_1,\hdots,\lambda_5\}$.

Let $\{U_0,U_1,\hdots,U_5\}$ be a partition of $V$
such that $\lambda_v$ is equal to $\lambda$ if $v\in U_0$ and, for $i=1,\hdots,5$, $\lambda_v=\lambda_i$ if $v\in U_i$. For an independent set $I\in \calI_G$ and $i=0,\hdots,5$, let $n_i(I) = |I \cap U_i|$ and $n_{-}(I)=n_1(I)+n_2(I)+n_4(I)+n_5(I)$. Then,
\begin{align*}
Z_G(\lambdab)&= \sum_{I\in \calI_G} \lambda^{n_0(I)}\prod^5_{i=1}\lambda^{n_i(I)}_i=\sum_{I\in \calI_G} (-1)^{n_{-}(I)}2^{n_4(I)}(1/4)^{n_5(I)}\lambda^{n_0(I)}(1+\lambda)^{n_1(I)}\\
&=\sum_{I\in \calI_G} (-1)^{n_{-}(I)}2^{n_4(I)}(1/4)^{n_5(I)}\lambda^{n_0(I)}\sum^{n_1(I)}_{k=0}\binom{n_1(I)}{k}\lambda^k.
\end{align*}
Thus, we have that $4^n Z_G(\lambdab)$ is an integer  polynomial of $\lambda$. Moreover, observe that the absolute values of the coefficients of $4^n Z_G(\lambdab)$ corresponding to an independent set $I\in\calI_G$ sum  to at most $4^n\cdot 2^n=8^n$.   Since $|\calI_G | \leq 2^n$, we have that the absolute values of the coefficients of $4^n Z_G(\lambdab)$ sum to at most $2^n \cdot 8^n=16^n$. The result now follows by applying Liouville's inequality (cf. Lemma~\ref{lem:use}).
\end{proof}

\subsection{Reduction from the multivariate problem to the single-activity problem}\label{sec:MVtoSV}

 The purpose of this section is to prove the following theorem.

\begin{theorem}\label{thm:MVtoSV}
Let $\Delta\geq 3$ and  
$\lambda\in \CQ$ be a complex number such that
$\lambda\not\in (\LambdaD \cup \Reals_{\geq 0})$. 
Then  
there is a polynomial-time Turing reduction
from the problem $\FactorMVBipHardCore{1.02}$ to
$\FactorHardCore{1.01}$.  
There is also a polynomial-time Turing reduction from the problem
$\ArgMVBipHardCore{9\pi/24}$ to the problem $\ArgHardCore{\pi/3}$.
\end{theorem}
 
\begin{proof}
Let $\lambda_1=-\lambda-1$, $\lambda_2=-1$, $\lambda_3=1$, $\lambda_4=-2$, $\lambda_5=-1/4$ so that $\Lbip(\lambda)=\{\lambda,\lambda_1,\hdots,\lambda_5\}$.
Let $M>1$ be a rational number such that $M>\max\{|\lambda|,|\lambda_1|, \hdots, |\lambda_5|\}$ and let $C_\lambda>1$ be the rational in Lemma~\ref{lem:LB}.

Let $G=(V,E)$ be an $n$-vertex bipartite graph with maximum degree at
most $\Delta$. 
Suppose that $\lambdab=\{\lambda_v\}_{v\in V}$ is an activity vector for~$G$ such that,
for each $v\in V$, $\lambda_v \in \Lbip(\lambda)$. Let $\{U_0,U_1,\hdots,U_5\}$ be a partition of $V$
such that $\lambda_v$ is equal to $\lambda$ if $v\in U_0$ and, for $i=1,\hdots,5$, $\lambda_v=\lambda_i$ if $v\in U_i$. Suppose further that, for every vertex $v\in \bigcup^5_{i=1}U_i$,  the degree of~$v$ in~$G$  is at most~$2$.
 
Let $\epsilon:=\frac{1}{10^{4}n(4M C_\lambda)^{n}}$. For $i=1,\hdots,5$, let 
$J_i$ be a bipartite graph with maximum degree at most~$\Delta$ with terminal~$v_i$ 
that implements $\lambda_i$ with accuracy~$\epsilon$.
Proposition~\ref{lem:complexexp} and Proposition~\ref{lem:realexp} guarantee that $J_i$ exists, and
that it can be constructed in  time $poly(\size{\eps})$.
The propositions also guarantee that $\Zout_{J_i,v_i}(\lambda)\neq 0$.
Let $\lambda_i' = \Zin_{J_i,v_i}(\lambda)/\Zout_{J_i,v_i}(\lambda)$, so that  $|\lambda_i' - \lambda_i| \leq \epsilon$. Note, we have the crude bound $|\lambda_i'|\leq 2M$ for all $i=1,\hdots,5$.

Let $\lambdab'$ be the activity vector for~$G$ formed from~$\lambdab$ by replacing 
every instance of $\lambda_i$ with $\lambda'_i$ for $i=1,\hdots,5$. Let $H$ be the bipartite graph constructed from~$G$ by  
replacing, for each $i\in[5]$,
every vertex $v\in  U_i$   with
a (distinct) copy of $J_i$, relabelling the terminal $v_i$ as $v$  and attaching the terminal to the (at most two) neighbours of~$v$ in~$G$ (note that this is the same construction as the one in Lemma~\ref{lem:transf}). Note that the maximum degree of~$H$ is at most~$\Delta$ and, by Lemma~\ref{lem:transf},
\begin{equation}
\label{eq:GH}
Z_{G}(\lambdab') = Z_{H}(\lambda)/C_H, \mbox{ where } C_H:=\prod^5_{i=1} \prod_{v\in U_i} \Zout_{J_i,v_i}(\lambda).
\end{equation}
Note that using the output of the algorithm provided by Proposition~\ref{lem:realexp} and Proposition~\ref{lem:complexexp}, we can compute $C_H$ exactly
in time $poly(\size{\eps})$. 

We will show that, whenever $Z_G(\lambdab)\neq 0$, it holds that
\begin{equation}
\label{finish:one}
\frac{1.01}{1.015} 
|Z_G(\lambdab')|
\leq 
|Z_G(\lambdab)|
\leq \frac{1.015}{1.01}
|Z_G(\lambdab')|
\end{equation}
and that
\begin{equation}
\label{finish:two}
\mbox{there are 
$a \in \arg(Z_G(\lambdab))$
and 
$a' \in \arg(Z_G(\lambdab'))$
such that $|a-a'|
  \leq \pi/30$.}
\end{equation}

Before proving  \eqref{finish:one} and \eqref{finish:two}, we show that they give the desired reductions. 
The reduction from $\FactorMVBipHardCore{1.02}$ to
$\FactorHardCore{1.01}$  goes as follows.
Suppose that $(G,\lambdab)$ is an input to 
$\FactorMVBipHardCore{1.02}$ and that
 $Z_G(\lambdab)\neq 0$. By \eqref{finish:one}, we obtain  that $Z_G(\lambdab')$ is non-zero and  therefore, by \eqref{eq:GH}, $|Z_H(\lambda)| \neq 0$ as well. Thus,
 an oracle for 
$\FactorHardCore{1.01}$ 
with input~$H$
gives an approximation $\hat{N}$
so that 
$\hat{N}/1.01 \leq |Z_H(\lambda)| \leq 1.01 \hat{N}$. By \eqref{eq:GH} and \eqref{finish:one}, we therefore obtain that
$$
\frac{\hat{N} }{1.015 |C_H|}
\leq |Z_G(\lambdab)| 
\leq 1.015 \frac{\hat{N}}{|C_H|}.$$
As we noted earlier, 
  the value $C_H$ can be computed exactly 
 in time $poly(\size{\eps})$.  
 Thus, it is easy, in time $poly(\size{\eps})$,
 to compute a value $\hat{C}$ such that
 $$
\frac{1.015}{1.02}\hat{C} \leq  {|C_H|} \leq \frac{1.02  }{1.015 }\hat{C}.$$
Thus, the algorithm for $\FactorMVBipHardCore{1.02}$
can return $\hat{N}/\hat{C}$ as an output.

For the reduction from 
$\ArgMVHardCore{9\pi/24}$ to 
$\ArgHardCore{\pi/3}$, suppose   that $(G,\lambdab)$ is an input to 
 $\ArgMVHardCore{9\pi/24}$.
 An oracle call to  $\ArgHardCore{\pi/3}$
 with input $H$ gives a value $\hat{A}$ such that, for some $h\in \arg(Z_H(\lambda))$,
   $|\hat{A}-h| \leq \pi/3$.
Consider $a$ and $a'$ from \eqref{finish:two}.
 By~\eqref{eq:GH} there is a 
 $c\in \arg(C_H)$ such that $a'=h-c$.
 Now, by the triangle inequality, 
 $$|a - (\hat{A}-c)| \leq |a-a'| + |a'-(h-c)| + |h-\hat{A}| \leq \pi/30 + 0 +  \pi/3 =  11 \pi/30.$$
 Adding
 an integer multiple of $2  \pi$ to both $a$ and $c$ on the
 left-hand-side, we conclude that 
 for every $\tilde{c}\in \arg(C_H)$
 there exists an $\tilde{a} \in \arg(Z_G(\lambdab))$
 such that
 $|\tilde{a} - (\hat{A} - \tilde{c})| \leq 11 \pi/30$. 
 In particular,
 taking $\tilde{c} = \Arg(C_H)$,
  there exists an $\tilde{a} \in \arg(Z_G(\lambdab))$
 such that
 $|\tilde{a} - (\hat{A} -  \Arg(C_H))| \leq 11 \pi/30$.  
 Thus, 
   the algorithm for   $\ArgMVHardCore{9\pi/24}$   can compute a value $\hat{C}$ such that
 $|\hat{C} - \Arg(C_H)| \leq 9 \pi/24 - 11\pi/30$
 and return the value $\hat{A} - \hat{C}$ as output.
    
So in the rest of the proof, we will establish \eqref{finish:one} and \eqref{finish:two}. First, we show that, whenever $Z_G(\lambdab)\neq 0$, it holds that
\begin{equation}\label{eq:difZllp34}
|Z_G(\lambdab) - Z_G(\lambdab')|\leq \frac{1}{10^4(C_\lambda)^n}\leq \frac{|Z_G(\lambdab)|}{10^4}.
\end{equation}
The rightmost inequality is immediate by Lemma~\ref{lem:LB}. For the leftmost inequality, note that for all positive integers $k$ and arbitrary complex numbers $x_1,y_1,\hdots,x_k,y_k$ we have the telescopic expansion $\prod^{k}_{i=1}x_i-\prod^{k}_{i=1}y_i=\sum^{k}_{j=1}(x_{j}-y_{j})\prod^{j-1}_{i=1}x_i\prod^{k}_{i=j+1}y_i$. Hence, for an arbitrary independent set $I\in \mathcal{I}_G$,  we have that (using  the crude bounds $|\lambda_i|,|\lambda_i'|<2M$ and $|\lambda|<2M$)
\[\Big|\prod_{v\in I}\lambda_v-\prod_{v\in I}\lambda_v'\Big|\leq \sum_{v\in I}(2M)^{|I|-1}|\lambda_v-\lambda_v'|\leq n (2M)^{n} \epsilon,\]
where in the last inequality we used that $|\lambda_v-\lambda_v'|\leq \epsilon$ for all $v\in V$. Since $|\mathcal{I}_G|\leq 2^n$ and $\epsilon=\frac{1}{10^{4}n(4M C_\lambda)^{n}}$, we obtain \eqref{eq:difZllp34}.

Now we are ready to show \eqref{finish:one} and \eqref{finish:two} whenever $Z_G(\lambdab)\neq 0$.
In particular, by the triangle inequality and \eqref{eq:difZllp34}, we have that
\[
\bigg|\frac{|Z_G(\lambdab')|}
{|Z_G(\lambdab)|} -1\bigg|=\frac{\big||Z_G(\lambdab')|-|Z_G(\lambdab)|\big|}
{|Z_G(\lambdab)|}\leq \frac{|Z_G(\lambdab')-Z_G(\lambdab)|}{|Z_G(\lambdab)|}\leq 10^{-4},
\]
which proves \eqref{finish:one}. In fact, using the inequality above and \eqref{finish:one}, it follows that the Ziv distance between $Z_G(\lambdab)$ and $Z_G(\lambdab')$ is at most $10^{-3}$, and therefore \eqref{finish:two} follows from  \cite[Lemma 2.1]{ComplexIsing}. \end{proof}

\section{Proof of  Proposition~\ref{lem:main121}}\label{sec:mainpart}

In this section, we prove  Proposition~\ref{lem:main121}. To do this, we will focus on understanding the following type of sequences.
\begin{definition}\label{def:generates}
Let $\Delta\geq 3$ and $\lambda\in \CQ\setminus\Reals$, and set $d:=\Delta-1$. A hard-core-program is a sequence $a_0, a_1,\hdots$, starting with $a_0=1$,  and
\begin{equation}\label{eq:4f4m7ujf43f}
a_{k} = \frac{1}{1+\lambda a_{i_{k,1}}\cdots a_{i_{k,d}}}\quad\mbox{for $k\geq1$, }
\end{equation} where $i_{k,1},\dots,i_{k,d}\in\{0,\dots,k-1\}$. We say that the hard-core-program generates $x\in \Complex$ if there exists integer $k\geq 0$ such that $a_k=x$.
\end{definition}
Our interest in hard-core-programs is justified by the following lemma.
\begin{lemma}\label{lem:hard-core-imp}
Let $\lambda\in \Complex$ and $d\geq 2$. Suppose that $a_0,a_1,\hdots$ is a hard-core program. Then, for every $k\geq 1$, there exists a tree of maximum degree at most $\Delta=d+1$ that implements the activity $\lambda a_k$.
\end{lemma}
\begin{proof}
The proof is by induction on $k$. For $k=1$, we have that $i_{0,1}=\hdots,i_{0,d}=0$, so $a_1=1/(1+\lambda)$. Let $T$ be the single-edge tree $\{u,v\}$. Then, we have that
\[\Zin_{T,v}(\lambda)=\lambda, \quad \Zout_{T,v}(\lambda)=1+\lambda,\]
and therefore $T$ with terminal $v$ implements the activity $\frac{\Zin_{T,v}(\lambda)}{\Zout_{T,v}(\lambda)}=\lambda a_1$, as wanted.

Suppose that the statement is true for all values $\leq k$ and suppose that 
\[a_{k+1}=\frac{1}{1+\lambda a_{i_{k,1}}\cdots a_{i_{k,d}}}\]
for some $i_{k,1},\dots,i_{k,d}\in\{0,\dots,k-1\}$. Let 
$J=\{j\in [d] \mid i_{k,j}\neq 0\}$ and note that for every $j\in [d]\backslash J$, it holds that $a_{i_{k,j}}=1$. By the induction hypothesis, for every $j\in J$, there exists a tree $T_j$ of maximum degree at most $\Delta$, with terminal $v_j$, such that $\frac{\Zin_{T,v}(\lambda)}{\Zout_{T,v}(\lambda)}=\lambda a_{i_{k,j}}$. Let $T$ be the tree obtained by taking the disjoint union of the trees $T_j$ with $j\in J$ and identifying all the terminals $v_j$ into a new vertex $v$. Then,
\begin{equation}\label{eq:45t5v3v5g5}
\Zin_{T,v}(\lambda)=\lambda\prod_{j\in J}\frac{\Zin_{T_j,v_j}(\lambda)}{\lambda}, \quad \Zout_{T,v}(\lambda)=\prod_{j\in J}\Zout_{T_j,v_j}(\lambda).
\end{equation}
Now consider the tree $T'$ obtained from $T$ by adding a new vertex $u$ whose single neighbour is the vertex $v$. Then,
\[\Zin_{T',u}(\lambda)=\lambda \Zout_{T,v}(\lambda), \quad \Zout_{T',u}(\lambda)=Z_{T}(\lambda).\]
Using this and \eqref{eq:45t5v3v5g5}, we therefore obtain that
\[\frac{\Zin_{T',u}(\lambda)}{\Zout_{T',u}(\lambda)}=\lambda\frac{\Zout_{T,v}(\lambda)}{Z_{T}(\lambda)}=\lambda\frac{1}{1+\frac{\Zin_{T,v}(\lambda)}{\Zout_{T,v}(\lambda)}}=\frac{\lambda}{1+\lambda \prod_{j\in [J]}a_{i_k,j}}=\frac{\lambda}{1+\lambda \prod_{j\in [d]}a_{i_k,j}}=\lambda a_{\ell+1},\]
where in the second to last equality we used that $a_{i_{k,j}}=1$ for $j\in [d]\backslash J$.
\end{proof}

\subsection{Getting close to a repelling fixpoint}\label{sec:f4ggefr}

In this section, we will show how to generate points that are arbitrarily close to a fixpoint of the function $f(x)=\frac{1}{1+\lambda x^{\Delta-1}}$ using a hard-core-program. We start with the following lemma.

\begin{lemma}\label{lele1}
Let $\Delta\geq 3$ and set $d:=\Delta-1$. Let $p_k$ be a polynomial in $\lambda$ defined by
\begin{equation*}
p_0=p_1=\dots=p_d=1, \mbox{ and } p_{k+1} = p_k + \lambda p_{k-d} \mbox{\ \ for $k\geq d$}.
\end{equation*}
Then, for all $k\geq 0$, all roots of $p_k$ are real.
\end{lemma}
\begin{proof}
Consider a graph $G_k=(V,E)$ with $V=\{d+1,\dots,k\}$ with $i,j$ are connected if $|i-j|\leq d$.
We will show, by induction, that the independent set polynomial of $G_k$ is $p_k$. The claim is true
for $k=0,\dots,d$ (since the graph is empty and $p_k=1$). For $k+1\geq d+1$ the claim follows by induction:
\begin{itemize}
\item the contribution of the independent sets with $k+1$ included is $\lambda$ times the independence polynomial
of $G_{k-d}$, which, by the inductive hypothesis, is $\lambda p_{k-d}$,
\item the contribution of the independent sets with $k+1$ not included is the independent set polynomial of $G_{k}$,
which, by the inductive hypothesis, is $p_{k}$.
\end{itemize}
Hence, the independent set polynomial of $G_{k+1}$ is $p_{k}+\lambda p_{k-d}=p_{k+1}$.

A claw is a graph with $4$ vertices $a,b_1,b_2,b_3$ and $3$ edges $ab_1,ab_2,ab_3$. A claw-free graph is a graph
that does not contain a claw as an induced subgraph. We will show that $G_k$ is claw-free.
Suppose $a,b_1,b_2,b_3\in\{d+1,\dots,k\}$ form a claw; w.l.o.g. $b_1<b_2<b_3$. Then $|a-b_i|\leq d$ for $i=1,2,3$.
If $b_2<a$ then $b_2-b_1\leq a-b_1\leq d$ and hence $b_1b_2$ is an edge---a contradiction with the
assumption that $a,b_1,b_2,b_3$ form a claw. If $b_2>a$ then $b_3-b_2\leq b_3-a\leq d$ and hence $b_2b_3$ is an edge---a contradiction with the
assumption that $a,b_1,b_2,b_3$ form a claw. Thus $G_k$ is claw-free.

Now we use \cite[Theorem 1.1]{DBLP:journals/jct/ChudnovskyS07} which states that the roots of
the independent set polynomial of a claw-free graph are all real.
\end{proof}

We will now show that we can get close to the fixpoint of $f$ with the smallest norm.
\begin{lemhoptwo}
Let $\Delta\geq 3$ and $\lambda\in \Complex\setminus\Reals$, and set $d:=\Delta-1$. Let $\omega$ be the fixpoint of $f(x)=\frac{1}{1+\lambda x^{d}}$ with the smallest norm. For $k\geq 0$, let $x_k$ be the sequence defined by
\begin{equation}\label{rec}
x_0=x_1=\dots=x_{d-1}=1, \quad x_k = \frac{1}{1+\lambda \prod_{i=1}^d x_{k-i}}\quad\mbox{for $k\geq d$}.
\end{equation}
Then, the sequence $x_k$ is well-defined (i.e., the denominator of \eqref{rec} is nonzero for all $k\geq d$) and converges to the fixpoint $\omega$ as $k\rightarrow \infty$. Moreover, there exist infinitely many $k$ such that $x_k\neq \omega$.
\end{lemhoptwo}
\begin{proof}
For $k\geq 0$, let $R_k$ be the sequence defined by
\begin{equation}\label{zzz}
R_0=R_1=\dots=R_d=1, \quad R_{k+1} = R_k + \lambda R_{k-d} \mbox{ for $k\geq d$}.
\end{equation}
Observe that $R_k=p_k(\lambda)$, where $p_k$ is the polynomial in Lemma~\ref{lele1}. Since $\lambda\in\Complex\backslash\Reals$, Lemma~\ref{lele1} implies  that  $R_k\neq 0$ for all $k$
and hence for $k\geq 0$ we can let
\begin{equation}\label{zio}
y_k=R_k/R_{k+1}.
\end{equation}
Note that for $k\in\{0,\dots,d-1\}$ we have $y_k = x_k$. For $k\geq d$ we have
$$
y_{k} = \frac{R_k}{R_{k+1}} = \frac{R_k}{R_k + \lambda R_{k-d}} =  \frac{1}{1 + \lambda \frac{R_{k-d}}{R_{k}}}
= \frac{1}{1+\lambda\prod_{j=1}^d \frac{R_{k-j}}{R_{k-j+1}}}
= \frac{1}{1+\lambda\prod_{j=1}^d y_{k-j}},
$$
and hence, by induction, $x_k = y_k$ for all $k$. It follows that the sequence $x_k$ is well-defined.

Let $\omega_1,\dots,\omega_{d+1}$ be the
fixpoints of $f(x)=\frac{1}{1+\lambda x^d}$  sorted in increasing order of their absolute value, so that $\omega_1=\omega$. Note, since $\lambda\in\Complex\backslash\Reals$, by Lemma~\ref{lem:4frf46} we have that $|\omega_i|\neq |\omega_j|$ for different $i,j\in[d+1]$.  Observe that the characteristic polynomial of the recurrence \eqref{zzz} is $z^{d+1}-z^d-\lambda$ and that the roots of the polynomial are $1/\omega_1,\hdots,1/\omega_{d+1}$ (to verify this, use that $\lambda \omega_j^{d+1}+\omega_j-1=0$ for $j \in [d+1]$). Therefore, from the theory of linear recurrences we have that there exist $\alpha_1,\dots,\alpha_{d+1}\in \Complex$ such that for all $k\geq 0$
$$
R_k = \sum_{j=1}^{d+1} \alpha_j(1/\omega_j)^k.
$$
Note that since $R_0=R_1=\dots=R_d=1$ we have that $\alpha_1,\dots,\alpha_{d+1}$ is the solution of the following
(Vandermonde) system
$$
\sum_{j=1}^{d+1} (1/\omega_j)^k \alpha_j = 1 \quad \mbox{for}\ k\in\{0,\dots,d\}.
$$
Suppose one of the $\alpha_1,\dots,\alpha_{d+1}$ is zero, say $\alpha_{i}=0$ for some $i\in[d+1]$. Let $\alpha_{d+2} = -1$ and $\omega_{d+2}=1$.
For $j\in [d+1]$, note that $\omega_j\neq 1$  (since that would imply $\lambda = 0$) and therefore $\omega_j\neq \omega_{d+2}$.  Then we have that
$\{\alpha_j\}_{j\in [d+2]\backslash\{i\}}$ is a non-zero solution of the following (again Vandermonde) system
\begin{equation}\label{vander}
\sum_{j\in [d+2]\backslash\{i\}} (1/\omega_j)^k \alpha_j = 0 \quad \mbox{for}\ k\in\{0,\dots,d\}.
\end{equation}
This is a contradiction, since the system only has a zero solution (since $\omega_1,\dots,\omega_{d+1},\omega_{d+2}$ are distinct). Thus none of the $\alpha_1,\dots,\alpha_{d+1}$ is zero and, in particular, $\alpha_1\neq 0$. It follows that $x_k=R_k/R_{k+1}$ converges to $\omega_1=\omega$ as $k\rightarrow \infty$.

To finish the proof, it remains to show that $x_k\neq \omega$ for infinitely many $k$. For the sake of contradiction, assume otherwise, and let $k_0$ be the largest integer such that $x_{k_0}\neq \omega$. From \eqref{rec}, we obtain
\[x_{k_0+d}=\frac{1}{1+\lambda \prod^{d-1}_{j=0}x_{k_0+j}}\]
By the choice of $k_0$, we have $x_{k_0+1}=\cdots=x_{k_0+d}=\omega$, which gives that $x_{k_0}=\omega$ (using that $\omega=\frac{1}{1+\lambda \omega^d}$), contradiction. 

This concludes the proof of Lemma~\ref{hop2}.
\end{proof}

Finally, we conclude this section with the following lemma, which will be useful later.
\begin{lemma}\label{uiopa}
For $\lambda\in \Complex\backslash \Reals$ and $d\geq 2$, let $\omega$ be the
fixpoint of $f(x)=\frac{1}{1+\lambda x^d}$  with the smallest norm. Then, $\omega\in \Complex\backslash \Reals$ and $0<|\omega-1|<1$.
\end{lemma}
\begin{proof}
Since $1-\omega = \lambda\omega^{d+1}$, we have that $|\omega-1|>0$ (otherwise, $\lambda=0$) and $\omega\in \Complex\backslash \Reals$ (otherwise, $\lambda\in \Reals$). We therefore focus on showing that $|\omega-1|<1$. 

Let $\omega_1,\dots,\omega_{d+1}$ be the
fixpoints of $f$  sorted in increasing order of their norm, so that $\omega_1=\omega$. By Lemma~\ref{lem:4frf46}, we have $|\omega_1|<\cdots<|\omega_{d+1}|$, so by $1-\omega_j = \lambda\omega_j^{d+1}$ ($j\in [d]$),  we obtain that
\begin{equation}\label{aa}
|1-\omega_1|<|1-\omega_2|<\dots<|1-\omega_{d+1}|.
\end{equation}
Note that $\omega_1-1,\dots,\omega_{d+1}-1$ are roots of $\lambda(y+1)^{d+1} + y = 0$;  the coefficient of $y^{d+1}$ and the coefficient of $y^0$ are both equal to $\lambda$, so by Vieta's formula, 
\begin{equation}\label{aaww}
\prod_{j=1}^{d+1} (1-\omega_j)=1.
\end{equation}
Equations \eqref{aa} and \eqref{aaww} imply $|1-\omega_1|<1$, as needed.
\end{proof}

\subsection{Bootstrapping a point next to a fixpoint to an $\eps$-covering}\label{sec:45g45g4b}

Once we have the ability to obtain points close to the fixpoint $\omega$, we proceed to the next step, which is creating a moderately dense set of points around $\omega$ (cf. Lemma~\ref{opop2}).

The main idea of the proof of Lemma~\ref{opop2} is that close to the  fixpoint the recurrence implements with a small error any polynomial with non-negative integer coefficients (evaluated at $\omega-1$). Then we use the fact that the values of these polynomials yield a dense set of points in $\Complex$. Before proceeding with the proof of Lemma~\ref{opop2} we state these ingredients formally.

\begin{lemma}\label{polydense}
Let $z\in\Complex\setminus\Reals$ be such that $|z|<1$. Let $S$ be the set of
values of polynomials with non-negative integer coefficients evaluated at $z$ (that is,
$S=\{ p(z)\,|\, p\in{\mathbb Z_{\geq 0}}[x]\}$). Then $S$ is dense in ${\mathbb C}$.
\end{lemma}

\begin{proof}
We can write $z = |z| \emm^{2\pi \im x}$ for some $x\in [0,1)$. Note that $x\neq 0$ and $x\neq 1/2$
(since we assumed $z\not\in\Reals$).

First suppose that $x$ is rational, that is, $x=p/q$ for integer co-prime $p,q$, where $q\geq 3$.
For any $k\in\{0,\dots,q-1\}$ we can obtain an arbitrarily small number on the ray with angle $2\pi k/q$
(by taking $(z^{p^{-1}\mod q})^{k+q\ell}$ for large $\ell$)
and hence we have a dense set of points on the ray (taking integer multiples of the small number). Now we show how using the points on the rays we
obtain a point arbitrarily close to any complex number $t\in\Complex$. First, we can write $t$ as a convex combination of
points on the rays, that is, $t=\alpha_0 r_0 + \dots + \alpha_{q-1} r_{q-1}$ where $r_k$ is
on the ray with angle $2\pi k/q$ (for $k\in\{0,\dots,q-1\}$), $\alpha_k\in [0,1]$ (for $k\in\{0,\dots,q-1\}$)
and  $\sum_{k=0}^{q-1}\alpha_k = 1$. For $\epsilon>0$, let $\hat{\alpha}_k$ be a rational such that
$|\hat{\alpha}_k-\alpha_k|\leq\eps/q$ and let $w$ be the product
of the denominators of $\hat{\alpha}_0,\dots,\hat{\alpha}_{q-1}$.
Since we have a dense set of points on each ray we can obtain $\hat{r_k}$ on the ray with angle $2\pi k/q$
such that $|r_k/w - \hat{r_k}|\leq\eps/w$ (for $k\in\{0,\dots,q-1\})$. Now we argue that
$\sum_{k=0}^{q-1} (w\hat{\alpha}_k) \hat{r}_k$ is close to $t$. We have
\begin{equation}
\begin{split}
\bigg|t- \sum_{k=0}^{q-1} (w\hat{\alpha}_k) \hat{r}_k\bigg| =
\bigg|\sum_{k=0}^{q-1} (\alpha_k-\hat{\alpha}_k) r_k + \sum_{k=0}^{q-1} \hat{\alpha}_k (r_k - w\hat{r}_k) \bigg|
\leq \eps\sum_{k=0}^{q-1} |r_k| + (1+\eps)\eps.
\end{split}
\end{equation}
Taking $\eps$ sufficiently small we get a point arbitrarily close to $t$.

Now suppose $x$ is irrational. Let $t=|t|\emm^{2\pi \im y}$ be a complex number where $y\in [0,1)$.
The fractional part $\{kx\}$ of $kx$ for positive integers $k$ is dense in $[0,1)$ and hence, for $\epsilon>0$, there exists
$k$ such that $|\{kx\} - y|\leq\eps$ and $|z|^k\leq\eps$. Let $w=\lfloor |t|/|z^k| \rfloor$ so that $0\leq |t|-w |z|^k\leq \eps$. Observe that $wz^k=w|z|^k\emm^{2\pi \im \{kx\}}$, so by the triangle inequality
\begin{equation*}
|t - w z^k| \leq \big|t - |t|\emm^{2\pi \im kx}\big|+ \big||t|\emm^{2\pi \im kx}-wz^k\big|=|t|\, | \emm^{2\pi i (y-\{kx\})}-1|+ \big|  |t|- w |z|^k\big| \leq \eps ( 1+ 2\pi |t| ),
\end{equation*}
where in the last inequality we used that for $\theta=2\pi(y-\{kx\})$ it holds that \[| \emm^{i \theta}-1|=\sqrt{(\cos \theta-1)^2+(\sin \theta)^2}=\sqrt{2-2\cos \theta}=2|\sin(\theta/2)|= 2\sin|\theta/2|\leq  |\theta|\leq 2\pi\eps.\] Taking $\eps$ sufficiently small we get a point ($w z^k$) arbitrarily close to $t$.
\end{proof}

\begin{remark}
Note that the assumption $|z|<1$ is necessary---the lemma would be false for, e.g., $z=i$.
\end{remark}

The following operation (as we will prove in Lemma~\ref{apx} below) is a first-order approximation of operation \eqref{eq:4f4m7ujf43f} when applied to points around $\omega$ perturbed by $a_1,\dots,a_d$:
\begin{equation}\label{enq1}
(a_1,\dots,a_d)\mapsto z\sum_{i=1}^d a_i.
\end{equation}

We will use a sequence of~\eqref{enq1} to implement polynomials with non-negative integer
coefficients; the complexity of a polynomial will be the number of steps in the sequence.

\begin{definition}
A straight-line-program with operation~\eqref{enq1} is a sequence of assignments starting with $a_0 = 0$, $a_1 = 1$, and
$$
a_k = z \left(a_{i_{k,1}}+\cdots+a_{i_{k,d}}\right),\quad\mbox{for $k=2,3,\dots$},
$$ where $i_{k,1},\dots,i_{k,d}\in\{0,\dots,k-1\}$. We say that the straight-line-program generates $x\in \Complex$ if there exists integer $k\geq 0$ such that $a_k=x$.
\end{definition}

Using a finite sequence of~\eqref{enq1}, we can implement any polynomial with non-negative integer coefficients, up to factors of $z$. More precisely, we have the following.
\begin{lemma}\label{getpo}
Let $p\in {\mathbb Z_{\geq 0}}[z]$ be a polynomial with non-negative integer coefficients.
There exist non-negative integers $k:=k(p)$ and $n:=n(p)$ and a straight-line-program with operation~\eqref{enq1}
such that $a_{k} = z^{n} p(z)$.
\end{lemma}

\begin{proof}
Let $p(z)=\sum_{j=0}^t c_j z^j$, where $c_t\neq 0$ or $t=0$.  We will prove the claim by induction on $t+\sum_{j=0}^t c_j$. The base case is $p(z)\equiv 0$,
here we can take $k(p)=0$ and $n(p)=0$. Now assume $t+\sum_{j=0}^t c_j\geq 1$.

First assume that $c_0\geq 1$. Let $q(z)=p(z)-1$. By the induction hypothesis there exist $n,k$ and a straight-line-program with operation~\eqref{enq1}
such that $a_k = z^{n} q(z)$. Let $a_{k+1} = z a_1$ and $a_{k+j} = z a_{k+j-1}$ for $j=2,\dots,n$ (note that
$a_{k+n} = z^n$). Finally, add $a_{k+n+1} = z (a_k + a_{k+n})$. Note that $a_{k+n+1} = z^{n+1} p(z)$.

Now assume $c_0=0$. Let $q(z)=p(z)/z$. By the induction hypothesis there exist $n,k$ and a straight-line-program with operation~\eqref{enq1}
such that $a_k = z^{n} q(z)$. Let $a_{k+1} = z a_k$. Note that $a_{k+1} = z^n p(z)$.
\end{proof}

From the Taylor expansion we have that close to the fixpoint the multi-variate hard-core
recurrence implements operation~\eqref{enq1} (with a small error).

\begin{lemma}\label{apx} 
Suppose that  $\lambda\in \Complex\backslash\Reals$ and $d\geq 2$. 
Let $\omega$ be the fixpoint of $f(x)=\frac{1}{1+\lambda x^{d}}$ with the smallest norm and set $z=\omega-1$.

There exist reals $C_0:=C_0(\omega,\lambda,d)>1$ and
$\delta_0:=\delta_0(\omega,\lambda,d)>0$ such that for
any $a_1,\dots,a_d\in\Complex$ with $|a_j|\leq\delta_0$ (for $j\in [d])$ we have
\begin{equation}
\frac{1}{1+\lambda \prod_{j=1}^d (\omega + a_j) } = \omega + z \bigg(\sum_{j=1}^d a_j\bigg) + \tau,
\end{equation}
where $|\tau|\leq C_0 \max_{j\in [d]} |a_j|^2$.
\end{lemma}
\begin{proof}
Let $b_1,\hdots,b_d$ be arbitrary complex numbers in the unit disc, i.e., $|b_1|,\hdots,|b_d|\leq 1$. Let
\[F(t) = \frac{1}{G(t)},\quad\mbox{where}\quad G(t)=1+\lambda \prod_{j=1}^d (\omega + b_j t).\]
Then, using that $\omega$ is a fixpoint of $f(x)=\frac{1}{1+\lambda x^{d}}$, we have $G(0)=1/\omega$ and
\begin{equation}\label{zilo}
F'(0) = - \frac{\lambda \left(\sum_{j=1}^d b_j\right) (\omega)^{d-1}}{\left(1+\lambda (\omega)^d\right)^2} = z \sum_{j=1}^d b_j.
\end{equation}
Note that for all $t\in [0,1]$
\begin{equation}\label{hiop}
|G'(t)| = \Big|\lambda\sum_{j=1}^d b_j \prod_{k\neq j} (\omega + b_k t)\Big|\leq d |\lambda| (1+|\omega|)^d.
\end{equation}
Similarly, we have that for all $t\in [0,1]$
\begin{equation*}
|G''(t)| \leq d^2 |\lambda| (1+|\omega|)^d.
\end{equation*}

Let $\delta_0: = \min\big\{\frac{1}{2 |\omega| d |\lambda| (1+|\omega|)^{d}},1\big\}$. Note that~\eqref{hiop} implies that for $t\in (0,\delta_0)$ we have
\[|G(t)| \geq |G(0)| - t d |\lambda| (1+|\omega|)^d \geq \frac{1}{|\omega|} - \delta_0  d |\lambda| (1+|\omega|)^d \geq \frac{1}{2|\omega|}.\]
and hence for $t\in (0, \delta_0)$
\begin{equation}\label{zzzzz1}
|F''(t)| = \Bigg| \frac{2G'(t)^2}{G(t)^3} - \frac{G''(t)}{G(t)^2}\Bigg| \leq 3 (2|\omega|)^3 d^4 |\lambda|^2 (1+|\omega|)^{2d} =: C,
\end{equation}
which implies
\begin{equation}\label{ziuq}
|F(t) - F(0) - t F'(0)| \leq Ct^2.
\end{equation}
Given $a_1,\dots,a_d$ such that $|a_j|\leq\delta_0$, let $t=\max_j |a_j|\in(0,\delta_0)$ and $b_j=a_j/t$ (for $j\in [d]$); note,
$$
\frac{1}{1+\lambda \prod_{j=1}^d (\omega + a_j) } = F(t) \quad\mbox{and}\quad t F'(0) = z \sum_{j=1}^d a_j.
$$
Let $C_0:=\max\{C,2\}>1$. The lemma now follows  from~\eqref{ziuq},~\eqref{zilo} and the fact that $F(0)=\omega$.
\end{proof}

Finally we can prove Lemma~\ref{opop2}, which we restate here for convenience.
\begin{lemopoptwo}
\statelemmaopoptwo
\end{lemopoptwo}
\begin{proof}
Consider arbitrary $\epsilon,\kappa>0$ and let $z:=\omega-1$. By Lemma~\ref{uiopa}, we have that $z\in \Complex\backslash \Reals$ and $0<|z|<1$. Therefore, Lemma~\ref{polydense} gives that polynomials with non-negative integer coefficients evaluated at $z$ are dense in $\Complex$. Hence there exists a finite collection ${\cal F}$ of polynomials with non-negative
integer coefficients whose values at $z$ form an $\eps/2$-covering of the unit disk (to obtain the collection
take a finite $\eps/4$-covering of the unit disk and for every point in the covering, using the density,
get a polynomial whose value at $z$ is at distance at most $\eps/4$ from the point).

By Lemma~\ref{getpo} every polynomial $p$  with non-negative integer coefficients can be generated, up to a factor  $z^{n(p)}$, using $k(p)$ operations~\eqref{enq1}. Let
$$
N:=\max_{p\in{\cal F}} n(p)\quad\mbox{and}\quad K:=N+\max_{p\in{\cal F}} k(p).
$$
For every $p\in{\cal F}$, there is a straight-line-program to generate $z^N p(z)$ using at most $K$ applications of~\eqref{enq1} (the extra $N$
in the definition of $K$ is to allow for extra operations~\eqref{enq1} to move from $z^{n(p)} p(z)$ to $z^N p(z)$).
Let $C_0>1$ and $\delta_0>0$ be the constants from Lemma~\ref{apx}. Let
\begin{equation}\label{dede}
\delta = \min\Big\{\delta_0 / d^K, \frac{\eps |z|^N}{2 Z(C_0,d,K)},\frac{1}{Z(C_0,d,K)},\frac{\kappa}{2|\lambda|} \Big\}, \mbox{ where $Z(C_0,d,k) := C_0 (2d)^{2k}$ for $k\geq 0$}.
\end{equation}
 By Lemma~\ref{hop2}, there is a hard-core-program $x_0,x_1,\hdots$ such that for sufficiently large $m,m'$, it holds that  $0<|x_m-\omega|\leq \delta$ and $|x_{m'}-\omega|\leq |x_m-\omega|^2$. Let $y_1=x_m$ and $y_0=x_{m'}$, so that $0<|y_1-\omega|\leq\delta$ and $|y_0-\omega|\leq |y_1-\omega|^2$. Finally, let 
\begin{equation*}
\mbox{$h:=y_1-\omega$ and set $r:=|hz^N|$.}
\end{equation*}
Since $|z|<1$, we have $r\leq |h|\leq \delta<\kappa/|\lambda|$.

We claim that for any $u$ such that $(\omega+u)\in B(\omega,r)$, there exists a hard-core-program that generates $\omega+u'$ with $|u-u'|\leq \epsilon r$. By Lemma~\ref{lem:hard-core-imp}, this implies that, for $\rho:=|\lambda| r\in(0,\kappa)$, any activity in the ball $B(\lambda \omega,\rho)$ can be implemented with accuracy $\epsilon \rho$ by a tree of maximum degree $\Delta$, as needed.

To obtain the desired hard-core-program, observe first that $u/r$ belongs to the unit disc, so there exists $p\in {\cal F}$ such that 
\begin{equation}\label{eq:g4g435f5g5w}
|p(z)-u/r|\leq\eps/2.
\end{equation}
Moreover, there is a straight-line-program $\hat{y}_0, \hat{y}_1,\hdots,\hat{y}_k$ with $k\leq K$ to generate $z^Np(z)$, i.e., $\hat{y}_0=0$, $\hat{y}_1=1$, and 
\begin{equation}\label{rty}
\hat{y}_\ell = z \left(\hat{y}_{i_{\ell,1}}+\cdots+\hat{y}_{i_{\ell,d}}\right),\quad\mbox{for $\ell=2,3,\dots,k$},
\end{equation}
where $i_{\ell,1},\dots,i_{\ell,d}\in\{0,\dots,\ell-1\}$ and $\hat{y}_k=z^Np(z)$.  Note that for all $\ell\in\{0,\dots,k\}$ we have
by induction (using $|z|<1$ from  Lemma~\ref{uiopa}) that
\begin{equation}\label{yuop}
|\hat{y}_\ell|\leq d^{\ell-1}.
\end{equation}
We will next convert the straight-line-program \eqref{rty} into a hard-core-program. Namely, for $y_0,y_1$ as above,  let $y_2,\dots,y_k$ be given by
\begin{equation}\label{rty2}
y_\ell = \frac{1}{1+\lambda \prod_{j=1}^d y_{i_{\ell,j}}}, \mbox{ for $\ell=2,3,\dots,k$}.
\end{equation}
We will prove that for all $\ell=0,1,\hdots, k$ it holds that
\begin{equation}\label{zhi}
y_\ell = \omega + h \hat{y}_\ell + \tau_\ell,\mbox{ where } |\tau_\ell|\leq |h|^2 Z(C_0,d,\ell).
\end{equation}
Assuming \eqref{zhi} for the moment, let us conclude the proof of the claim by showing that  $y_k$ is at distance $\leq r\eps$ from $\omega+u$ and that it can be generated by a hard-core-program.  In particular, for $\ell=k$, \eqref{zhi} gives  (using that $\hat{y}_k=z^N p(z)$, $|h|=|y_1-\omega|\leq\delta$ and \eqref{dede})
\[y_k=\omega+h z^N p(z) + \tau_k, \mbox{ where $|\tau_k|\leq |h|^2 Z(C_0,d,K)\leq |h| |z|^N\epsilon/2$}.\]
Combining this with \eqref{eq:g4g435f5g5w} and recalling that $r=|h z^N|$, we obtain that
\[|y_k - (\omega+u)| = |hz^N p(z) + \tau_k  - u | \leq  r |p(z) -u/r|+|\tau_k|\leq  r\eps.\]
To finish the proof of the claim, we only need to observe that, since $y_0,y_1$ can be generated using a hard-core-program, we can also generate the sequence $y_0,y_1,\hdots,y_k$ using a hard-core-program.

It remains to prove \eqref{zhi}. We do this by induction. 
Note that $y_0 = \omega + 0\cdot h + \tau_0$ where $|\tau_0|=|y_0-\omega|\leq |h|^2$ and $y_1 = \omega+1\cdot h+0$, covering the base cases of~\eqref{zhi}. For the induction step, assume that for all $0\leq \ell'<\ell$ it holds that
\begin{equation}\label{eq:zhizhizhi}
y_{\ell'} = \omega + h \hat{y}_{\ell'} + \tau_{\ell'},\mbox{ where } |\tau_{\ell'}|\leq |h|^2 Z(C_0,d,\ell').
\end{equation}
For all $\ell\leq K$, we have by \eqref{dede} that $Z(C_0,d,\ell)\leq Z(C_0,d,K)\leq 1/\delta\leq 1/h$ and $|h| d^{\ell - 1}\leq \delta_0/d$. Therefore \eqref{yuop} and \eqref{eq:zhizhizhi} give that, for all $0\leq \ell'<\ell$,
\begin{equation}\label{zhito}
|h \hat{y}_{\ell'} + \tau_{\ell'}| \leq |h| d^{\ell'-1} + |h|^2 Z(C_0,d,\ell') \leq |h| d^{\ell - 1} + |h| \leq \delta_0.
\end{equation}
Hence, we can apply Lemma~\ref{apx} with $a_j$'s of the form $h \hat{y}_{\ell'} + \tau_{\ell'}$ ($\ell'\in \{0,1,\hdots,\ell-1\}$).
From Lemma \ref{apx} and \eqref{zhito} we have
\begin{equation*}
y_\ell = \frac{1}{1+\lambda \prod_{j=1}^d (\omega + h \hat{y}_{i_{\ell,j}} + \tau_{i_{\ell,j}})} =\omega +
h z \Big(\sum_{j=1}^d \hat{y}_{i_{\ell,j}} \Big) + \tau_\ell, \mbox{ where } \tau_\ell = z \sum_{j=1}^d \tau_{i_{\ell,j}} + \tau,
\end{equation*}
and
$$
|\tau|\leq C_0 \max_{\ell'=0,\hdots,\ell-1}|h \hat{y}_{\ell'} + \tau_{\ell'}|^2 \leq C_0 |h|^2 (d^{\ell-1} + 1 )^2 \leq C_0|h|^2 d^{2\ell}.
$$
Thus
$$
|\tau_\ell| \leq d |h|^2 Z(C_0,d,\ell-1) + C_0 |h|^2 d^{2\ell} = |h|^2 C_0 ( d (2d)^{2\ell-2} + d^{2\ell} )
\leq |h|^2 C_0 (2d)^{2\ell}  =  |h|^2 Z(C_0,d,\ell),
$$
competing the induction step. This finishes the proof of \eqref{zhi} and therefore the proof of Lemma~\ref{opop2}.
\end{proof}

\subsection{Bootstraping $\eps$-covering to arbitrary density}\label{sec:tv45vef}

Our next step is to use the ``moderately dense'' set of points in a small disk around $\omega$ to create a
dense set of points.  We first need a few technical results.

\begin{lemma}\label{apxinv} Suppose that  $\lambda\in \Complex\backslash\Reals$ and $d\geq 2$. 
Let $\omega$ be the fixpoint of $f(x)=\frac{1}{1+\lambda x^{d}}$ with the smallest norm and let $z=\omega-1$.

There exist reals $C_1:=C_1(\omega,\lambda,d)>0$ and
$\delta_1:=\delta_1(\omega,\lambda,d)>0$ such that for
any $a_1,\dots,a_{d}\in\Complex$ with $|a_j|\leq\delta_1$ (for $j\in [d])$ we have
\begin{equation}
\frac{1}{\lambda}\Big(\frac{1}{\omega+a_{d}}-1\Big) \prod_{j=1}^{d-1} (\omega + a_k)^{-1} = \omega + \left( \frac{a_d}{z} - \sum_{j=1}^{d-1} a_j \right) + \tau,
\end{equation}
where $|\tau|\leq C_1 \max_{k\in [d]} |a_k|^2$.
\end{lemma}

\begin{proof}
Let $b_1,\hdots,b_d$ be arbitrary complex numbers in the unit disc, i.e., $|b_1|,\hdots,|b_d|\leq 1$. Let
$$
F(t) = \frac{1}{\lambda} \Bigg( \frac{1}{G_d(t)} - \frac{1}{G_{d-1}(t)} \Bigg),\quad\mbox{where }G_k(t):=\prod_{j=1}^k (\omega + b_j t) \mbox{ for $k\in [d]$}.
$$
Then,  $G_k(0)=\omega^{k}$ and $G_k'(0)=\omega^{k-1}\sum^k_{j=1}b_j$, so
\begin{equation}\label{oozilo}
F'(0) =-\frac{G_d'(0)}{\lambda\big(G_d(0)\big)^2} +\frac{G_{d-1}'(0)}{\lambda\big(G_{d-1}(0)\big)^2}=-\frac{\sum^{d}_{j=1}b_j}{\lambda \omega^{d+1}}+\frac{\sum^{d-1}_{j=1}b_j}{\lambda \omega^{d}}=\frac{b_d}{z} - \sum_{j=1}^{d-1} b_j,
\end{equation}
where in the last equality we used that $\lambda \omega^{d+1}=1-\omega=-z$. As in \eqref{hiop}, for $k\in[d]$ and all $t\in [0,1]$, we have that
\begin{equation}\label{oohiop}
|G_k'(t)| \leq d  (1+|\omega|)^d.
\end{equation}
Let $\delta_1:= \min\big\{\frac{|\omega|^d}{2 d (1+|\omega|)^{d})},1\big\}>0$. Note that~\eqref{oohiop} implies that for $k\in [d]$ and $t\in (0,\delta_1)$ we have
$$
|G_k(t)| \geq |G_k(0)| - t d (1+|\omega|)^d \geq \frac{|\omega|^d}{2}.
$$
Similarly to~\eqref{oohiop}, for $k\in[d]$ and all $t\in [0,1]$, we have 
\begin{equation}\label{oohiop2}
|G_k''(t)| \leq d^2 (1+|\omega|)^d,
\end{equation}
and hence for $t\in (0, \delta_1)$, following the same argument as in~\eqref{zzzzz1}
$$
|F''(t)| \leq \frac{6}{|\lambda|} (2/|\omega|^d )^3 d^4 (1+|\omega|)^{2d} =: C_1,
$$
which implies
\begin{equation}\label{ooziuq}
|F(t) - F(0) - t F'(0)| \leq C_1 t^2.
\end{equation}
Given $a_1,\dots,a_d$ such that $|a_j|\leq\delta_1$, let $t=\max_j |a_j|$ and $b_j=a_j/t$ (for $j\in [d]$); then, from~\eqref{oozilo}
$$
\frac{1}{\lambda}\Big(\frac{1}{\omega+a_{d}}-1\Big) \prod_{j=1}^{d-1} (\omega + a_k)^{-1} = F(t) \quad\mbox{and}\quad t F'(0) = \frac{a_d}{z} - \sum_{j=1}^{d-1} a_j.
$$
The lemma now follows from~\eqref{ooziuq} and the fact that $F(0)=\omega$.
\end{proof}

\begin{lemma}\label{derile}
Suppose that  $\lambda\in \Complex\backslash\Reals$ and $d\geq 2$. 
Let $\omega$ be the fixpoint of $f(x)=\frac{1}{1+\lambda x^{d}}$ with the smallest norm and let $z=\omega-1$.

There exist reals $C_2:=C_2(\omega,\lambda,d)>0$ and
$\delta_2:=\delta_2(\omega,\lambda,d)>0$ such that for
any $a_1,\dots,a_{d}\in\Complex$ with $|a_j|\leq\delta_2$ (for $j\in [d])$ we have
\begin{equation}
\frac{\partial}{\partial x} \frac{1}{1+\lambda (\omega+x) \prod_{j=1}^{d-1}(\omega + a_j)} \Big |_{x=a_d}  = z + \tau,
\end{equation}
where $|\tau|\leq C_2 \max_{j\in [d]} |a_j|$.
\end{lemma}

\begin{proof}
Let $C_0,\delta_0$ be the constants in Lemma \ref{apx} and let $\delta_2:=\min\{1,\delta_0\}$. 

Suppose that $a_1,\dots,a_{d}\in\Complex$ with $|a_j|\leq\delta_2\leq 1$.  From Lemma~\ref{apx} and the fact that $|z|<1$ (cf. Lemma~\ref{uiopa}), we have 
\begin{equation}\label{csmwdz1}
\Big|\frac{1}{1+\lambda\prod_{j=1}^d (\omega+a_j)}-\omega\Big| \leq \Big| z \sum^d_{j=1} a_j\Big| + C_0 \max_{j\in [d]} |a_j|^2 \leq  C \max_{j\in [d]} |a_j|, \mbox{ where } C:=d+C_0.
\end{equation} 
 From~\eqref{hiop}, we have 
\begin{equation}\label{csmwdz2}
\Big| \lambda\prod_{j=1}^{d-1}(\omega + a_j) - \lambda \omega^{d-1}\Big | \leq  C'\max_{j\in [d-1]} |a_j|,\quad\mbox{where}\quad C':= d|\lambda| (1+|\omega|)^d.
\end{equation}
Note 
$$
\frac{\partial}{\partial x} \frac{1}{1+\lambda (\omega+x) \prod_{j=1}^{d-1}(\omega + a_j)} \Big |_{x=a_d} = - \frac{\lambda\prod_{j=1}^{d-1}(\omega + a_j)}{\Big(1+\lambda\prod_{j=1}^d (\omega+a_j)\Big)^2} 
= -  (\lambda\omega^{d-1} + \tau_1) (\omega + \tau_2)^2,
$$
where $|\tau_1|\leq C' \max_{j\in [d]} |a_j|$ and $|\tau_2|\leq C \max_{j\in [d]} |a_j|$. Notice 
$$
-  (\lambda\omega^{d-1} + \tau_1) (\omega + \tau_2)^2 = -\lambda\omega^{d+1} + \tau,
$$
where $|\tau|\leq |\lambda||\tau_2||\omega|^{d-1} (2|\omega|+ |\tau_2|)+|\tau_1|(|\omega|+|\tau_2|)^2$. Using that $\max_{j\in [d]} |a_j|\leq 1$, we obtain $|\tau|\leq C_2\max_{j\in [d]} |a_j|$, where 
\[C_2:=|\lambda||\omega|^{d-1} C(2|\omega|+ C)+C'(|\omega|+C)>0.\]
This finishes the proof.
\end{proof}

\begin{lemma}\label{lem:perturbcontract}
Let $\Delta\geq 3$ and $\lambda\in \CQ\setminus\Reals$, and set $d:=\Delta-1$. Let $\omega$ be the fixpoint of $f(x)=\frac{1}{1+\lambda x^{d}}$ with the smallest norm. There is a set of activities $\{\lambda_0',\lambda_1',\hdots,\lambda_t'\}\subseteq \CQ\backslash\{0\}$ and a real $r>0$ such that the following hold for all $\hat{\omega}\in B(\omega,r)$. 
\begin{enumerate}
\item for $i=0$, $\frac{1}{1+\lambda_{0}'}\in B(\hat{\omega},2r)$,
\item for $i=1,\hdots,t$, the map $\Phi_i:x\mapsto \frac{1}{1+\lambda_i'x}$ is contracting on  the ball $B(\hat{\omega},2r)$,
\item $B(\hat{\omega},2r)\subseteq \bigcup^{t}_{i=1}\Phi_i(B(\hat{\omega},2r))$.
\end{enumerate}
Moreover, for $i=0,1,\hdots,t$, there is a bipartite graph $G_i'$ of degree $\Delta$ with a vertex $w_i$ such that 
\[\lambda_i'=\frac{\Zin_{G_i',w_i}(\lambda)}{\Zout_{G_i',w_i}(\lambda)}\quad\mbox{and}\quad \mathrm{deg}_{w_i}(G_i')=\begin{cases} \Delta-1 & \mbox{ if $i=0$,}\\ \Delta-2 & \mbox{ if $i=1,\hdots,t$.}\end{cases}\]
\end{lemma}
\begin{proof}
Take $C=\max\{C_0,C_1,C_2,1\}>0$ and $\delta =\min\{\delta_0,\delta_1,\delta_2,|\omega|,1\}\in (0,1]$, where $C_0,C_1,C_2,\delta_0,\delta_1,\delta_2$ are the constants from Lemmas~\ref{apx},~\ref{apxinv} and \ref{derile}. Moreover, let $z=\omega-1$ and recall by Lemma~\ref{uiopa} that $0<|z|<1$. Let $\eps = \frac{1}{9}|z|>0$. By Lemma~\ref{opop2}, $(\Delta,\lambda)$ implements a set of activities $S=\{\lambda_1,\hdots,\lambda_t\}$ which is an $(\eps \rho)$-covering  of $B(\lambda \omega,\rho)$ for some $\rho>0$ satisfying
\begin{equation}\label{smallr}
\rho<\frac{|\lambda|}{10C} \delta|z|(1-|z|). 
\end{equation}
 For convenience, for $i\in[t]$, define $\zeta_i$ by setting
\[\lambda_i=\lambda(\omega+\zeta_i), \mbox{ so that }\max_{i\in [t]} |\zeta_i|\leq \frac{\rho}{|\lambda|}<\frac{1}{5C} \delta|z|(1-|z|)<\delta.\]
By Lemma~\ref{opop2}, $(\Delta,\lambda)$ also implements an activity $\lambda_0$ such that $\lambda_0=\lambda(\omega+\zeta_0)$ with $d|\zeta_0|<\epsilon \rho/|\lambda|$ (in particular, $|\zeta_0|<\delta$). Note that since $\lambda\in \CQ$ we have that $S\subseteq \CQ$. 

Let $r:=\displaystyle\rho\frac{|z|}{3|\lambda|}$. Moreover, let $\lambda_0':=\lambda(\omega+\zeta_0)^{d}$, while for $i=1,\hdots,t$, let 
\[\lambda_i':=\lambda(\omega+\zeta_i)(\omega+\zeta_0)^{d-2}  \mbox{ and $\Phi_i$ be the map }x\mapsto \frac{1}{1+\lambda_i' x}.\] 
Note that $\{\lambda_0',\hdots, \lambda_t'\}\subseteq \CQ\backslash\{0\}$, since $\lambda\in \CQ$, $S\subseteq \CQ$ and for $i=0,1,\hdots,t$ it holds that $|\zeta_i|<\delta\leq |\omega|$.

Consider an arbitrary $\hat{\omega}\in B(\omega,r)$. We first show that
\begin{equation}\label{eq:contra123contra}
\Phi_i\mbox{ is contracting on the ball } B(\hat{\omega},2r) \mbox{ for every $i\in[t]$}.
\end{equation}
Let $x\in B(\hat{\omega},2r)$. Since $|\hat{\omega}-\omega|\leq r$, we have $x\in  B(\omega,3r)$. Note that $3r<\delta$ and $\max_{i\in [t]} |\zeta_i|<\delta$, therefore from Lemma~\ref{derile} (applied to $a_d=x-\omega$, $a_1=\zeta_i$ and $a_2=\cdots=a_{d-1}=\zeta_0$), we have 
\[\Phi'_i(x)=z+\tau, \mbox{ where } |\tau|\leq C \max\{\zeta_0,\zeta_i,|x-\omega|\}\leq C\max\Big\{\frac{\rho}{|\lambda|}, 3r\Big\}= C\frac{\rho}{|\lambda|}\leq \frac{1-|z|}{2},\]
and hence $|\Phi_i'(x)|\leq \frac{1+|z|}{2}$, proving \eqref{eq:contra123contra}. We next show that
\begin{equation}\label{eq:rf34f23342d}
\mbox{for every $x\in B(\hat{\omega},2r)$, there exists $i\in [t]$ such that $\Phi_i^{-1}(x)\in B(\hat{\omega},2r)$}.
\end{equation}
We first calculate $\Phi_i^{-1}(x)$ for $i\in[t]$. By Lemma~\ref{apxinv} (applied to $a_d=x-\omega$, $a_1=\zeta_i$ and $a_2=\cdots=a_{d-1}=\zeta_0$), we have
$$
\Phi_i^{-1}(x) = \frac{1}{\lambda}\Big(\frac{1}{x} - 1\Big)\frac{1}{(\omega+\zeta_i)(\omega+\zeta_0)^{d-2}} = \omega + \Big(\frac{x-\omega}{z} - \zeta_i-(d-2)\zeta_0 \Big) + \tau,
$$
where 
\[|\tau|\leq C \max\{|\zeta_0|^2,|\zeta_i|^2,|x-\omega|^2\}\leq  C\Big(\frac{\rho}{|\lambda|}\Big)^2\leq\eps \frac{\rho}{|\lambda|} \quad\mbox{and}\quad (d-2)|\zeta_0|\leq d|\zeta_0|<\eps \frac{\rho}{|\lambda|}.\]
It follows that
\begin{equation}\label{eq:t5342f254g}
\Phi_i^{-1}(x) - \omega=(\omega+\zeta) - (\omega+\zeta_i)+\tau',\mbox{ where } \zeta:=\frac{x-\omega}{z} \mbox{ and }|\tau'|\leq 2\eps \frac{\rho}{|\lambda|}.
\end{equation}
Note that
\[|\lambda \zeta|\leq |\lambda| \frac{|x-\omega|}{|z|}\leq |\lambda|\frac{3r}{|z|}= \rho,\]
so $\lambda(\omega+\zeta)$ belongs to the ball $B(\lambda \omega,\rho)$. In particular, we can choose $\lambda_i$ from the $(\epsilon \rho)$-covering such that 
\[\mbox{$\big|\lambda(\omega+\zeta)-\lambda_i\big|\leq \epsilon\rho$, which gives that }\big|(\omega+\zeta)-(\omega+\zeta_i)\big|\leq \eps \frac{\rho}{|\lambda|}.\]
Combining this with \eqref{eq:t5342f254g} and using that $\epsilon=\frac{1}{3}|z|$, we obtain that $|\Phi_i^{-1}(x) - \omega|\leq 3\eps\displaystyle \frac{\rho}{|\lambda|}=r$, and therefore by the triangle inequality $|\Phi_i^{-1}(x) - \hat{\omega}|\leq 2r$ (since $|\hat{\omega}-\omega|\leq r$). This proves \eqref{eq:rf34f23342d}. Finally, we show that
\begin{equation}\label{eq:4tg44h575h4}
\frac{1}{1+\lambda_{0}'}\in B(\hat{\omega},2r).
\end{equation}
From Lemma~\ref{apx} (applied to $a_1=\cdots=a_d=\zeta_0$), we obtain that
\[\frac{1}{1+\lambda_{0}'}=\omega+d\zeta_0+\tau, \mbox{ where } |\tau|\leq C|\zeta_0|^2\leq \eps \frac{\rho}{|\lambda|}=r/3.\]
Since $d|\zeta_0|\leq \eps \frac{\rho}{|\lambda|}= r/3$, we obtain that $\frac{1}{1+\lambda_{0}'}\in B(\omega,r)$ and therefore $\frac{1}{1+\lambda_{0}'}\in B(\hat{\omega},2r)$ as well (since $|\hat{\omega}-\omega|\leq r$), finishing the proof of \eqref{eq:4tg44h575h4}.

In light of \eqref{eq:contra123contra}, \eqref{eq:rf34f23342d} and \eqref{eq:4tg44h575h4}, to prove the lemma it remains to show the existence of the graphs $G_0',G_1',\hdots, G_t'$ claimed in the stamement. Since $(\Delta,\lambda)$ implements the activities $\lambda_0,\lambda_1,\hdots,\lambda_t$, for $i=0,1,\hdots,t$ there exists a bipartite graph $G_i$ of maximum degree at most $\Delta$ with terminal $v_i$ such that $\lambda_i=\frac{\Zin_{G_i,v_i}(\lambda)}{\Zout_{G_i,v_i}(\lambda)}$. Consider first the case where we want to implement $\lambda_i'$ for some $i\neq 0$. Construct the graph $G_i'$ by taking $d-2$ disjoint copies of the graph $G_0$, one copy of the graph $G_i$ and identifying their terminals into a single vertex $w_i$. Then, the degree of $w_i$ in $G_i$ is $d-1=\Delta-2$ and we have that
\[\Zin_{G_i',w_i}(\lambda)=\lambda\Big(\frac{\Zin_{G_0,v_0}(\lambda)}{\lambda}\Big)^{d-2}\Big(\frac{\Zin_{G_i,v_i}(\lambda)}{\lambda}\Big), \quad \Zout_{G_i',w_i}(\lambda)=(\Zout_{G_0,v_0}(\lambda))^{d-2}\Zout_{G_i,v_i}(\lambda),\]
and therefore 
\[\frac{\Zin_{G_i',w_i}(\lambda)}{\Zout_{G_i',w_i}(\lambda)}=\lambda (\lambda_0/\lambda)^{d-2}(\lambda_i/\lambda)=\lambda(\omega+\zeta_i)(\omega+\zeta_0)^{d-2}=\lambda_i',\]
as needed. The construction for the case $i=0$ is analogous; to construct $G_0'$, we take $d$ disjoint copies of the graph $G_0$ and identify their terminals into a single vertex $w_0$. Then, $w_0$ has degree $d=\Delta-1$ and it holds that $\lambda_0=\Zin_{G_0',w_0}(\lambda)/\Zout_{G_0',w_0}(\lambda)$.
\end{proof}

Using Lemma~\ref{lem:perturbcontract}, we can now prove  Proposition~\ref{lem:main121} by applying Lemmas~\ref{lem:precision} and~\ref{lem:pathpath}.
\begin{lemmainonetwoone}
\statelemmainonetwoone
\end{lemmainonetwoone}
\begin{proof}
Let $r>0$, $\{\lambda_0',\lambda_1',\hdots,\lambda_t'\}\subseteq \CQ\backslash\{0\}$ and $G_0',\hdots,G_t'$ (with vertices $w_0,\hdots,w_t$, respectively)  be as in Lemma~\ref{lem:perturbcontract}.  Let $\hat{\omega}\in \CQ$ be such that $|\omega-\hat{\omega}|< r$ and let $\rho>0$ be a rational such that $|\omega-\hat{\omega}|<\rho/|\lambda|<r$. Note that the choice of $\hat{\omega}$ and  $\rho$ ensures that 
\begin{equation}\label{eq:x0qwexa13}
B(\omega,\rho/|\lambda|)\subset B(\hat{\omega},2r)
\end{equation}
since for all $x$ with $|x-\omega|\leq \frac{\rho}{|\lambda|}$, we have by the triangle inequality $|x-\hat{\omega}|\leq |x-\omega|+|\omega-\hat{\omega}|< 2r$. Let 
\begin{equation}\label{eq:x0qwexa12}
x_0:=\frac{1}{1+\lambda_{0}'}, \mbox{ so that, by Lemma~\ref{lem:perturbcontract}, } x_0\in B(\hat{\omega},2r).
\end{equation}

Now, suppose that we are given inputs $\lambda'\in B(\lambda \omega,\rho)\cap \CQ$ and rational $\eps>0$.  We want to output in time $poly(\size{\lambda',\epsilon})$ a bipartite graph of maximum degree $\Delta$ that implements $\lambda'$ with accuracy $\epsilon$. By Lemma~\ref{lem:perturbcontract}, we have that the maps $\Phi_i(x)=\frac{1}{1+\lambda_i'x}$ with $i\in[t]$ satisfy the hypotheses of Lemma~\ref{lem:precision} (with $z_0=\hat{\omega}$ and radius $2r$). Moreover, by \eqref{eq:x0qwexa13} and \eqref{eq:x0qwexa12}, we have that $x_0$ and $x^*=\frac{\lambda'}{\lambda}$ belong to the ball $B(\hat{\omega},2r)$. Let $\hat{\epsilon}=\epsilon/|\lambda|$ and $\epsilon'\in (0,\hat{\epsilon})$ be a rational such that $\size{\epsilon'}=poly(\size{\hat{\epsilon}})=poly(\size{\epsilon})$. Using the algorithm of Lemma~\ref{lem:precision} on input $x_0$, $x^*$ and $\epsilon'$, we obtain in time $poly(\size{x_0,x^*,\eps'})=poly(\size{\lambda',\eps})$ a number $\hat{x}\in B(\hat{\omega},2r)$ and a sequence $i_1,\hdots,i_k\in[t]$ such that 
\begin{equation}\label{eq:mg5gm54}
\hat{x}=\Phi_{i_k}(\Phi_{i_{k-1}}(\cdots\Phi_{i_1}(x_0)\cdots))\mbox{ and }\Big|\hat{x}-\frac{\lambda'}{\lambda}\Big|\leq \epsilon'\leq\epsilon/|\lambda|.
\end{equation}
For convenience, let $i_0=0$.

Now,  let $P$ be a path of length $k+1$ whose vertices are labelled as $v_0,v_1,\hdots,v_{k}, v_{k+1}$. Let $\lambdab$ be the activity vector on $P$ given by
\[\lambda_{v_0}=\lambda_{0}'=\frac{1-x_0}{x_0}, \quad \lambda_{v_j}=\lambda_{i_j}' \mbox{ for $j\in[k]$},  \quad \lambda_{v_{k+1}}=\lambda .\]
Then, by Lemma~\ref{lem:pathpath}, it holds that $\Zout_{P,v_{k+1}}(\lambdab)\neq 0$ and $\frac{\Zin_{P,v_{k+1}}(\lambdab)}{\Zout_{P,v_{k+1}}(\lambdab)}=\lambda \hat{x}$; moreover, we can also compute the values $\Zin_{P,v_{k+1}}(\lambdab),\Zout_{P,v_{k+1}}(\lambdab)$ in time polynomial in $k=poly(\size{\lambda',\eps})$ and $\size{x_0,\lambda_1',\hdots,\lambda_t'}=O(1)$. 

Now, let $G'$ be the bipartite graph obtained from the path $P$ by taking for each $j=0,1,\hdots,k$ a disjoint copy of the graph $G_{i_j}'$ and identifying its vertex $w_{i_j}$ with the vertex $v_j$ of the path $P$. Using the degree specifications in Lemma~\ref{lem:perturbcontract}, we have that $G'$ has maximum degree $\Delta$. Moreover, by Lemma~\ref{lem:perturbcontract} we have that $\lambda_i'=\Zin_{G_i',w_i}(\lambda)/\Zout_{G_i',w_i}(\lambda)$ for all $i=0,1,\hdots,t$, so analogously to  Lemma~\ref{lem:transf} we have that
\begin{equation}\label{eq:pbgbprpew3532}
\Zin_{G',v_{k+1}}(\lambda)= C\cdot \Zin_{P,v_{k+1}}(\lambdab),\qquad \Zout_{G',v_{k+1}}(\lambda)=C\cdot \Zout_{P,v_{k+1}}(\lambdab),
\end{equation}
where $C=\prod^t_{i=0}\big(\Zout_{G_i,v_i}(\lambda)\big)^{|\{j\in\{0,\hdots, k\}\,\mid\, \lambda_{v_j}=\lambda_i'\}|}$. We conclude that
\[\frac{\Zin_{G',v_{k+1}}(\lambda)}{\Zout_{G',v_{k+1}}(\lambda)}=\frac{\Zin_{P,v_{k+1}}(\lambdab)}{\Zout_{P,v_{k+1}}(\lambdab)}=\lambda \hat{x}.\]
Combining this with \eqref{eq:mg5gm54}, we obtain that $G'$ with terminal $v_{k+1}$ is a bipartite graph of maximum degree $\Delta$ which implements $\lambda'$ with accuracy $\epsilon$. Moreover, using \eqref{eq:pbgbprpew3532}, we can also compute the values $\Zin_{G',v_{k+1}}(\lambda),\Zout_{G',v_{k+1}}(\lambda)$.
\end{proof}

\bibliographystyle{plain}
\bibliography{\jobname}

\end{document}